\documentclass[11pt]{article}
\usepackage{amssymb,amsfonts}
\usepackage{amsthm}
\usepackage{amsmath}
\usepackage{graphicx}
\usepackage[curve,matrix,arrow,color]{xy}
\usepackage[mathscr]{eucal}

\usepackage[all,dvips,color]{xypic}
\xyoption{line}

\usepackage{epsf}
 \usepackage{epsfig}

 \usepackage{epstopdf}
\usepackage{tikz}

\DeclareSymbolFont{extraup}{U}{zavm}{m}{n}
\DeclareMathSymbol{\vardiamond}{\mathalpha}{extraup}{87}

\makeatletter
\@addtoreset{equation}{section}
\makeatother
\renewcommand{\theequation}{\thesection.\arabic{equation}}
\def\begeqar{\begin{eqnarray}}
\def\endeqar{\end{eqnarray}}
\def\begeq{\begin{equation}}
\def\endeq{\end{equation}}

\def\wgta#1#2#3#4{\hbox{\rlap{\lower.35cm\hbox{$#1$}}
\hskip.2cm\rlap{\raise.25cm\hbox{$#2$}}
\rlap{\vrule width1.3cm height.4pt}
\hskip.55cm\rlap{\lower.6cm\hbox{\vrule width.4pt height1.2cm}}
\hskip.15cm
\rlap{\raise.25cm\hbox{$#3$}}\hskip.25cm\lower.35cm\hbox{$#4$}\hskip.6cm}}

\def\wgtb#1#2#3#4{\hbox{\rlap{\raise.25cm\hbox{$#2$}}
\hskip.2cm\rlap{\lower.35cm\hbox{$#1$}}
\rlap{\vrule width1.3cm height.4pt}
\hskip.55cm\rlap{\lower.6cm\hbox{\vrule width.4pt height1.2cm}}
\hskip.15cm
\rlap{\lower.35cm\hbox{$#4$}}\hskip.25cm\raise.25cm\hbox{$#3$}\hskip.6cm}}

\def\begeqar{\begin{eqnarray}}
\def\endeqar{\end{eqnarray}}

\newcommand{\scal}{scaling~}
\newcommand{\sscal}{Scaling~}

\newcommand{\surj}{\twoheadrightarrow}
\newcommand{\inj}{\hookrightarrow}

\newcommand{\proj}{P}

\newcommand{\one}{\boldsymbol{1}}
\newcommand{\tensor}{\otimes}
\newcommand{\tensore}{\boxtimes}

\newcommand{\twr}{[\one]_{>2}}
\newcommand{\btwr}{\bar{[\one]}_{>2}}
\newcommand{\strt}{T}
\newcommand{\bstrt}{\bar{\strt}}

\newcommand{\q}{\mathfrak{q}}
\newcommand{\ffrac}[2]{\mbox{\footnotesize$\displaystyle\frac{#1}{#2}$}}

\newcommand{\PTL}[1]{TL^a_{#1}}
\newcommand{\TL}[1]{TL_{#1}}
\newcommand{\JTL}[1]{JTL_{#1}}

\newcommand{\ATL}[1]{T^a_{#1}}

\newcommand{\Soc}{\mathrm{Soc}}
\newcommand{\Top}{\mathrm{Top}}
\newcommand{\Mid}{\mathrm{Mid}}
\newcommand{\Char}{\mathrm{char}}

\newcommand{\LQG}{U_{\q} s\ell(2)}
\newcommand{\LQGi}{U_{i} s\ell(2)}

\newcommand{\LQGodd}{U^{\text{odd}}_{\q} s\ell(2)}
\newcommand{\LQGoddi}{U^{\text{odd}}_{i} s\ell(2)}

\newcommand{\K}{\mathsf{K}}
\newcommand{\F}{\mathsf{F}}
\newcommand{\FF}[1]{\F_{#1}}
\newcommand{\f}{\mathsf{f}}

\newcommand{\E}{\mathsf{E}}

\newcommand{\EE}[1]{\E_{#1}}
\newcommand{\h}{\mathsf{h}}
\newcommand{\e}{\mathsf{e}}

\newcommand{\half}{%
  \mathchoice{\ffrac{1}{2}}{\frac{1}{2}}{\frac{1}{2}}{\frac{1}{2}}}

\newcommand{\HXXZ}{H_{\mathrm{XXZ}}(K)}
\newcommand{\PXXZ}{P_{\mathrm{XXZ}}(K)}

\newcommand{\inn}[2]{\langle#1,#2\rangle}

\newcommand{\chid}{\chi^{\dagger}}
\newcommand{\etad}{\eta^{\dagger}}
\newcommand{\bbeta}{\bar{\eta}}
\newcommand{\bbchi}{\bar{\chi}}
\newcommand{\bbetad}{\bar{\eta}^{\dagger}}
\newcommand{\bbchid}{\bar{\chi}^{\dagger}}
\newcommand{\ferm}{\theta}
\newcommand{\fermd}{\theta^{\dagger}}
\newcommand{\sferm}{\psi}
\newcommand{\bsferm}{\bar{\psi}}
\newcommand{\sfermp}{\psi^{2}}
\newcommand{\sfermm}{\psi^{1}}
\newcommand{\bsfermp}{\bar{\psi}^{2}}
\newcommand{\bsfermm}{\bar{\psi}^{1}}
\newcommand{\phip}{\phi^{2}}
\newcommand{\phim}{\phi^{1}}

\newcommand{\Hilb}{\mathcal{H}}
\newcommand{\cHilb}{\overline{\Hilb}}

\newcommand{\veven}[1]{|v^{\text{even}}\rangle}
\newcommand{\vodd}[1]{|v^{\text{odd}}\rangle}

\newcommand{\vac}{\boldsymbol{\Omega}}
\newcommand{\lvac}{\boldsymbol{\omega}}

\newcommand{\oN}{\mathbb{N}}
\newcommand{\oC}{\mathbb{C}}
\newcommand{\oZ}{\mathbb{Z}}

\newcommand{\step}{\epsilon}
\newcommand{\Endo}{\mathrm{End}}

\newcommand{\Hom}{\mathrm{Hom}}

\newcommand{\Vir}{\mathcal{V}}

\newcommand{\VirN}{\boldsymbol{\mathcal{V}}}
\newcommand{\interLie}{\mathfrak{S}}
\newcommand{\interch}{\interLie}
\newcommand{\interchco}{\overline{\interLie}}
\newcommand{\interchalg}{\mathcal{S}}
\newcommand{\interfin}[1]{\interLie_{#1}}
\newcommand{\Clif}[1]{\mathscr{C}_{#1}}
\newcommand{\glinf}{\gl_{\infty}}
\newcommand{\bglinf}{\overline{\gl}_{\infty}}
\newcommand{\glinfc}{\gl'_{\infty}}
\newcommand{\glinfco}{\bglinf'}
\newcommand{\spinf}{\ssp_{\infty}}
\newcommand{\interinfin}{\interLie_{\infty}}
\newcommand{\interinfinco}{\overline{\interLie}_{\infty}}
\newcommand{\interM}[1]{\mathscr{X}_{#1}}
\newcommand{\interP}[1]{\mathscr{Y}_{#1}}
\newcommand{\interMco}[1]{\overline{\mathscr{X}}_{#1}}
\newcommand{\interPco}[1]{\overline{\mathscr{Y}}_{#1}}

\newcommand{\Asp}{\mathcal{A}}
\newcommand{\Bsp}{\mathcal{B}}
\newcommand{\Csp}{\mathcal{C}}

\newcommand{\Egl}{\mathcal{E}}

\newcommand{\SU}{s\ell}
\newcommand{\nSU}{\boldsymbol{s\ell}}
\newcommand{\SUrep}[1]{C_{#1}}
\newcommand{\gl}{\mathfrak{gl}}
\newcommand{\ssp}{\mathfrak{sp}}

\newcommand{\EndcJTL}{\mathrm{End}_{\rule{0pt}{6.5pt}%
{\centJTL}}}

\newcommand{\HomVir}{\mathrm{Hom}_{\rule{0pt}{6.5pt}%
{\Vir(2)}}}

\newcommand{\chVv}{\Hilb_{N}}
\newcommand{\chV}{V}
\newcommand{\rep}{\pi}
\newcommand{\repgl}{\pi_{\gl}}

\newcommand{\repQG}{\rho_{\gl}}
\newcommand{\cent}{\mathfrak{Z}}
\newcommand{\centTL}{\cent_{\mathsf{TL}}}
\newcommand{\centJTL}{\cent_{\JTL{}}}

\newcommand{\N}{\mathrm N}

\newcommand{\bimnch}{\boldsymbol{\mathcal{H}}}
\newcommand{\bimn}{\bimnch^{\pm}}
\newcommand{\bimnl}{\bimnch^{\pm,(l)}}
\newcommand{\bimnr}{\bimnch^{\pm,(r)}}
\newcommand{\bimnbos}{\bimnch^{+}}
\newcommand{\bimnbosl}{\bimnch^{+,(l)}}
\newcommand{\bimnbosr}{\bimnch^{+,(r)}}
\newcommand{\bimnferl}{\bimnch^{-,(l)}}
\newcommand{\bimnferr}{\bimnch^{-,(r)}}
\newcommand{\bimnfer}{\bimnch^{-}}

 \newcommand{\nfer}{\boldsymbol{\psi}}
 \newcommand{\nferp}{\nfer^2}
 \newcommand{\nferm}{\nfer^1}

\newcommand{\first}{first }
\newcommand{\second}{second }

\newcommand{\energy}{energy }
\newcommand{\bz}{\bar{z}}
\newcommand{\bw}{\bar{w}}
\newcommand{\der}{\partial}
\newcommand{\bder}{\bar{\partial}}
\newcommand{\enrg}{S}
\newcommand{\virm}[1]{L^{(#1)}}
\newcommand{\bvirm}[1]{\bar{L}^{(#1)}}

\newcommand{\enrgpl}{\enrg_{\text{pl.}}}
\newcommand{\enrgcyl}{\enrg_{\text{cyl.}}}
\newcommand{\LL}{L'}

\newcommand{\Xodd}[1]{\mathsf{X}_{#1}}
\newcommand{\XX}{\mathsf{X}}
\newcommand{\PP}{\mathsf{P}}

\newcommand{\PPodd}[1]{\mathsf{T}_{#1}}
\newcommand{\TLX}{d^0}

\newcommand{\WP}{\mathcal{H}}

\newcommand{\VX}{\mathcal{X}}
\newcommand{\VP}{\mathcal{P}}
\newcommand{\IrrTL}[1]{(d^0_{#1})}
\newcommand{\PrTL}[1]{\mathscr{P}_{#1}}
\newcommand{\StTL}[1]{\mathscr{W}_{#1}}
\newcommand{\AIrrTL}[2]{\mathscr{L}_{#1,#2}}
\newcommand{\APrTL}[1]{\widehat{\mathscr{P}}_{#1}}

\newcommand{\AStTL}[2]{\mathscr{W}_{#1,#2}}

\newcommand{\filt}{\mathcal{F}}
\newcommand{\Firr}[2]{F_{#1,#2}^{(0)}}

\newcommand{\vand}{\mathscr{V}}
\newcommand{\vandT}{\mathscr{V}^{T}}
\newcommand{\Lv}{\mathsf{L}}
\newcommand{\Tv}{\mathsf{T}}

\newcommand{\Appfourier}{A~}
\newcommand{\Appglinf}{B}
\newcommand{\AppLim}{C}
\newcommand{\AppHln}{D}
\newcommand{\AppChar}{E}

\newcommand{\nord}[1]{\boldsymbol{\colon}\!\!{#1}\boldsymbol{\colon}\!\!}

\newtheorem{Thm}[subsection]{Theorem}
\newtheorem{thm}[subsubsection]{Theorem}
\newtheorem{Lemma}[subsection]{Lemma}
\newtheorem{lemma}[subsubsection]{Lemma}
\newtheorem{Prop}[subsection]{Proposition}
\newtheorem{prop}[subsubsection]{Proposition}

\newtheorem{conj}[subsubsection]{Conjecture}
\newtheorem{Cor}[subsection]{Corollary}
\newtheorem{cor}[subsubsection]{Corollary}

\theoremstyle{definition}

\newtheorem{dfn}[subsubsection]{Definition}
\newtheorem{rem}[subsubsection]{Remark}

\begin{document}
\begin{center}

\Large{Associative algebraic approach to  logarithmic CFT in the bulk:\\ the continuum limit of the $\gl(1|1)$ periodic spin chain, \\Howe duality and  the interchiral algebra.}
\vskip 1cm

{\large A.M. Gainutdinov\,$^{a}$, N. Read\,$^{b}$, and
H. Saleur\,$^{a,c}$}

\vspace{1.0cm}

{\sl\small $^a$  Institut de Physique Th\'eorique, CEA Saclay,\\
Gif Sur Yvette, 91191, France\\}
{\sl\small $^b$ Department of Physics, Yale University, P.O. Box 208120,\\ New Haven, Connecticut 06520-8120, USA\\}
{\sl\small $^c$ Department of Physics and Astronomy,
University of Southern California,\\
Los Angeles, CA 90089, USA\\}

\end{center}

\begin{abstract}
We develop in this paper the principles of an  associative algebraic approach to \textit{bulk} logarithmic conformal field theories (LCFTs). We concentrate on the closed $\gl(1|1)$ spin-chain and its continuum limit --  the $c=-2$ symplectic fermions theory -- and rely on two technical companion papers, \textit{Continuum limit and symmetries of the periodic $\gl(1|1)$ spin chain} [Nucl.\;Phys.\;B 871 (2013) 245-288] and \textit{Bimodule structure in the periodic $\gl(1|1)$ spin chain} [Nucl.\;Phys.\;B 871 (2013) 289-329].

Our main result  is that the algebra of local Hamiltonians, the Jones--Temperley--Lieb algebra $\JTL{N}$, goes over in the continuum limit to a bigger algebra than $\VirN$, the product of the left and right Virasoro algebras. This algebra, $\interchalg$ -- which we call \textit{interchiral}, mixes the left and right moving sectors, and is generated, in the symplectic fermions case, by the additional field $S(z,\bar{z})\equiv S_{\alpha\beta}\psi^\alpha(z)\bar{\psi}^\beta(\bar{z})$,  with a symmetric form $S_{\alpha\beta}$ and conformal weights $(1,1)$. We discuss in details how  the space of states  of the LCFT (technically, a Krein space) decomposes onto representations of this algebra, and how this decomposition is related with properties of the finite spin-chain. We show that there is a complete  correspondence between algebraic properties of finite periodic spin chains and the continuum limit.

An important technical aspect of our analysis involves the fundamental new observation that the action of $\JTL{N}$  in the $\gl(1|1)$ spin chain is in fact isomorphic to  an enveloping algebra of a certain Lie algebra, itself a non semi-simple version of $\ssp_{N-2}$. The semi-simple part of $\JTL{N}$  is represented by $U \ssp_{N-2}$, providing a beautiful example of a classical Howe duality, for which we have a non semi-simple version in the full $\JTL{N}$ image represented in the spin-chain. On the continuum side,  simple modules over $\interchalg$ are identified with ``fundamental'' representations of~$\spinf$.

\end{abstract}


\section{Introduction}

Our understanding of logarithmic conformal field theory (LCFT) has
greatly improved recently thanks to a renewed focus on purely algebraic features and  in particular a systematic study of Virasoro
indecomposable modules.

Most of the developments have concerned   boundary LCFTs, where two lines of attack have been pursued. The
first one is rather abstract, and follows the pioneering work of
Rohsiepe~\cite{Rohsiepe} and Gaberdiel~\cite{Gab1}, where, in
particular, fusion is interpreted as the tensor product of the
symmetry algebra~\cite{GK1,MatRid,EbFl}. The other~\cite{Pearceetal,
ReadSaleur07-2} is somehow more concrete, in that it relies on lattice
regularizations, and exploits the (not entirely understood) similarity~\cite{PasquierSaleur} between the properties of lattice models and their conformally invariant continuum limits.
  The models involved in this second approach provide~\cite{Martinbook} representations of
associative algebras such as the Temperley--Lieb algebra. It turns out
that the representation theory of these algebras is, in a certain (categorical, see Sec.~\ref{categorical})
sense~\cite{ReadSaleur07-1}, similar to the one of chiral algebras in
LCFT. The structure of indecomposable modules and fusion rules can
then be predicted from the analysis of the scaling limit of the lattice
models~\cite{Pearceetal, ReadSaleur07-2,RP}. The results are in perfect agreement
(so far) with the conclusions of the first approach based on (considerably more involved) calculations in the Virasoro algebra~\cite{MatRid,KitRid}.

To be now a little more precise, the general philosophy of the lattice approach in the boundary case relies on the  analysis of  the microscopic  model as a bi-module
over two algebras. In physical terms, one of these algebras is
generated by the local hamiltonian density, and the other is the
`symmetry' commuting with the hamiltonian densities. More precisely,
the types of models we are interested in carry alternating fundamental
representations of a super Lie algebra, and admit a single invariant
nearest neighbor coupling, the `super equivalent' of the Heisenberg
interaction $H_i\approx \vec{S}_i\cdot\vec{S}_{i+1}$. The different
$H_i$'s turn out to obey some non-trivial algebraic relations, so that
the models provide a representation of a certain abstract algebra $A$
-- the ordinary and boundary Temperley--Lieb algebras for instance.

Meanwhile, the local Hamiltonians are invariant under the super Lie
algebra. In fact, they typically are invariant under a bigger algebra
which is, in technical terms, the centralizer of $A$. We recall here
that the centralizer $\cent_A$ of an associative algebra $A$ acting on a
representation space $\chV$ is the algebra of all intertwining
operators $\Endo_{A}(\chV)$, {\it i.e.}, $\cent_A$  is the algebra of the maximum dimension
such that $[\cent_{A},A]=0$. For the open case of finite $\gl(m|n)$-symmetric
spin chains, this centralizer was dubbed ${{\cal A}_{m|n}}$ in
\cite{ReadSaleur07-1}, where it was shown to be Morita equivalent to
$\LQG$ with $\q+\q^{-1}=m-n$, actually to the corresponding finite-dimensional $\q$-Schur algebra. Usually, the representation theory of
this centralizer -- typically a quantum group or an algebra Morita
equivalent to one -- is easier to study than the representation theory
of $A$ or its \scal limit, a chiral algebra. In general, all these
algebras are not semi-simple, and give rise to complicated
indecomposable modules. The point is, in part, that these modules have, even for finite spin chains, properties which are closely related with those
of the associated LCFT.

While the key observation in \cite{ReadSaleur07-1,Pearceetal} about the similarity between the algebraic properties of the lattice models and their
logarithmic continuum limit is not entirely understood, it can be
then  interpreted in some cases at least in terms of quantum groups. Recent studies of  centralizers of chiral algebras
(Virasoro and W-algebras)  in continuum  logarithmic models
have  indeed unraveled a remarkable equivalence~\cite{[FGST2],[FGST4],BFGT} between
the representation theory  and the fusion rules of the chiral algebras and
certain quantum groups. It is exactly the same quantum groups  that appear as centralizers in the lattice models, hence providing the link between finite size and scaling limit
properties.

\bigskip

Turning now to bulk LCFTs, progress has been more modest. The problem
is that, on the more abstract side, one now expects indecomposability
under the left and right actions of the Virasoro algebras, leading to
potentially very complicated modules which have proven too hard to
study so far, except in some special cases. These include bulk
logarithmic theories~\cite{GabRun,GabRunW} with W-algebra
symmetries~\cite{[K-first],[FGST3]}, and WZW models on
supergroups which, albeit very simple as far as LCFTs go, provide
interesting lessons on the coupling of left and right
sectors~\cite{SaleurSchomerus}. On the more concrete side, while it is
possible to define and study lattice models whose continuum limit is a
(bulk) LCFT, the underlying structures are also very difficult to get:
symmetries are smaller, and the lattice algebras have much more
complicated representation theory. Nevertheless, it looks possible
to generalize the approach in~\cite{ReadSaleur07-1,ReadSaleur07-2}
thanks to recent results about  the centralizer $\centJTL$ of the Jones--Temperley--Lieb (JTL)
 algebra $\JTL{N}$ acting on periodic super-symmetric spin chains  with $N$ tensorands/sites. This
paper present our first results in this direction for the $\gl(1|1)$
case. It is based on two technical companion papers.  In the first one~\cite{GRS1},
we focussed on the symmetries of the spin chain, that is, the
centralizer, in the alternating product of the $\gl(1|1)$ fundamental
representation and its dual, of the JTL algebra. We proved that this
centralizer is only a subalgebra of $\LQG$ at $\q=i$ that we dubbed
$\LQGodd$. We then analyzed the continuum limit $N\to\infty$ of the JTL algebra:
using general arguments about the regularization of the stress-energy tensor, we identified families of JTL elements going over to
the Virasoro generators $L_n$ and $\bar{L}_n$. We also discussed the well
known $\SU(2)$ symmetry of symplectic fermions from the lattice point
of view, and showed that this symmetry, albeit present in the
continuum limit, does not have a simple, useful analog on the lattice.
In our second paper~\cite{GRS2}, we analyzed the decomposition of the spin chain over
the JTL algebra $\JTL{N}$, and obtained the full decomposition
as a bimodule over $\LQGodd$ and $\JTL{N}$.

Equipped with these results, we can now explore to what extent the remarkable properties observed in the open case \cite{ReadSaleur07-1} carry over to the bulk case. Our conclusion is that the algebraic properties of the finite periodic spin chain and the bulk LCFT are again very similar. The crucial new ingredient is that the JTL algebra goes over, in the continuum limit,
to a bigger operator algebra than $\VirN$, the product of the left and right Virasoro algebras. This algebra -- which we call \textit{interchiral}, mixes the left and right moving sectors, and is generated, in the symplectic fermions case, by the additional field
\begin{equation}\label{eq:interch-f}
S(z,\bar{z})\equiv S_{\alpha\beta}\psi^\alpha(z)\bar{\psi}^\beta(\bar{z}), \qquad S_{12} = S_{21}=1, \quad S_{11} = S_{22}=0,
\end{equation}
 with conformal weights $(1,1)$. More formally, we construct an inductive system of $\JTL{N}$ algebras,
together with their spin-chain representations, and identify the inductive limit with an infinite-dimensional
operator algebra $\interchalg$ generated by the modes of the field $S(z,\bar{z})$. Most of the present work is devoted to explicitly identifying this interchiral algebra $\interchalg$, studying its properties, and  using it to provide a new analysis of   the symplectic fermions LCFT. We believe that the concept of interchiral algebra will prove fundamental in the analysis of more complicated cases -- in particular those at central charge $c=0$ -- as will be discussed in forthcoming work.

\medskip

The paper is organized as follows. We start in section \ref{sec:remind} by a reminder of the main algebraic features, both on the lattice and in the scaling limit, for the open $\gl(1|1)$ spin chain (and its $\gl(n|n)$ generalizations). In section \ref{sec:remindclosed} we similarly remind the reader of the main algebraic properties of the closed $\gl(1|1)$ spin chain. The latter involves now a `periodic' version of the Temperley--Lieb algebra which we call, following~\cite{ReadSaleur07-1}, the Jones--Temperley--Lieb algebra $\JTL{N}$, and the symmetry algebra $\LQGodd$ which is, up to some trivial elements, its centralizer $\centJTL$. The structure of the corresponding bimodule is recalled in Fig.~\ref{JTL-Uqodd-bimod}. We also mention briefly the antiperiodic spin chain (where the $\gl(1|1)$ symmetry is broken), where one now deals with a twisted version $\JTL{N}^{tw}$ of the Jones--Temperley--Lieb algebra, while the centralizer is just $U s\ell(2)$ algebra. Section~\ref{Howe} is the first containing new results. We discuss there how, remarkably, the image of the $\JTL{N}$ in the $\gl(1|1)$ spin chains is in fact isomorphic to (the image of) the enveloping algebra of a certain non-semisimple Lie algebra containing $\ssp_{N-2}$ as its maximal semi-simple subalgebra. The semi-simple part of $\JTL{N}$ (that is, after quotienting by the radical) is represented by $U \ssp_{N-2}$, providing a beautiful example of a classical Howe duality~\cite{H3}, for which we have a non semi-simple version in the full $\JTL{N}$ case. The main goal of the rest of this paper is to understand
 the scaling limit of $\JTL{N}$ defined as a particular inductive limit $N\to\infty$ of algebras. To do this, we start in section~\ref{sec:scal-lim} by discussing the scaling limit of the $\gl(1|1)$ spin chain, in particular the inductive limit of $\JTL{N}$ as algebras of \textit{bilinears in fermion modes} -- subsections~\ref{sec:scaling},   \ref{commrels} and, of  the centralizer and the bimodule structure in~\ref{sec:bimod-scal-lim}. In subsection~\ref{sec:interinfin-sm}, we describe the Virasoro algebra content of the inductive limits of simple $\JTL{N}$-modules that appear in the periodic $\gl(1|1)$ model.  While on the finite lattice with $N$ sites the simples are just fundamental representations of $\ssp_{N-2}$, in the limit simple modules over the scaling limit of $\JTL{N}$ are identified with appropriate simple $\spinf$-modules. In~\ref{sec:scal-lim-tw}, we also describe the scaling limit of the anti-periodic $\gl(1|1)$ spin chains.

We then turn to section  \ref{sec:interchiral} where we identify the scaling limit of the $\JTL{N}$ algebras as  the {\sl interchiral algebra}. This identification requires making contact with several physics concepts and introduction of completions of the inductive-limit algebras in section~\ref{sec:interch-interinfin}.
While the sections~\ref{Howe} and~\ref{sec:scal-lim} are mathematically the most important of this paper -- they contain our main theorems equipped with proofs --
the section  \ref{sec:interchiral} is conceptually the most important from physics point of view and relies on results provided with only ideas of a proof or conjectures which we are not proving in the present paper.
Accordingly, section  \ref{sec:interchiral} is less rigorous, and  can  be considered   as the ``physics part'' of this paper.

The idea of the interchiral algebra is that, for general models,  the scaling limit of $\JTL{N}$ will contain, in addition to the chiral and anti-chiral Virasoro algebras, the modes of non chiral fields such as the degenerate conformal field $\Phi_{2,1}\times\bar{\Phi}_{2,1}$. In the particular case of symplectic fermions, this field becomes  the \textit{interchiral field} $\enrg(z,\bar{z})$ from~\eqref{eq:interch-f}, and the scaling limit of $\JTL{N}$, which we discuss in section  \ref{sec:scal-lim}  from the point of view of  \textit{ bilinears in fermion modes}, can also be  identified using  bilinears in fermionic fields, which are discussed in subsection \ref{sec:interch-modes}. The interchiral algebra proper is introduced in subsection \ref{sec:interdef} where it is denoted by $\interchalg$. Subsection \ref{sec:interchalg-tw} extends the discussion to the antiperiodic model or the so-called ``twisted sector''.  In section \ref{sec:simpl-JTL-Vir} we discuss modules over the interchiral algebra in the LCFT and their relation with  $\JTL{N}$ modules on the lattice. Of particular importance is our discussion of the vacuum module over the interchiral algebra in subsections~\ref{sec:interch-simple} and~\ref{sec:interch-simple-2}. In subsection \ref{sec:indecompmod} we discuss indecomposable modules over $\interchalg$ and their relation with the  symplectic  fermion theory. This involves an analysis of the  space  of states of this theory as a module over the left-right Virasoro algebra $\VirN(2)$ (which we believe is new) in subsubsection \ref{sec:non-chiral-sympl-fer}.  We then show that this analysis agrees with the scaling limit of the bimodule over $\JTL{N}$ and $\centJTL$. A few conclusions -- and pointers to subsequent developments -- are given in the conclusion. Finally, several technical aspects are addressed in five appendices. In particular, our inductive limits constructions are in App.~\AppLim.

\subsection{Notations}
To help the reader navigate through this long paper, we provide a partial list of notations (consistent with all other papers in the series)

\begin{itemize}

\item[\mbox{}]$\TL{N}$ --- the (ordinary) Temperley--Lieb algebra,

\item[\mbox{}]$\ATL{N}$ --- the periodic Temperley--Lieb algebra with the translation $u$, or the algebra of affine \\\mbox{}\qquad\quad diagrams,

\item[\mbox{}]$\JTL{N}$ --- the Jones--Temperley--Lieb algebra,

\item[\mbox{}]$\centJTL$ --- the  centralizer of $\JTL{N}$,

\item[\mbox{}]$\repgl$ --- the spin-chain representation of $\JTL{N}$,

\item[\mbox{}]$\repQG$ --- the spin-chain representation of the quantum group $\LQG$,

\item[\mbox{}]$\e$, $\f$ --- the Lusztig's divided powers in $\LQG$,

\item[\mbox{}]$\Egl_{i,j}$ --- the elementary matrix with exactly one
nonzero entry, which is $1$ in the $(i,j)$ position,

 \item[\mbox{}]$\XX_{1,n}$ ---  the simple $\LQGi$-modules,

 \item[\mbox{}]$\PP_{1,n}$ ---  the projective $\LQGi$-modules,

 \item[\mbox{}]$\Xodd{n}$ --- the simple $\LQGoddi$- and $\centJTL$-modules,

 \item[\mbox{}]$\PPodd{n}$ ---  the indecomposable summands in spin-chain decomposition over the centralizer $\centJTL$,

\item[\mbox{}]$\StTL{j}$ --- the standard modules over $\TL{N}$,

\item[\mbox{}]$\PrTL{j}$ --- the projective modules over $\TL{N}$,

\item[\mbox{}]$\AIrrTL{j}{(-1)^{j+1}}$ --- the simple modules over $\JTL{N}$
  for which we also use the notation $(\TLX_{j})$,

\item[\mbox{}]$\AStTL{j}{(-1)^{j+1}}$ --- the standard modules over $\JTL{N}$,

\item[\mbox{}]$\APrTL{j}$ --- the indecomposable summands in spin-chain decomposition over $\JTL{N}$,

\item[\mbox{}]$\interfin{N}$ --- a Lie algebra introduced in Sec.~\ref{S_N-def},

\item[\mbox{}]$\interfin{N}'=\interfin{N}\oplus\oC\one$ --- a central extension of $\interfin{N}$,

\item[\mbox{}]$\glinf$ ---  the Lie algebra  of infinite matrices with finite number of non-zero elements, see Sec.~\ref{sec:spinf},

\item[\mbox{}]$\glinfc$ --- a central extension of the Lie algebra $\glinf$,

\item[\mbox{}]$\spinf$ --- the symplectic Lie algebra of infinite matrices,

\item[\mbox{}]$\interinfin$ --- a Lie algebra of infinite matrices -- the scaling limit of $\interfin{N}'$ -- introduced in Sec.~\ref{sec:interinfin},

\item[\mbox{}]$\interch$ --- a Lie algebra of local operators introduced in Sec.~\ref{sec:comm-rel-interch},

\item[\mbox{}]$\interchalg$ --- the interchiral algebra, see Sec.~\ref{sec:interdef},

\item[\mbox{}]$\Vir(2)$ --- the left Virasoro algebra with $c=-2$,

\item[\mbox{}]$\VirN(2)$ --- the product of the left and right Virasoro algebras,

\item[\mbox{}]$\VX_{j,1}$ --- the simple Virasoro modules,

\item[\mbox{}]$\VP_{n,1}$ --- the staggered Virasoro modules,

\item[\mbox{}] $\interM{j}$ --- a module over $U\interinfin$ which is obtained in the scaling limit of the JTL  modules $\AIrrTL{j}{(-1)^{j+1}}$,

\item[\mbox{}]  $\interP{j}$ --- a module over $U\interinfin$ which is obtained in the scaling limit of the JTL  modules $\APrTL{j}$,

\item[\mbox{}]$\bimnbos$ --- the bosonic  space of states in symplectic fermion theory,

\item[\mbox{}]$\bimnfer$ --- the fermionic space of states in symplectic fermion theory.

\end{itemize}

All the algebras in this paper are associative and defined over the field $\oC$ of complex numbers.

\section{A reminder of the open case}\label{sec:remind}

 \subsection{The open $\gl(1|1)$ super-spin chain}\label{sec:super-spin-ch-def}

The open  $\gl(1|1)$ super-spin chain~\cite{ReadSaleur07-1,GRS1} is a
 tensor product
 representation $\chVv=\tensor_{j=1}^{N}\oC^2$ of the Temperley--Lieb
 (TL) algebra of zero fugacity parameter $m$. We recall that the
 ordinary TL algebra denoted by $\TL{N}(m)$ is generated by $e_j$,
 with $1\leq j\leq N-1$, and has the defining relations
\begin{eqnarray}\label{TL-rel1}
e_j^2&=&me_j,\nonumber\\
e_je_{j\pm 1}e_j&=&e_j,\label{TL}\\
e_je_k&=&e_ke_j\qquad(j\neq k,~k\pm 1),\nonumber
\end{eqnarray}
where  $m$ is a real parameter.

 The representation space $\chVv$ consists of
 $N=2L$ tensorands
labelled $j=1,\ldots,2L$ with the fundamental representation of $\gl(1|1)$ on even sites and its dual on odd sites.
The representation of each TL generator $e_j$ is given by projectors on the $\gl(1|1)$-invariant in the
product of two neighbour tensorands
\begin{equation}\label{rep-TL-1}
e_j^{\gl}= (f_j+f_{j+1})(f_j^\dagger+f_{j+1}^\dagger),\qquad 1\leq
j\leq 2L-1,
\end{equation}
where we use a free fermion representation based on operators $f_j$
and $f_j^\dagger$ acting non-trivially only on $j$th tensorand and obeying
\begin{eqnarray}\label{f_j-rel}
\{f_j,f_{j'}\}=0,\quad\{f^{\dagger}_j,f^{\dagger}_{j'}\}=0,\quad\{f_j,f_{j'}^\dagger\}=(-1)^{j}\delta_{jj'},
\end{eqnarray}
where the minus sign  for an odd $j$ is due to presence of the dual representations of $\gl(1|1)$.
The generators $e_j^{\gl}$  provide then a
representation of $\TL{2L}(m=0)$ which
is known to be faithful.

 The representation space $\Hilb_{2L}$ is equipped with an inner
 product $\inn{\cdot}{\cdot}$ such that $\inn{f_j
 x}{y}=\inn{x}{f_j^{\dagger}y}$ for any $x,y\in\Hilb_{2L}$.  We stress
 that the inner product is indefinite because of the sign factors in
 the relations~\eqref{f_j-rel}. The Hamiltonian operator
 \begin{equation*}
 H_{\mathrm{op.}}=-\sum_{j=1}^{2L-1}e^{\gl}_j,
 \end{equation*}
with the `hamiltonian densities' $e^{\gl}_j$ defined in~\eqref{rep-TL-1}, is self-adjoint $H_{\mathrm{op.}}=H_{\mathrm{op.}}^{\dagger}$ with respect to this inner product (actually, each $e_j^{\gl}$ is a self-adjoint operator).
Its eigenvalues are real and the eigenvectors can easily be  computed. Because of the indefinite inner product, the self-adjoint Hamiltonian can have non-trivial Jordan cells and here it is indeed the case -- the Jordan cells are of rank two~\cite{ReadSaleur07-2}.

\medskip
The open $\gl(1|1)$ spin-chain exhibits a large symmetry algebra
dubbed ${\cal A}_{1|1}$ in~\cite{ReadSaleur07-1}. This algebra is the
centralizer $\centTL$ of $\TL{N}(0)$ and is
generated by the identity and the five generators
\begin{gather}
F_{(1)}=\sum_{1\leq j\leq N} f_j,\qquad
F^\dagger_{(1)}=\sum_{1\leq j\leq N} f_j^\dagger,\label{F1-def}\\
F_{(2)}=\sum_{1\leq j<j'\leq N}f_jf_{j'},\qquad
F_{(2)}^\dagger=\sum_{1\leq j<j'\leq N} f_{j'}^\dagger f_j^\dagger,\qquad
{\N}=\sum_{1\leq j\leq N} (-1)^jf_j^\dagger f_j.\label{F2-def}
\end{gather}
We note that the formulas give, after trivial redefinitions, a
representation of the (Lusztig) quantum group $\LQG$ at $\q=i$. The fermionic
generators above, those with the subscript `$(1)$', are from the nilpotent part and the
bosonic ones form the $s\ell(2)$ subalgebra in $\LQG$. It will be
convenient in what follows to introduce slightly modified generators\footnote{We
dispense with the more correct notation $\repQG$ used in~\cite{GRS1}.} of $\LQG$
\begin{align}
\h&=\half\sum_{j=1}^{N} (-1)^jf_j^\dagger f_j-\ffrac{L}{2},\;
&\e &= \q^{-1}\sum_{1\leq j_1<j_2\leq N}f_{j_1}^\dagger f_{j_2}^\dagger,\;
&\f &= \q\sum_{1\leq j_1<j_2\leq N}f_{j_1}f_{j_2},\label{QG-ferm-2}\\
\K&=(-1)^{2\h},\;
&\E& = \sum_{j=1}^N f_j^\dagger\,\K,\;
&\F &= \q^{-1}\sum_{j=1}^N f_j\label{QG-ferm-1}
\end{align}
obeying in particular
\begin{align}
 \K\E\K^{-1}&=-\E,&
  \K\F\K^{-1}&=-\F,&
  \E^{2}=\F^{2}&=0,&\\
[\h,\e]&=\e,& [\h,\f]&=-\f,& [\e,\f]&=2\h,&
\end{align}
see more details in~\cite{GRS1}.
The $s\ell(2)$ Cartan generator $\h$ is related to the total-spin operator $S^z$  in the XX language, or the Fermion number $\N$ as
\begin{equation}
2\h=S^z=\N-L.
\end{equation}
For our alternating chain, the values of $S^z$ are integer since the
number of sites is even, and thus $\h$ takes integer or half integer
values.

\subsection{Bimodule on a finite lattice}

The decomposition of the open spin-chain as a bimodule over the pair
$(\TL{N},{\cal A}_{1|1})$ of mutual centralizers
is shown on Fig.~\ref{openbimodule-fin} for  the $N=8$ case and borrowed
from~\cite{ReadSaleur07-2}.
 \begin{figure}\centering
 \leavevmode
 \epsfysize=80mm{\epsffile{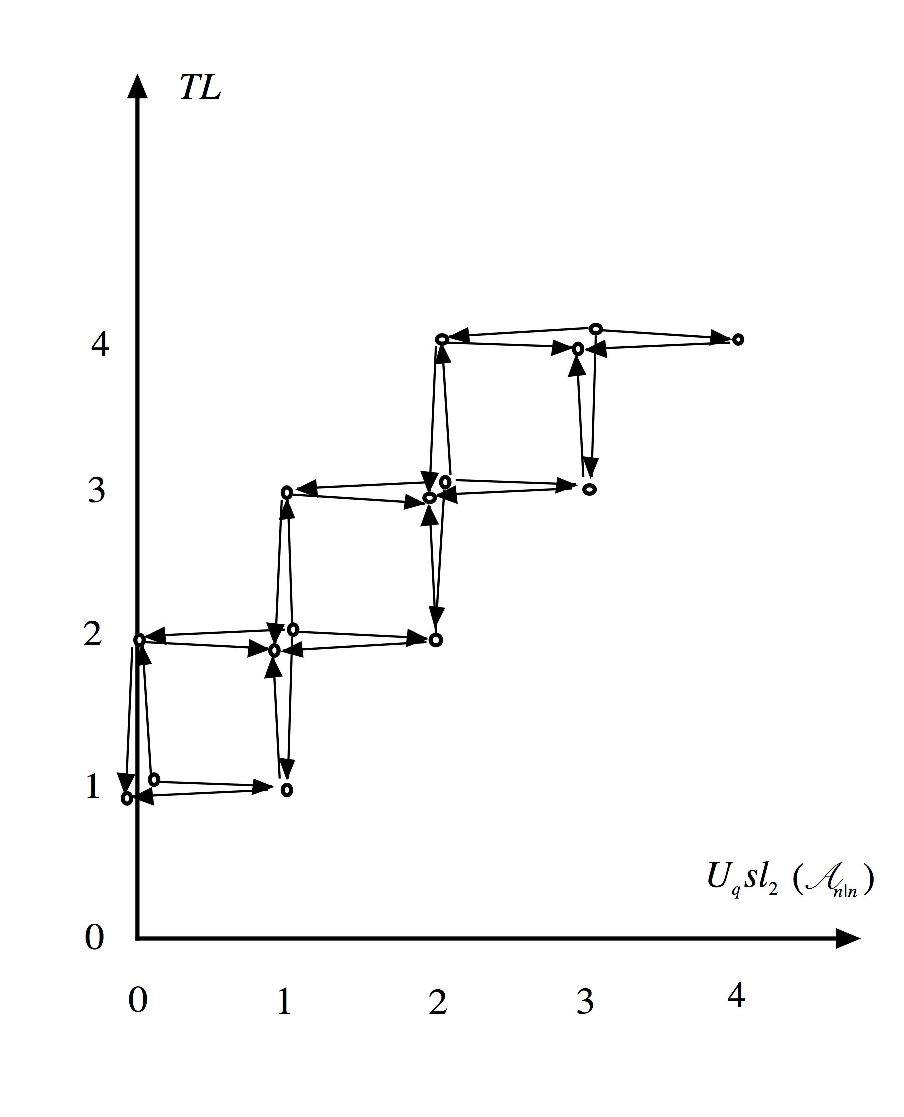}}
  \caption{The structure of the open $\gl(1|1)$ spin-chain (i.e., $n=1$) for $N=8$
  sites, as a representation of the product $\TL{N}\boxtimes\LQGi$. Some nodes
  occur twice and have been separated slightly for
  clarity.}\label{openbimodule-fin}
 \end{figure}
The label $j$ in the horizontal direction  corresponds  to
the double value of the $s\ell(2)$ spin $n$ -- the highest weight of
an $s\ell(2)$-module.  Each node with a Cartesian coordinate $(j,k)$ in the bimodule
diagram corresponds to a simple subquotient over the tensor product
$\TL{N}\boxtimes\LQGi$ of the commuting algebras and arrows show the action of both the algebras --
the Temperley--Lieb $\TL{N}$ acts in the vertical direction
(preserving the coordinate $j$), while $\LQGi$ acts in the horizontal
way.  Indecomposable projective $\TL{N}$-modules $\PrTL{k}$ can be
recovered by ignoring all the horizontal arrows, while tilting
$\LQGi$-modules $\PP_{1,j}$ (these are also projective~\cite{BFGT}) are obtained
by ignoring all the vertical arrows. For more details, see~\cite{ReadSaleur07-2,GRS2}.

\subsection{Bimodule in the continuum limit}\label{subsec:bimodopen}

The crucial observation of \cite{ReadSaleur07-2} is that an identical
bimodule structure, extending to arbitrarily high values of the $s\ell(2)$ spin $n=j/2$, is
present in the continuum limit. This is illustrated on
Fig.~\ref{openbimodule-cont}, where the same comments as in the finite
chain apply exactly, with the replacement of $\TL{N}$ by the Virasoro
algebra at central charge $c=-2$ denoted by $\Vir(2)$.

  \begin{figure}\centering
 \leavevmode
 \epsfysize=80mm{\epsffile{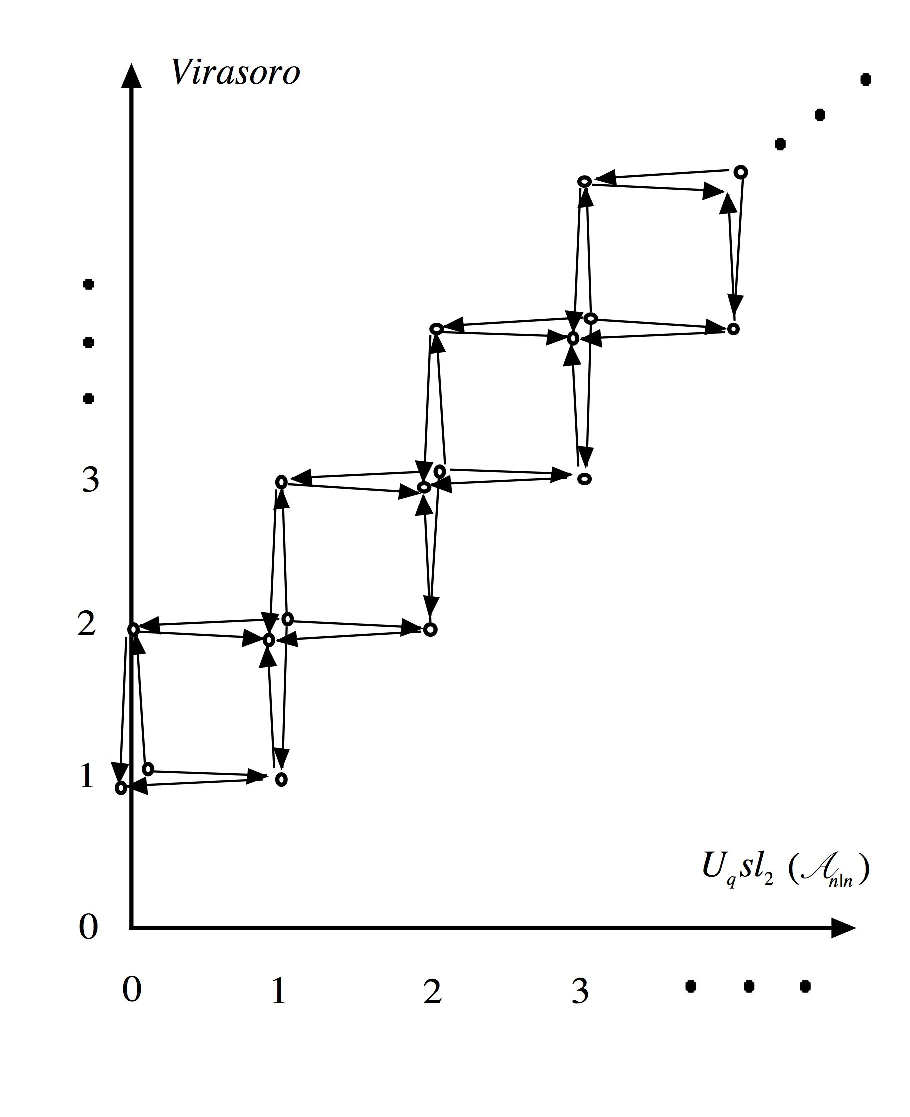}}
  \caption{Space of states of the continuum theory  as a representation of $\Vir(2)\boxtimes\LQGi$. The vertical labels $n=1,2,\dots$  are for the Virasoro $\Vir(2)$ irreducible representations of the conformal weights $\Delta_{n,1}=n(n-1)/2$.}\label{openbimodule-cont}
 \end{figure}

It is useful here to comment Fig.~\ref{openbimodule-cont} further.  In the boundary symplectic fermion theory,
 the space of states decomposes as a direct sum of bosonic $\WP^+$ and
 fermionic $\WP^-$ sectors mixed by the fermionic part of the
 $\gl(1|1)$ algebra. We remind that the  symplectic fermions form an
 $s\ell(2)$ doublet~\cite{Kausch} and the fermionic part of the
 $\gl(1|1)$ is formed by the two zero modes of the fermions.
 Each of the sectors is further decomposed  as a direct sum
 of modules over the product $\SU(2)\boxtimes\Vir(2)$ of
 the two  commuting algebras, $\SU(2)$ and Virasoro, as (see also~\cite{[FFHST]})
 \begin{equation}\label{W-proj-ch-decomp}
\WP^{+} = \bigoplus_{n\in\oN_0} \SUrep{n}\boxtimes\VP_{2n+1,1},\qquad
\WP^{-} = \bigoplus_{n\in\oN-\half} \SUrep{n}\boxtimes\VP_{2n+1,1},
\end{equation}
where $\SUrep{n}$ denotes a $(2n+1)$-dimensional $\SU(2)$-module of
 the (iso)spin $n$, and Virasoro modules $\VP_{n,1}$ are the so-called
 staggered modules introduced in~\cite{Rohsiepe}. Note that we use the
 term `isospin' for $\SU(2)$-modules on the CFT side in order to
 distinguish them from ones on the lattice, where the highest weight
 of an $\SU(2)$-module we call in general by `spin'. In the open case,
 the $U \SU(2)$ symmetry of the symplectic fermion theory~\cite{Kausch}
 and the $U \SU(2)$ part of the full quantum group $\LQG$ are
 actually coincident but this will not be true in the non-chiral
 case. Recall again that the value of $S^z$ for the highest weight at
 (isospin) $n$ is $j=2n$, the horizontal label used in
 Fig.~\ref{openbimodule-cont}.

The staggered
 $\Vir(2)$-modules  are indecomposable and have the following subquotient structure
\begin{equation}\label{sf-chiral-stagg-pic}
   \xymatrix@C=10pt@R=15pt@M=2pt@W=2pt{%
    &&{\VX_{2n+1,1}}\ar[dl]\ar[dr]&\\
    \VP_{2n+1,1} \quad = &\VX_{2n,1}\ar[dr]&&\VX_{2n+2,1}\ar[dl]\\
    &&{\VX_{2n+1,1}}&
 }
\end{equation}
where $\VX_{2n+1,1}$ is the irreducible $\Vir(2)$-module with the
conformal dimension $\Delta_{2n+1,1}=n(2n+1)$, with a non-negative
integer or half-integer $n$. We note that a
south-east arrow represents an action of negative Virasoro modes while
a south-west arrow represent postive modes action. In the
diagram~\eqref{sf-chiral-stagg-pic}, the staggered module is a
`glueing'/extension of two indecomposable Kac modules which are
highest-weight modules. The one in the top composed of
$\VX_{2n+1,1}$ and $\VX_{2n+2,1}$ is the quotient of the Verma module
with the weight $\Delta_{2n+1,1}$ by the singular vector at the level $2n+1$,
and the second Kac module in the bottom composed of
$\VX_{2n,1}$ and $\VX_{2n+1,1}$ is a similar quotient (at the level $2n$) of the Verma module with
the weight $\Delta_{2n,1}$.

The different terms in $\WP^+$ and $\WP^-$ from~\eqref{W-proj-ch-decomp} are in turn connected by
the action of $\gl(1|1)$, resulting into the diagram on
Fig.~\ref{openbimodule-cont}, where each
node is a product of an irreducible $\Vir(2)$-module and an irreducible $\SU(2)$-module  of dimension $j+1$. We
see that the symmetry algebra in this CFT, the centralizer of $\Vir(2)$,
is the semi-direct product of the fermionic part of $\gl(1|1)$ and
 the enveloping algebra $U s\ell(2)$. This centralizer turns out to
 coincide with (a representation of) the quantum group
 $\LQG$ at $\q=i$, which we also denote by ${\cal
 A}_{1|1}$, as on a finite spin-chain.

\subsection{A note on other spin-chains}\label{categorical}
 Note that while we discuss only $\gl(1|1)$ spin-chains here, a diagram identical
 to Fig.~\ref{openbimodule-cont}
 describes the continuum limit of alternating  $\gl(n|n)$ spin-chains.
 We recall that these spin-chains  are
 defined in a similar fashion by introducing projections on
 $\gl(n|n)$-invariants~\cite{ReadSaleur01}. They give also  faithful
 representations of $\TL{N}(0)$ where the
 horizontal action in the bimodule diagram is  due to the symmetry algebra being now
 ${\cal A}_{n|n}$, an algebra which is Morita equivalent to ${\cal
 A}_{1|1}$~\cite{ReadSaleur07-1}.

 A similar analysis for other cases -- representations of $\TL{N}(m)$
 for other values of $m$ -- such as the alternating $\gl(2|1)$ spin-chain
 or the XXZ spin-chains with $\LQG$ symmetry at different roots of unity show that the lattice
 bimodule structure can be used along the same lines to infer
 properties of all known boundary logarithmic CFTs. In particular, the
 staggered Virasoro modules for different central charges abstractly discussed in \cite{GK1} or \cite{MatRid}
 can quickly be recovered in this fashion, at least their subquotient
 structure can be deduced from the bimodules. This opens in particular
 the way to measuring~\cite{DubJacSal2,VasJacSal} indecomposability parameters (also called
 $\beta$ invariants~\cite{KitRid}) characterising Virasoro-module structure completely, or to computing fusion
 rules using an induction procedure~\cite{ReadSaleur07-1,ReadSaleur07-2, [GV]}.

These similarities in subquotient structures and fusion rules indicate
an equivalence between corresponding tensor categories.  Braided
tensor structure or fusion data on the lattice part is given by an
induction bi-functor associated with two lattices of arbitrary sizes
joined to each other. Rigorously, this bi-functor gives a braided tensor
structure only in the infinite size limit, where a construction
of inductive limits of the finite categories of modules over the
Temperley--Lieb algebras is required. We note that a  systematic way to construct
these limits is based on the spin-chain bimodules.  Then, the inductive
limits should be compared with braided tensor categories of modules over the
Virasoro algebras. We believe that a direct construction of
centralizers of the Virasoro algeras in the LCFTs would give
the desired equivalence.

This equivalence can be established by direct calculation
in the case of symplectic fermions, where the continuum limit can be
explicitly carried out (for steps in this direction, see~\cite{GRS1}).
For other cases however, such a calculation seems
completely out of reach since models are only Bethe ansatz solvable,
and precious little is known about their continuum limits, apart from
critical exponents and some indecomposable modules. The equivalence thus
remains a very reasonable conjecture, which will take considerably
more work -- in particular, in defining the continuum or scaling limit --
before being rigorously established. It can nevertheless be easily
understood if one recognizes the similar role played by the quantum
group $\LQG$ both on the lattice and in the continuum~\cite{BFGT,BGT}.

Meanwhile, our strategy consists in postulating that similar
equivalences between the lattice and continuum models are present
in the bulk, or non-chiral, case as well. We shall then study the case
of closed $\gl(1|1)$ spin-chains
in detail, to draw lessons that we will then apply to more
complicated -- and physically interesting -- models.

\section{A reminder of the algebraic aspects in the closed case}\label{sec:remindclosed}

\subsection{The Jones--Temperley--Lieb algebra $\JTL{N}(0)$.}
\label{subsec:TL-alg-def}

A closed periodic spin-chain can be obtained by adding a last generator
$e_{N}^{\gl}$ defined using formula~\eqref{rep-TL-1} but this time
identifying the labels modulo $N$, in particular $f_{N+1}\equiv
f_1$, $f^\dagger_{N+1}\equiv f_1^\dagger$:
\begin{equation}\label{rep-JTL-2}
e_{N}^{\gl}=(f_{N}+f_{1})(f_{N}^\dagger+f_{1}^\dagger).
\end{equation}
The generators $e_j^{\gl}$, with $1\leq j\leq N$, satisfy the
same relations~\eqref{TL-rel1} but now the indices are interpreted
modulo $N$.

The relations~\eqref{TL-rel1} (defined modulo $N$) define an infinite
dimensional associative algebra denoted in our companion papers by
$\PTL{N}$~\cite{GL,GL1}. The algebra $\PTL{N}$ is known in the physics
literature as the periodic Temperley--Lieb
algebra~\cite{MartinSaleur,MartinSaleur1}. In fact, the formulas for
the generators in~\eqref{rep-TL-1} and~\eqref{rep-JTL-2} lead to many more relations; as a
result, the periodic $\gl(1|1)$ spin-chain provides  a non-faithful
representation of
$\PTL{N}$. Moreover, it  provides  a non-faithful
representation of a finite-dimensional quotient --  called \textit{the Jones--Temperley--Lieb}
algebra  $\JTL{N}(m=0)$ -- of a slightly bigger algebra
  enlarging $\PTL{N}$ by a translation operator $u^2$, see \textit{e.g.}~\cite[Sec.~2.1]{GRS2}. We denote this representation in what follows by
$\repgl:\JTL{N}(0)\to\Endo_{\oC}(\chVv)$.

To define the algebra $\JTL{2L}(m)$ (the following description is taken
almost verbatim
from \cite{ReadSaleur07-1}), we first introduce an operator $u^2$
(with inverse $u^{-2}$) which translates any state to the right by 2
sites (so as to be consistent with the distinction of two types of
sites carrying dual representations). We thus have $u^2e_ju^{-2}=
e_{j+2}$ and we also impose the relation $u^{2L} = 1$.  We then define
abstractly an algebra of diagrams as is customary for the ordinary
Temperley--Lieb (TL) algebra, but this time on an annulus (or finite
cylinder), in which a general basis element corresponds to a diagram
of $2L$ sites on the inner, and $2L$ on the outer boundary; the sites
are connected in pairs, but only configurations that can be
represented using lines inside the annulus that do not cross are
allowed.
Multiplication is defined in a natural way on these
diagrams, by joining an inner to an outer annulus, and removing the
interior sites. We emphasize that whenever a closed loop is produced
when diagrams are multiplied together, this loop must be replaced by a
numerical factor $m$ (as for the TL algebra), even for loops that wind
around the annulus, as well as for those that are homotopic to a
point. We also impose that non-isotopic (in the annulus) diagrams connecting  the same sites
 are identified. Then, the resulting algebra  is finite-dimensional
 and is generated by the elements $e_j$ and $u^2$, and
they obey the relations~\eqref{TL-rel1} defined modulo $N$, which
however are not a complete set of defining relations.  We note that the
numerical factor $m$ for winding loops is not a consequence of the
stated relations, but a separate assumption.

In what follows, we consider only the case $m=0$ and use the notation $\JTL{N}\equiv\JTL{N}(0)$.

\medskip

The Hamiltonian operator $H_{\mathrm{per.}}$ in the periodic case will
 be denoted simply by $H$
 and is given by
 \begin{equation}\label{hamil-def}
 H=-\sum_{j=1}^{2L}e^{\gl}_j.
 \end{equation}
This operator is also self-adjoint $H=H^{\dagger}$ with respect to the
inner product, as in the open case described above.  Its eigenvalues
are also real and the eigenvectors can easily be computed. As a
consequence of the indefinite inner product, the self-adjoint
Hamiltonian has non-trivial Jordan cells of rank-two~\cite{GRS1}, see
also Sec.~\ref{sec:scaling} below.

\subsection{The centralizer $\centJTL$ of  $\JTL{N}$}\label{subsec:cent}

In the closed case, while the $\gl(1|1)$ symmetry
generated by $F_{(1)}$, $F^\dagger_{(1)}$, and ${\N}$ remains, the
bosonic $s\ell(2)$ generators $F_{(2)}$ and $F^\dagger_{(2)}$ defined in~\eqref{F2-def} do not
commute with the action of $\JTL{N}$~\cite{GRS1}. What remains of  ${\cal A}_{1|1}$  or $\LQG$ is
only the fermionic  generators
\begin{equation}\label{Fodd-sym-def}
\begin{split}
F_{(2n+1)}&=\!\!\!\!\sum_{\substack{1\leq
  j_1<j_2<\,\dots\\\dots\,<j_{2n+1}\leq 2L}}f_{j_1}f_{j_2}\dots f_{j_{2n+1}},\\
F^\dagger_{(2n+1)}&=\!\!\!\!\sum_{\substack{1\leq
  j_1<j_2<\,\dots\\\dots\,<j_{2n+1}\leq
  2L}}f^\dagger_{j_1}f^\dagger_{j_2}\dots f^\dagger_{j_{2n+1}},
\end{split}
\qquad n\geq0,
\end{equation}
which, together with the bosonic operators  $\N$ and $F_{(2L)}$ and $F^\dagger_{(2L)}$, generate the
centralizer $\centJTL$ of the representation of $\JTL{2L}$.

 For further reference, we recall another and more convenient for us description of the centralizer in terms of standard $\LQG$ generators. We introduce the generators
\begin{equation}\label{LQGodd-LQG-hom}
\EE m\equiv \e^m\E\,\ffrac{\K^2+\one}{2},\qquad \FF n \equiv
\f^n\F\,\ffrac{\K^2+\one}{2},\qquad\quad m,n\geq0,
\end{equation}
which are represented on the spin-chain following the
expressions~\eqref{QG-ferm-2} and~\eqref{QG-ferm-1}.
They generate a subalgebra in $\LQG$ which we denote as
$\LQGodd$, see~\cite[Dfn. 3.3.1]{GRS1}. This subalgebra has the  basis $\EE n \FF m \h^k\K^l$, with
$n,m,k\geq0$ and $0 \leq l \leq 3$.
The positive Borel
subalgebra is generated by $\h$ and $\EE n$ while the negative subalgebra -- by $\h$ and $\FF n$, for $n\geq0$.

The result~\cite[Thm. 3.3.3]{GRS1} is then that the
centralizer $\centJTL$ of the (representation $\repgl$ of) $\JTL{2L}$ on the alternating
periodic $\gl(1|1)$ spin-chain $\Hilb_{2L}$ is the algebra
generated by (the representation of) $\LQGodd$ and $\f^L$, $\e^L$. The correspondence with the fermionic expressions~\eqref{Fodd-sym-def} is $\FF n\sim F_{(2n+1)}$, $\EE n\sim F^\dagger_{(2n+1)}$, and $\N\sim \h$, \textit{etc}.
For more details, see~\cite[Sec.~3.3.4]{GRS1}.

The representation theory of $\LQGodd$ was studied in our second paper~\cite{GRS2} of this series. The
indecomposable $\LQGodd$ modules~$\PPodd{n}$, with integer $1\leq n\leq L$, which
appear in the decomposition of $\Hilb_{2L}$ as direct summands are restrictions to the subalgebra $\LQGodd$
of the well known projective covers $\PP_{1,n}$ over $\LQG$. They are all
indecomposable, with dimension $4n$, the same as the dimension of
$\PP_{1,n}$.  Their subquotient structure is given also
in~\cite[Sec.~3.3]{GRS2}.

\subsection{Bimodule over $\JTL{N}$ and $\centJTL$}\label{ind-chain-bimod-subsec}

Recall that the representation $\repgl$ of $\JTL{N}$ is non-faithful (faithfulness
aspects will be discussed in more detail in our next paper
\cite{GRS4}.) There is thus no evident direct way to get a
decomposition of the periodic $\gl(1|1)$ spin-chain like in the open (faithful) case. For
example, the general theory~\cite{GL0} of projective modules over a
cellular algebra (which includes $\TL{N}(m)$ and $\JTL{N}(m)$
algebras), which can be applied for a faithful representation~\cite{GRS4}, is not directly useful here. Instead, an indirect
strategy is needed, which was discussed in detail in~\cite[Sec~5.2]{GRS2}, and
which we do not recall in detail here, concentrating only on the
essential aspects.

 \begin{figure}\centering
\mbox{}\bigskip\\
\mbox{}\medskip
  \def\svgwidth{470pt}
    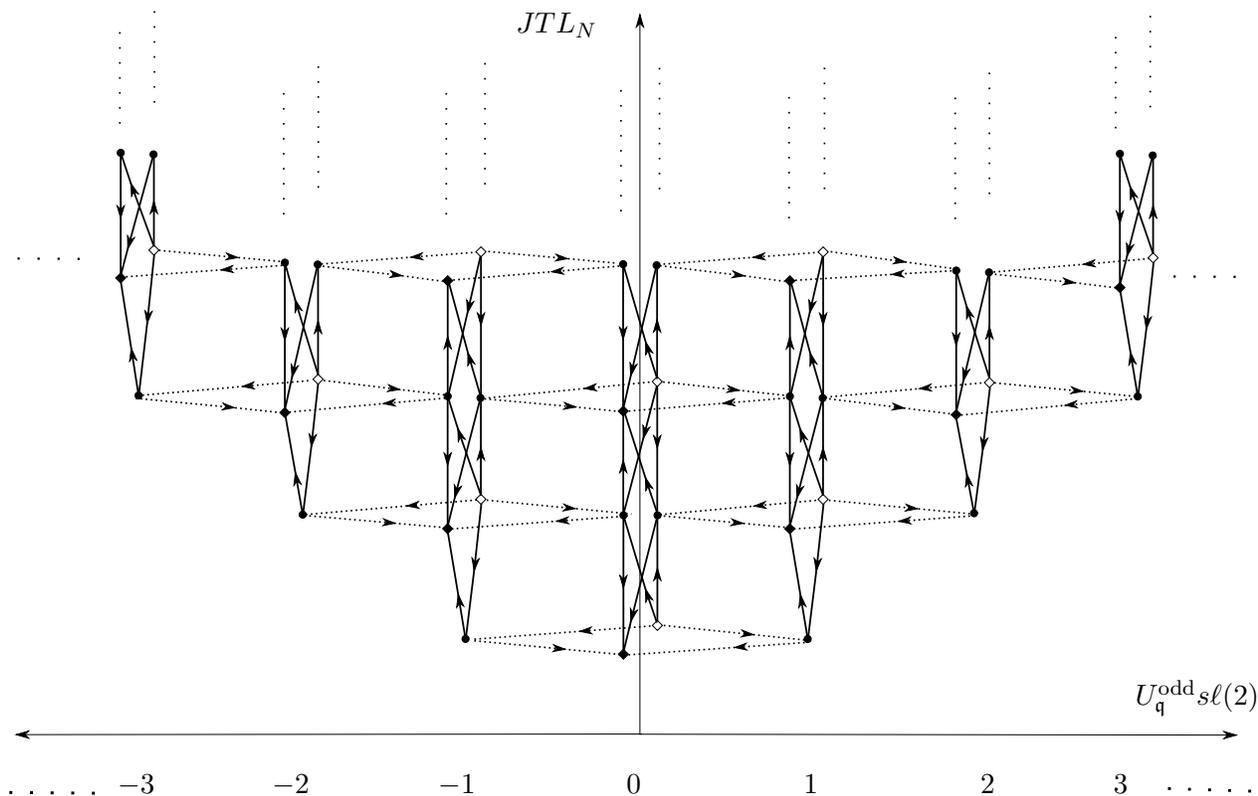
      \caption{Bimodule over the pair $(\JTL{N},\LQGodd)$ of commuting
      algebras.  The action of $\JTL{N}$ is depicted by vertical
      arrows while the action of $\LQGodd$ is shown by dotted
      horizontal lines. Each label $j$ in the horizontal axis corresponds to the sector for $S^z$ and the label runs from $-L$ on the left to $j=L$ on 	
      the right. Each vertical tower at a label $j$ is the diagram for $\APrTL{j}$.
      The first horizontal layer at the bottom contains four
      nodes $(\TLX_1)$ and dotted arrows mixing them compose the
      $\LQGodd$-module $\PPodd{1}$. The second layer contains eight
      nodes $(\TLX_2)$ and the dotted arrows depict the action in the
      indecomposable module $\PPodd{2}$, {\it etc.} We
      suppress long-range arrows representing action of the generators
      $\FF{>0}$ and $\EE{>0}$ in order to simplify diagrams. For
      example, the second layer of the bimodule contains in addition
      four long arrows going from the node $\diamond$ at $j=\mp1$ to the
      node $\bullet$ at $j=\pm2$, and from the node $\bullet$ at
      $j=\pm2$ to the node {\scriptsize$\vardiamond$} at $j=\mp1$. We also note that {\scriptsize$\vardiamond$}'s have only incoming arrows and their union is the socle (maximum semisimple submodule or the bottom level) of the bimodule; further, the union of $\bullet$'s is the socle of the quotient by the socle -- the middle level, while the $\diamond$'s have only outgoing arrows and their union is the top level.}
    \label{JTL-Uqodd-bimod}
    \end{figure}

We first give a diagram describing the subquotient structure of the
bimodule $\Hilb_{2L}$ over the pair $(\JTL{N},\centJTL)$.  The two
commuting actions are presented in Fig.~\ref{JTL-Uqodd-bimod} where we
show a direct sum of indecomposable spin-chain modules $\APrTL{j}$ over
$\JTL{N}$, and simple subquotients over $\JTL{N}$ will be denoted below by
$(\TLX_{|j|})$. The direct sum is depicted as a (horizontal) sequence of
diagrams for $\APrTL{j}$ from $j=-L$ on the left to $j=L$ on the
right. Each node in the diagram is a simple subquotient over the
product $\JTL{N}\boxtimes\LQGodd$. The action of $\JTL{N}$ is depicted
by vertical arrows while the action of $\LQGodd$ is shown by dotted
horizontal lines connecting different $\JTL{N}$-modules.

In the diagram on  Fig.~\ref{JTL-Uqodd-bimod}, the first horizontal
layer (at the bottom)  contains the space of ground states
and it consists of four
nodes, which are simple $\JTL{N}$-modules $(\TLX_1)$, and dotted
arrows mixing them. These nodes and arrows describe the indecomposable $\LQGodd$-module
$\PPodd{1}$. The second layer contains eight nodes corresponding to $(\TLX_2)$
and the dotted arrows contribute to the indecomposable module
$\PPodd{2}$, {\it etc}. We emphasize that we do not draw long-range arrows
representing action of the generators $\FF{>0}$ and $\EE{>0}$ in
modules $\PPodd{n>1}$ in order to simplify diagrams but the arrows can
be easily recovered using
the subquotient structure of $\PPodd{n}$ given in~\cite[Sec.~3.3]{GRS2} -- for example, the second layer of the
bimodule contains in addition four long arrows going from the node
$\diamond$ at $j=\mp1$ to the node $\bullet$ at $j=\pm2$, and from the
node $\bullet$ at $j=\pm2$ to the node {\scriptsize$\vardiamond$} at $j=\mp1$. With this
comment about arrows in mind, the reader can compare complexity of
this bimodule with the open-case bimodule in
Fig.~\ref{openbimodule-fin}.

As a $\JTL{N}$-module, the $\Hilb_{N}$ has a decomposition
\begin{equation}\label{decomp-JTL}
\Hilb_N|_{\rule{0pt}{7.5pt}%
\JTL{N}} =  \bigoplus_{j=-L+1}^{L-1} \APrTL{j}\boxtimes\Xodd{j} \oplus
\APrTL{L}\boxtimes\Xodd{L},
\end{equation}
where $N=2L$, $\Xodd{j}$ is the one-dimensional (with $S^z=j$) and  $\Xodd{L}$ is the two-dimensional simple
$\centJTL$-modules (the representation theory of the centralizer $\centJTL$ is described
 in~\cite[Sec.~3]{GRS2}), and
 $\APrTL{L}$ is the one-dimensional $\JTL{2L}$-module with the action of $u^2=(-1)^{L-1}\one$.
 For any sector with non-zero  $j=S^z$,
the subquotient structure for $\APrTL{j}$ is given in
Fig.~\ref{FF-JTL-mod}, on the right, while the tower for $\APrTL{0}$
is presented on the left.  All the towers are ended by the pair of
simple subquotients $(\TLX_L)$.  We note that the
two simple subquotients at each level of the ladders are
isomorphic. The Hamiltonian $H$ from~\eqref{hamil-def} acts by Jordan
blocks of rank $2$ on each pair of isomorphic simple subquotients
with one at the top (having only outgoing arrows) and the second
subquotient in the socle of
the module (having only ingoing arrows). The Jordan block structure is due to presence of zero
fermionic modes in the Hamiltonian as it is observed in~\cite{GRS1},
see also Sec.~\ref{sec:ham-spec}.
 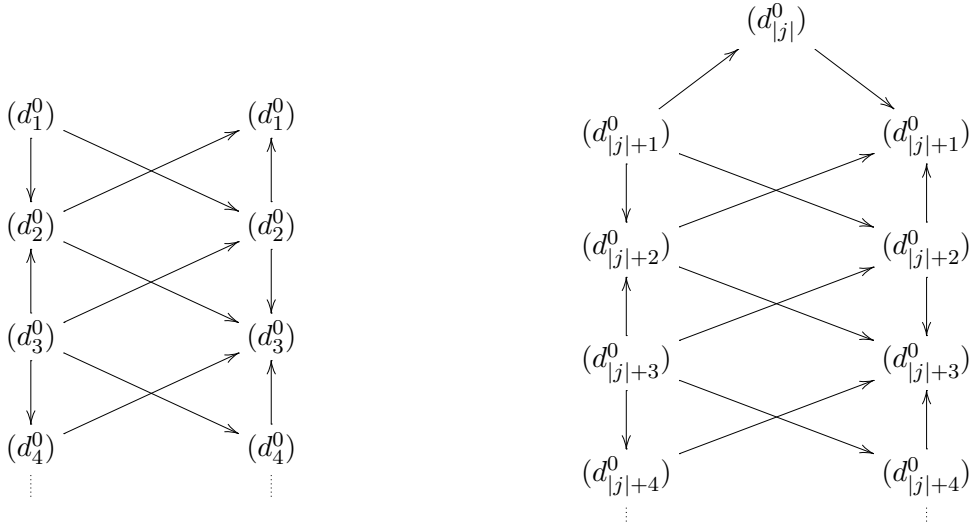
\begin{figure}\centering
\begin{equation*}
  \xymatrix@R=24pt@C=30pt{
    &{\mbox{}}&&\\
    {(\TLX_1)}\ar@[][d]\ar[drr]
    &&{(\TLX_1)}\\
    {(\TLX_2)}\ar@[][urr]\ar[drr]
    &&{(\TLX_2)}\ar[u]\ar@[][d]\\
    {(\TLX_3)}\ar[u]\ar@[][d]\ar@[][urr]\ar[drr]
    &&{(\TLX_3)}\\
    {(\TLX_4)}\ar@[][urr] \ar@{.}[]-<0pt,18pt>
    &&{(\TLX_4)}\ar[u]
    \ar@{.}[]-<0pt,18pt>}
\qquad\qquad\qquad
  \xymatrix@=22pt
  {{}&(\TLX_{|j|})\ar[dr]&&\\
    {(\TLX_{|j|+1})}\ar[ur]\ar@[][d]\ar[drr]
    &&{(\TLX_{|j|+1})}\\
    {(\TLX_{|j|+2})}\ar@[][urr]\ar[drr]
    &&{(\TLX_{|j|+2})}\ar[u]\ar@[][d]\\
    {(\TLX_{|j|+3})}\ar[u]\ar@[][d]\ar@[][urr]\ar[drr]
    &&{(\TLX_{|j|+3})}\\
    {(\TLX_{|j|+4})}\ar@[][urr] \ar@{.}[]-<0pt,18pt>
    &&{(\TLX_{|j|+4})}\ar[u]
    \ar@{.}[]-<0pt,18pt>}
\end{equation*}
      \caption{The two-strands structure of the spin-chain
      $\JTL{N}$-modules $\APrTL{j}$ for $j=0$ on the left and $j\ne0$
      on the right side. The towers are ended by the pair of
      $(\TLX_L)$.}
    \label{FF-JTL-mod}
    \end{figure}

 We note that the $\JTL{N}$-modules $\APrTL{j}$ in
Fig.~\ref{JTL-Uqodd-bimod} are drawn in opposite direction `from
bottom to top' comparing to diagrams in Fig.~\ref{FF-JTL-mod}, in
order to  the space of ground states to be in the bottom of the diagram.

Finally, we recall some information about the simple $\JTL{N}$
modules, which we denote in~\cite{GRS2} by $\AIrrTL{j}{(-1)^{j+1}}$ or more
conveniently by $(\TLX_{j})$ above. The number $d_j^0$ is their dimension,
obtained from the dimension of the spin-chain $\JTL{N}$-modules $\APrTL{j}$
given by
\begin{equation}
\dim \APrTL{j} \equiv \widehat{d}_j=\binom{2L}{L+j}
\end{equation}
which is just the dimension of the sector with $S^z=j$ in the related
  XX spin chain with $2L$ sites. We  then have from the subquotient
  structure for $\APrTL{j}$ in Fig.~\ref{FF-JTL-mod} that
\begin{equation*}
\dim \AIrrTL{j}{(-1)^{j+1}} \equiv d_j^0=\sum_{j'\geq j} (-1)^{j'-j}d_{j'}\qquad
\text{with}\qquad d_j = \binom{2L}{L+j} -  \binom{2L}{L+j+1},
\end{equation*}
which can be written as
\begin{equation}\label{dim-irr-per}
d_j^0=
\binom{2L-2}{L-j}-
\binom{2L-2}{L-j-2},
\end{equation}
where $0\leq j\leq L$. Note that $d_0^0=0$ and we thus have that all simple modules over $\JTL{N}$ that appear in this spin-chain are those $\AIrrTL{j}{(-1)^{j+1}}$ with $1\leq j\leq L$, see also more details in our second paper~\cite{GRS2}.
We note also the obvious identities
\begin{equation*}
\sum_{j=-L}^L \widehat{d}_j=\sum_{j=1}^L 4d_j^0=2^{2L}
\end{equation*}
which give the full dimension of the $\gl(1|1)$ spin-chain indeed.

It turns out that in this very degenerate case of $\JTL{N}$ representations the $\AIrrTL{j}{(-1)^{j+1}}$ modules are also simple for the subalgebra $\TL{N}$. Of course, the difference with the open case is in the structure of indecomposable modules -- the complexity of towers in Fig.~\ref{JTL-Uqodd-bimod} can be compared with the simpler `diamond'-type structure of TL modules in Fig.~\ref{openbimodule-fin}.

\subsection{The antiperiodic spin-chain}\label{subsec:antiper-mod}
We can also consider the alternating $\gl(1|1)$ spin-chain with
antiperiodic conditions for the fermions, obtained by setting
$f^{(\dagger)}_{2L+1}=-f^{(\dagger)}_1$. The generators $e_j$, for
$1\leq j\leq 2L-1$, have the same representation~\eqref{rep-TL-1}
while the last generator is then given by
\begin{equation}\label{rep-JTL-anti}
e_{2L}=(f_{2L}-f_{1})(f_{2L}^\dagger-f_{1}^\dagger),
\end{equation}
to be compared with~\eqref{rep-JTL-2}.
 This expression provides  a
representation of another quotient (different from $\JTL{N}$) of the  affine TL algebra:
in the diagram language, non-contractible loops now should be replaced
 by the weight~$2$ (the
 dimension of the fundamental or its dual, instead of the
 superdimension in the periodic case);
the relation $u^{N}=(-1)^j$ is also imposed in the sector
with $2j$ through-lines.
 We will call the corresponding algebra $\JTL{N}^{tw}$, see a precise
definition in our \second paper~\cite[Sec.~6.2]{GRS2}.

  We emphasize that this antiperiodic $\gl(1|1)$ spin chain does not
have $\gl(1|1)$ symmetry any longer. Instead, we have~\cite[Thm.3.4.1]{GRS1} that
the centralizer of  $\JTL{N}^{tw}$  is generated
 by the $s\ell(2)$ generators $\e$ and $\f$, or $F_{(2)}^\dagger$ and
 $F_{(2)}$ defined in~\eqref{F2-def}.

We then recall  the bimodule structure~\cite[Sec.~6.2]{GRS2} over the pair $(\JTL{N}^{tw},U s\ell(2))$:
\begin{equation}\label{antip-bimod}
\Hilb_N|_{\rule{0pt}{7.5pt}%
\JTL{N}^{tw}\boxtimes U s\ell(2)} = \bigoplus_{j=0}^{L}\AIrrTL{j}{(-1)^j} \boxtimes \XX_{1,j+1}
\end{equation}
where $\XX_{1,j}$ is the $j$-dimensional simple $U s\ell(2)$-module
 and $\AIrrTL{j}{(-1)^j}$ is a simple module over $\JTL{N}^{tw}$.
 The dimension of $\AIrrTL{j}{(-1)^j}$ is $2d^0_{j+1} + d^0_{j} +
d^0_{j+2}$ and is
\begin{equation}\label{dimIrrATL-anti}
\dim \AIrrTL{j}{(-1)^j} = \widehat{d}_j -\widehat{d}_{j+2} = \binom{2L}{L+j} - \binom{2L}{L+j+2}.
\end{equation}
Recall that the dimension of each $S^z=j$ sector $\Hilb_{[j]}$ is $\binom{2L}{L+j}$.


\section{JTL algebra and Howe duality}\label{Howe}

We now give an interpretation of the simple modules over
 $\JTL{N}$ and $\JTL{N}^{tw}$ with the dimensions~\eqref{dim-irr-per}
and~\eqref{dimIrrATL-anti}, respectively, from the point of view of  representation  theory of
 symplectic Lie algebras.
 We first recall a convenient  notation for  lattice fermions and then describe the representation theory of $\JTL{N}$ in the periodic and anti-periodic spin chains in the  context of (a non-semisimple version of)
 Howe duality.

 \subsection{A Lie algebra of fermion bilinears}\label{Howe-ss-0}
Recall that the $e_j$ generators of the $\JTL{N}$ (or $\JTL{N}^{tw}$) algebras in the (anti-)periodic spin-chains are linear combinations of bilinears in the fermions $f_j$ and $f_j^{\dagger}$. For the periodic case, we have the representation
\begin{equation*}
e_j^{\gl}= (f_j+f_{j+1})(f_j^\dagger+f_{j+1}^\dagger),\qquad 1\leq
j\leq 2L,
\end{equation*}
identifying the labels modulo $N$, in particular $f_{N+1}\equiv
f_1$, and $f^\dagger_{N+1}\equiv f_1^\dagger$
 (see also expressions in~\eqref{rep-TL-1} and~\eqref{rep-JTL-2}).  For the anti-periodic model, we have a similar expression~\eqref{rep-JTL-anti}, where the fermions are anti-periodic. It turns out that the commutators of these combinations of bilinears can be expressed again in fermionic bilinears, and of course   belong to the JTL algebra. Therefore, they generate a finite-dimensional Lie algebra. On the other hand, the spin-chain images of the JTL algebra contain many operators which are not bilinears. These non-bilinear operators are generated by the Lie algebra elements\footnote{Note that the translation generator $u^2$ action in our spin-chains  can be essentially expressed by products of $e_j$'s, see more concrete statements in~\cite{GRS1,GRS2}.} because the $e_j$ generators belong to the Lie algebra. To say things differently, the full image of  the JTL algebra is isomorphic to an enveloping algebra of the Lie algebra of the fermionic bilinears (of course, it is also true for the open case discussed in Sec.~\ref{sec:remind} that the TL algebra, which is a subalgebra in the JTL algebra, is  generated by a Lie algebra of fermionic bilinears).
We begin our analysis with the anti-periodic spin-chain which is much simpler  than the periodic one because of its semi-simplicity.

The centralizer of  $\JTL{N}^{tw}$  is generated as well by special bilinears in the fermions and these are
 the $s\ell(2)$ generators $\e$ and $\f$, or $F_{(2)}^\dagger$ and
 $F_{(2)}$ defined in~\eqref{F2-def}.  It is a well-known fact due to Howe~\cite{H3} that,  in the Clifford algebra $\Clif{4L}$, the centralizing algebra of this image of $U s\ell(2)$ -- which is the enveloping algebra of the symplectic Lie algebra $\ssp_2$ as well -- is the enveloping algebra of $\ssp_{2L}$. Note that the dimensions~\eqref{dimIrrATL-anti} of the simple $\JTL{N}^{tw}$-modules coincide indeed  with dimensions of fundamental representations of $\ssp_{2L}$. We can thus easily get the following theorem.
  \begin{thm}\label{thm:Lie-JTLtw}
The image $\repgl\bigl(\JTL{N}^{tw}\bigr)$ of the twisted Jones--Temperley--Lieb algebra in the anti-periodic $\gl(1|1)$ spin-chain is isomorphic to the image of a representation of the enveloping algebra $U \ssp_{2L}$ with all simple $\ssp_{2L}$-modules corresponding to one-column
 Young diagrams.
\end{thm}
 Recall that the anti-periodic spin-chain is completely reducible as a $\JTL{N}^{tw}$-module -- see the corresponding semisimple $(\JTL{N}^{tw},U\ssp_{2})$-bimodule in Sec.~\ref{subsec:antiper-mod} -- and we thus have the double-centralizing property~\cite[Thm.~4.1.13]{goodwall}. Then, the proof of Thm.~\ref{thm:Lie-JTLtw} is a simple consequence of the classical ($\ssp(n)$,\,$\ssp(m)$)-type Howe duality, where $n=2$ and $m=2L$ in our case.

We next describe explicitly the representation of $\ssp_{2L}$ from Thm.~\ref{thm:Lie-JTLtw}, and introduce a  more convenient fermion basis which we will use below in studying scaling limits.

\subsubsection{The Clifford algebra}\label{sec:ham-spec}
The periodic spin chain admits the natural   action of  a Clifford algebra $\Clif{4L}$ with
$4L$ generators: $\chi_{p}$, $\eta_p$ and their
adjoints $\chi^{\dagger}_{p}$, $\eta^{\dagger}_p$, where we set
$p=m\pi/L$ and $0\leq m\leq L-1$. As operators, they are defined in
App.~\Appfourier in~\eqref{eq:chi-eta-def} and  satisfy the following anti-commutation relations
\begin{equation}\label{chi-eta-rels-0}
\left\{\chi^{\dagger}_p,\chi_{p'}\right\} = \left\{\eta^{\dagger}_p,\eta_{p'}\right\} =
\delta_{p,p'}, \qquad
\left\{\chi_p,\eta_{p'}\right\} =
\left\{\chi^{\dagger}_p,\eta^{(\dagger)}_{p'}\right\} =
\left\{\eta^{\dagger}_p,\chi^{(\dagger)}_{p'}\right\} = 0.
\end{equation}

The Hamiltonian from~\eqref{hamil-def} can then be written~\cite[Sec.~4.1]{GRS1} in this Clifford algebra as
\begin{equation}\label{eq:def-hamilt-gl}
H=
2\sum_{\substack{p=\step\\\text{step}=\step}}^{\pi-\step}\sin{p}\bigl(\chi^{\dagger}_p\chi_{p}
- \eta^{\dagger}_p\eta_{p}\bigr) + 4 \chi^{\dagger}_0\eta_0,\qquad \step=\ffrac{\pi}{L},
\end{equation}
{\it i.e.},
the $\chi^{(\dagger)}_p$ and $\eta^{(\dagger)}_p$ fermions are creation and
annihilation operators that generate an eigenvector of
$H$  of the energy $2\sin(p)$ and momentum $p$. We note that $H$ is non-diagonalizable and  its off-diagonal part
$\chi^{\dagger}_0\eta_0$  generates  Jordan cells of rank $2$.
The creation operators generate  all root-
and eigen-vectors of $H$ from the space of ground states.
The latter has the following structure:
\begin{equation}\label{diag:space-gr-st}
   \xymatrix@=20pt{
     &\lvac\ar[dl]_{\chi^{\dagger}_0} \ar[dr]^{-\eta_0} \ar[dd]^(0.45){H}&&\\
     \phip\ar[dr]_{\eta_0}
       &&\phim\ar[dl]^{\chi^{\dagger}_0} &\\
     &\vac&&
   }
\end{equation}
where the two bosonic
states -- the vacuum $\vac$ and the state $\lvac$ -- have $S^z=0$ and
 form a
two-dimensional Jordan cell for the lowest eigenvalue for $H$, while the two fermionic states $\phip$ and $\phim$ belong to the
sectors with $S^z=+1$ and $S^z=-1$, respectively. We also
showed in~\eqref{diag:space-gr-st} the action of the quantum-group generators $\F\sim\eta_0$
and $\E\sim\chi^{\dagger}_0$.

All the excitations over the ground states
$\lvac$ and $\vac$ are thus generated by the free action of the
 fermionic  creation modes which are $\chi^{\dagger}_{p}$ and $\eta_p$, for $p\in(0,\pi)$,
and
$\eta^{\dagger}_0$ and
$\chi_0$ (see precise definitions of these operators in~\eqref{eq:chi-eta-def}). The
annihilation modes acts as
\begin{equation}
\chi_{p}\vac = \eta^{\dagger}_p\vac = \chi^{\dagger}_0\vac =
\eta_0\vac =0, \qquad p\in(0,\pi).
\end{equation}

\medskip

For the antiperiodic spin-chains, we have also the action of  a Clifford algebra with
$4L$ generators  $\chi_{p}$, $\eta_p$ and
adjoint ones $\chi^{\dagger}_{p}$, $\eta^{\dagger}_p$, where now the
momenta run over a different set: $p=(m-\half)\pi/L$ and $1\leq m\leq L$. In this case, there are
no zero modes, and in particular the ground state of the Hamiltonian with
anti-periodic conditions is non-degenerate~\cite{GRS1}.

\subsubsection{A special basis in (twisted) JTL algebras}
We now go back to the anti-periodic model and describe explicitly the representation of $\ssp_{2L}$ from Thm.~\ref{thm:Lie-JTLtw}.
First, recall the well known fact  that the bilinears in  the generators of the Clifford algebra $\Clif{4L}$ introduced in~\eqref{chi-eta-rels-0}  that commute with the operator $S^z=2\h$ (or with the fermion number operator)  give a basis in the Lie algebra
$\gl_{2L}$. Indeed, introducing the elementary matrices $\Egl_{m,n}$ with matrix elements unity  at the corresponding $(m,n)$-th position and zero otherwise (a standard basis
in $\gl_{2L}$), we can write them in terms of  fermions as
\begin{equation}\label{gl-stand-basis}
\begin{split}
\Egl_{m,n} &= \chid_p\,\chi_q,\qquad\quad\Egl_{m,L+n} = \chid_p\,\bbeta_q,\\
\Egl_{L+m,n} &= \bbetad_p\,\chi_q,\qquad\Egl_{L+m,L+n} =
\bbetad_p\,\bbeta_q,
\end{split}, \qquad 1\leq m,n\leq L
\end{equation}
where we set  $p=(m-\half)\pi/L$ and
$q=(n-\half)\pi/L$,
 and we used also the notation $\bbeta_p=\eta_{\pi-p}$,
$\bbetad_p=\etad_{\pi-p}$. One can  check using the
relations~\eqref{chi-eta-rels-0} that the defining relations
for $\gl_{2L}$:
\begin{equation}\label{rel-gln-stand}
[\Egl_{m,n},\Egl_{k,l}] = \delta_{n,k}\Egl_{m,l} - \delta_{m,l}\Egl_{k,n}, \qquad  1\leq m,n,k,l\leq2L.
\end{equation}
 are satisfied. We then note that the linear combinations
\begin{equation}\label{sp-stand-basis}
\begin{split}
\Asp_{m,n} &= \Egl_{m,n} - \Egl_{L+n,L+m} ,\qquad\Bsp_{m,n} =
\Egl_{m,L+n} + \Egl_{n,L+m},\\
\Csp_{m,n} &= \Egl_{L+m,n} + \Egl_{L+n,m},
\end{split}\qquad\qquad\quad 1\leq m,n\leq L,
\end{equation}
span a Lie subalgebra in $\gl_{2L}$ isomorphic to $\ssp_{2L}$ (note that $\Bsp_{n,m} = \Bsp_{m,n}$ and $\Csp_{n,m} = \Csp_{m,n}$), and
these combinations commute with the action of $\ssp_{2}$ spanned by $\e$, $\f$, and $\h$.
Moreover, these combinations span the maximum Lie subalgebra in $\gl_{2L}$ commuting with the generators of $\ssp_{2}$: it is straightforward to check using~\eqref{QG-ferm-2} and transformations in App.~\Appfourier that the complement of the subspace $\ssp_{2L}$ in the $\gl_{2L}$ does not commute with the action of $\e$ and $\f$.
We then recall~\cite{H3} that the generators of the Howe-dual for the $\ssp_{2}$ are also bilinears in the fermionic operators.
We thus obtain that the Lie algebra $\ssp_{2L}$ from Thm~\ref{thm:Lie-JTLtw} acts as in~\eqref{sp-stand-basis} and its centralizer is the $U s\ell(2)$ acting by~\eqref{QG-ferm-2}. Because of the double-centralizing property the basis elements  of $\ssp_{2L}$ therefore
generate
the action of the centralizer of $U s\ell(2)$ which is $\repgl\bigl(\JTL{N}^{tw}\bigr)$,
as was stated in Thm.~\ref{thm:Lie-JTLtw}. We do not give here explicit expressions of $e_j$'s in terms of these basis elements of $\ssp_{2L}$ but will give some formulas below. What is important to note now is  that taking products of the $\ssp_{2L}$ basis elements we obtain a special basis in the image of $\JTL{N}^{tw}$ which is used later for taking the scaling limit of the JTL algebras.

Note finally that the  $\Asp_{n,n}$ generators are basis elements in the Cartan subalgebra of $\ssp_{2L}$. Then, the simple modules  $\AIrrTL{j}{(-1)^{j}}$ over
$\JTL{N}^{tw}$ are  highest-weight
representations of  $\ssp_{2L}$ with the weights\footnote{Actually, for our choice of Cartan elements one should replace $1$ by $-1$ in order to obtain correct weights with respect to $\Asp_{n,n}$.} $(1,1,\dots,1,0,\dots,0)$, which are  sequences of length $L$ of consecutive $1$'s and then $0$'s,  with $j$ of $0$'s and $0\leq j\leq L$. We thus obtain a decomposition~\eqref{antip-bimod} of the anti-periodic spin-chain with respect to the two commuting Lie algebras, $\ssp_2$ and $\ssp_{2L}$, where the $\ssp_2$ representation with $j$ boxes corresponds to the $j$th fundamental  (one-column) representation for $\ssp_{2L}$.

We then extend this classical Howe duality to the non-semisimple case which corresponds to the periodic model. The final result is given in Thm.~\ref{Thm:US-JTL-iso} and Cor.~\ref{cor:Howe-nss} but first we discuss the semi-simple part of the action of the $\JTL{N}$ algebra which parallels the previous discussion.

\subsection{Semisimple part of $\JTL{2L}$ and Lie algebra $\ssp_{2L-2}$ }\label{Howe-ss}

For the periodic model,
we introduce
the standard basis for $\gl_{2L}$ as in~\eqref{gl-stand-basis},
where now $0\leq m,n\leq L-1$ and $p=m\pi/L$, and
$q=n\pi/L$. Note once again that we have zero modes $\chi_0^{(\dagger)}$ and $\eta_0^{(\dagger)}$ in this case, with the identification $\bar{\eta}_0^{(\dagger)}=\eta_0^{(\dagger)}$. Then, the linear combinations
from~\eqref{sp-stand-basis}, with $1\leq m,n\leq L-1$, span a Lie
subalgebra $\ssp_{2L-2}$. Note that we do not include zero modes because half of them do not commute with the symmetry algebra $\LQGodd$.
 Recall that the $\LQGodd$ generators $\EE m$ and $\FF n$ are represented
in the spin-chain by $\e^m\E$ and $\f^n\F$, see~\eqref{LQGodd-LQG-hom} and note that
$\K^2=\one$ for even $N$. We  can write~\cite{GRS1} fermionic expression
for these generators in the $\chi$-$\eta$ notation as
\begin{equation}\label{eq:EE-chi-eta}
\EE m = (-1)^{S^z}\sqrt{N}\left(\sum_{p=\step}^{\pi-\step}
\eta^{\dagger}_p\chi^{\dagger}_{\pi-p}\right)^m\chid_0,\qquad
\FF n = -i \sqrt{N}\left(\sum_{p=\step}^{\pi-\step}
\eta_p\chi_{\pi-p}\right)^n\eta_0.
\end{equation}
It is then easy to check that the combinations of bilinears that give a basis in the $\ssp_{2L-2}$ commute with this action of $\LQGodd$.
Moreover, we  show
below that the bilinears in fermions (\ref{sp-stand-basis}) are the only bilinears that
do not contain zero modes and belong to the image of $\JTL{N}$ under
the representation $\repgl$. Put a bit differently, we show that the operators
$\Asp_{m,n}$, $\Bsp_{m,n}$, and $\Csp_{m,n}$ generate the semisimple
part of the associative algebra $\repgl\bigl(\JTL{N}\bigr)$, while we will see later that similar combinations
containing zero modes $\chid_0$ and $\eta_0$ generate its Jacobson
radical.  Indeed, the dimensions~\eqref{dim-irr-per} of the simple $\JTL{2L}$-modules coincide  with dimensions of fundamental representations  of $\ssp_{2L-2}$, as in the anti-periodic case but now for a smaller symplectic Lie algebra. This observation suggests the following lemma.
\newcommand{\rad}{\mathrm{Rad}}
\begin{lemma}\label{lem:ssJTL-sp}
The image of the representation~\eqref{sp-stand-basis} of the enveloping algebra $U \ssp_{2L-2}$ is isomorphic to the quotient $\repgl\bigl(\JTL{N}\bigr)/\rad$, where $\rad$ is the Jacobson radical of the image of the Jones--Temperley--Lieb algebra in the periodic $\gl(1|1)$ spin-chain.
\end{lemma}
\begin{proof}
We note that all simple modules $\AIrrTL{j}{(-1)^{j+1}}$ over $\JTL{N}$ that appear in  the spin-chain module $\Hilb_N$ appear also in the socle (maximum semisimple submodule) of $\Hilb_N$, see Fig.~\ref{FF-JTL-mod}. Then, in order to study the semisimple part of the action of $\JTL{N}$, it is enough to consider the action restricted on the socle of $\Hilb_N$. Recall that  simple modules over $\JTL{N}$ are also simple modules over its subalgebra $\TL{N}$, see Sec.~\ref{ind-chain-bimod-subsec}, and the centralizer of the $\TL{N}$ action is given the Lusztig's (restricted specialization of) $\LQG$ at $\q=i$~\cite{M1} -- the one with the divided powers $\e$ and $\f$. Therefore,  the centralizer of the action restricted on the socle  -- the intersection of the kernels of $\F$ and $\E$
in $\Hilb_N$ -- is  given by a representation of $U \ssp_2$ realized by the divided powers of $\LQGi$. Since the socle is freely generated from the vacuum state $\vac$ by the action of a Clifford algebra $\Clif{4L-4}$  with
$4L-4$ generators   $\chi^{(\dagger)}_{p}$, $\eta^{(\dagger)}_p$, where
$p=m\pi/L$ and $1\leq m\leq L-1$, we have a classical symplectic Howe duality~\cite{H3} between $\ssp_{2}$ and $\ssp_{2L-2}$. Indeed, the socle is a multiplicity-free semisimple bimodule over $\JTL{N}$ and $U \ssp_{2}$, where both algebras are represented  by appropriate bilinears in the generators of  $\Clif{4L-4}$. By  Howe duality, we obtain that the action of $\JTL{N}$ on the socle of $\Hilb_N$, and thus its semisimple part on $\Hilb_N$, is generated by the $\ssp_{2L-2}$ action.
\end{proof}

In the open case, the simple modules are the same as for $\JTL{N}$ in the periodic model. Therefore, we conclude that the semisimple part of $\TL{N}$, or the quotient $\TL{N}/\rad$ by its Jacobson radical, is  also generated by an  $\ssp_{N-2}$ action.

Using Lem.~\ref{lem:ssJTL-sp} and a correspondence between weights in the symplectic Howe duality~\cite{H3,BrockerDieck}, we obtain the following  corollary.
\begin{cor}\label{cor:JTL-sp-simples}
The simple modules $\AIrrTL{j}{(-1)^{j+1}}$, with $1\leq j\leq L$, over $\JTL{2L}$ (or $\TL{2L}$) are simple
modules over the enveloping algebra for $\ssp_{2L-2}$. The dimensions~\eqref{dim-irr-per}
of all these simple modules correspond to all highest-weight
representations of $\ssp_{2L-2}$ of the weights $(1,1,\dots,1,0,\dots,0)$, which are  sequences of length $L-1$ of consecutive $1$'s and then $0$'s, with $j-1$ of $0$'s.
\end{cor}
We note once again that the generators $\Asp_{n,n}$, now with $1\leq n\leq L-1$, span a basis in the Cartan subalgebra of the  $\ssp_{2L-2}$.
The diagonal part of the Hamiltonian $H$ has a very simple expression in terms of these generators:
\begin{equation}
H= 2\sum_{m=1}^{L-1}
\sin{(p)}\, \Asp_{m,m} + 4\,\Egl_{0,L}.
\end{equation}
This allows us to describe $\JTL{N}$ simple modules as highest-weight representations of $\ssp_{2L-2}$ where highest-weight vectors play the role of charged vacua for the Hamiltonian. The vacuum state $\vac$ in the  space  of ground states coincides with the  highest-weight vector of the unique $\ssp_{2L-2}$ module of  the weight $(1,1,\dots,1,1)$, the next one corresponds to the weight $(1,1,\dots,1,0)$, \textit{etc}.

\subsection{The image  of $\JTL{N}$ as a Lie algebra representation}
We  introduce now special
elements  in $\JTL{N}$ which span a subspace of all elements in $\JTL{N}$
that are  bilinear in the $\chi$-$\eta$ fermions. We begin with the definition of the Lie algebra $\interfin{N}$.

\begin{dfn}\label{S_N-def}
We define the  Lie algebra $\interfin{N}$ to be generated by $\Asp_{m,n}$,
$\Bsp_{m,n}$, $\Csp_{m,n}$ from~\eqref{sp-stand-basis} and the elementary matrices
$\Egl_{0,n}$, $\Egl_{0,L}$, $\Egl_{0,L+n}$, $\Egl_{m,L}$, and
$\Egl_{L+m,L}$,
where $1\leq m,n\leq
L-1$. This Lie algebra can be schematically depicted by  matrices
of the form (in the standard basis of $\gl_{2L}$)
\begin{equation}\label{interfin-def}
\interfin{N}:\quad
\begin{pmatrix}
0 &  \times & \dots & \times   & \times  & \times & \dots & \times\\
0 &  &  &            & \times  &   &       & \\
0 &  & \Asp_{m,n} &  & \times  &   & \Bsp_{m,n}    & \\
0 &  &  &            & \times  &   &         & \\
0 & 0 & \dots & 0    & 0  & 0 & \dots & 0 \\
0 &  &  &            & \times  &   &       & \\
0 &  & \Csp_{m,n} &  & \times  &   &     & \\
0 &  &  &            & \times  &   &         &
\end{pmatrix},
\end{equation}
where the crosses $\times$ stand for the corresponding elements
$\Egl_{0,n}$, $\Egl_{0,L}$, $\Egl_{0,L+n}$, $\Egl_{m,L}$, and
$\Egl_{L+m,L}$.
\end{dfn}
Note that $\interfin{N}$ is a non-semisimple Lie algebra and admits
$\ssp_{2L-2}$ spanned by $\Asp_{m,n}$,
$\Bsp_{m,n}$, and $\Csp_{m,n}$ as a Lie subalgebra.  The dimension of $\interfin{N}$ is
\begin{equation}
\dim \interfin{N} = \ffrac{(N-1)(N+2)}{2} -1,
\end{equation}
where we recall that $\Bsp$- and $\Csp$-blocks are symmetric matrices.
The Lie radical of $\interfin{N}$  is generated by $\Egl_{0,n}$, $\Egl_{0,L}$, $\Egl_{0,L+n}$, $\Egl_{m,L}$, and
$\Egl_{L+m,L}$: these are the  generators corresponding the the crosses  $\times$ in~\eqref{interfin-def}.

\medskip
We next show that the representation of this Lie algebra of fermionic bilinears in the periodic spin-chain generates the image of $\JTL{N}$, proving the following theorem.

\begin{Thm}\label{Thm:US-JTL-iso}
The image of the representation~\eqref{gl-stand-basis}-\eqref{sp-stand-basis} of the enveloping algebra $U \interfin{N}$ is equal to the image $\repgl\bigl(\JTL{N}\bigr)$ of the Jones--Temperley--Lieb algebra in the periodic $\gl(1|1)$ spin-chain.
\end{Thm}
\begin{proof}
We first prove that the centralizer of the $U \interfin{N}$ action is given by $\centJTL$ -- the centralizer of $\repgl\bigl(\JTL{N}\bigr)$. Recall that $\centJTL$ is described in Sec.~\ref{subsec:cent}. By a direct computation, we see that  the image of the representation~\eqref{gl-stand-basis}-\eqref{sp-stand-basis} of the enveloping algebra $U \interfin{N}$    commutes with the $\LQGodd$ action given in~\eqref{eq:EE-chi-eta}.
We then show that the generators of $\interfin{N}$ are the only bilinears in the $\chi$-$\eta$ fermions that commute with  $\LQGodd$.  Let us for this introduce bilinears as in~\eqref{sp-stand-basis} but with the opposite sign and denote them by $\Asp'_{m,n}$, $\Bsp'_{m,n}$, and $\Csp'_{m,n}$, respectively. Together with the bilinears  $\Egl_{n,0}$ and  $\Egl_{L,n}$, they belong  to the complement of the   vector space $\interfin{N}$. We then compute  commutators of all these bilinears which are not in $\interfin{N}$ with the generator $\EE 1$. It turns out that all these commutators are non-zero and linearly independent bilinears. So, there are no linear combinations among    $\Asp'_{m,n}$, $\Bsp'_{m,n}$,  $\Csp'_{m,n}$, $\Egl_{n,0}$, and  $\Egl_{L,n}$ that would commute with $\EE 1$. This proves the statement that $\interfin{N}$ is the maximum Lie subalgebra of bilinears (in the $\chi$-$\eta$  fermionic operators) which commute with $\LQGodd$.
We then note that the image of $\JTL{N}$ is generated by $e_j$'s which are also bilinears in the $\chi$-$\eta$ and they commute with $\LQGodd$. This image is thus contained in the image of $U \interfin{N}$ and both algebras have the common centralizer $\centJTL$ (it is easy to check that $\interfin{N}$ also commutes with the $\f^L$ and $\e^L$).

To finish
our proof, we need to show that the image of $U \interfin{N}$ is not bigger than  the image of $\JTL{N}$. To show this, we use a representation-theoretic approach. We recall the  subquotient structure for each $S^z=\pm j$ sector $\APrTL{j}$
considered as a module over the centralizer of $\centJTL$ which is
isomorphic by the definition to the algebra $\EndcJTL(\Hilb_N)$. The
centralizer obviously contains $\repgl\bigl(\JTL{N}\bigr)$ and, following the previous paragraph, the image of $U \interfin{N}$ as well. The opposite inclusion  is not true as was shown in our second paper~\cite[Sec.~5]{GRS2} and we thus can not rest on a double-centralizing argument.

\newcommand{\thmendoPPodd}{3.4.4~}
 The subquotient
structure can be obtained using intertwining operators respecting
$\centJTL$ action. These are described in Thm.~\thmendoPPodd in~\cite{GRS2}. The only
difference from the diagrams for $\JTL{N}$ in Fig.~\ref{FF-JTL-mod} is
that there are additional (`long') arrows mapping a top
subquotient~$\IrrTL{j}$ (having only outgoing arrows) to~$\IrrTL{k}$
in the socle (having only ingoing arrows) whenver $|j-k|\geq4$ is an
even number. We note these long arrows are not composites of any short
arrows mapping from the top to the middle level\footnote{the level  consisting of all nodes having  both ingoing and outgoing arrows.}, and from the middle
to the socle.
It turns out that $\interfin{N}$ generators correspond only to
these short arrows and not to the long ones, and therefore there is no
 element from $\interfin{N}$ represented by a long arrow. This can be
shown using a direct calculation with the fermionic expressions~\eqref{gl-stand-basis}-\eqref{sp-stand-basis}.
Due to  Lem.~\ref{lem:ssJTL-sp},  we need to analyze only the radical of $\interfin{N}$ generated by $\Egl_{0,n}$, $\Egl_{0,L}$, $\Egl_{0,L+n}$, $\Egl_{m,L}$, and
$\Egl_{L+m,L}$. These are represented by bilinears in the Clifford algebra $\Clif{4L}$ generators. As was noted before, $N-2$ creation modes
$\chi^{\dagger}_{p>0}$ and $\eta_{p>0}$ generate
the bottom level -- the intersection of the kernels of $\F$ and $\E$
in $\Hilb_N$ -- from the vacuum state $\vac$, and also the top level
from one cyclic vector $\lvac$ which is involved with $\vac$ into a
Jordan cell for the Hamiltonian $H$. Among the Clifford algebra generators, there are
two -- zero modes  $\eta_0$ and $\chi^{\dagger}_0$ -- proportional to
$\F$ and $\E\K^{-1}$, respectively. These are the only
generators mapping vectors from the top level to the middle level, and
from the middle to the bottom level. Among the generators  in the radical of $\interfin{N}$, there is only one, $\Egl_{0,L}$, which is the product of the two zero modes. The product maps the top
to the bottom but commutes
with the $\JTL{N}$ action
and   thus maps a top subquotient $\IrrTL{j}$ only to the bottom $\IrrTL{j}$.
 All other generators of the radical are bilinears in  fermions containing only
one of the zero modes and they thus map only by one level down. Therefore, they correspond to the short arrows in the diagram  in Fig.~\ref{FF-JTL-mod}.
We conclude that the action of $U\interfin{N}$ can
not correspond to  long arrows connecting the top and the bottom
and which are not composites of any short arrows because any element of the radical in the image of $U \interfin{N}$ is a linear combination of monomials in its generators.

 We can thus conclude that  the modules over
the enveloping algebra of $\interfin{N}$ in the spin-chain representation have the same subqutient structure as for the $\repgl\bigl(\JTL{N}\bigr)$ and therefore their Jacobson radicals are isomorphic. In addition to the analysis  on their semisimple parts given in Lem.~\ref{lem:ssJTL-sp} -- both algebras have the same equivalence classes of  simple modules -- we also conclude  that they are isomorphic as associative algebras. This finally proves the theorem.
\end{proof}

As a consequence of  Thm.~\ref{Thm:US-JTL-iso}, we get the following corollary which we consider as a non-semisimple version of the (symplectic) Howe duality.

\begin{Cor}\label{cor:Howe-nss}
The centralizer $\cent_{U \interfin{N}}$ of the universal enveloping algebra $U \interfin{N}$ in the periodic $\gl(1|1)$ spin-chain representation in~\eqref{gl-stand-basis}-\eqref{sp-stand-basis} is given by $\centJTL$, which is generated by the image of $\LQGodd$ in~\eqref{eq:EE-chi-eta} and $\f^L$ and $\e^L$.
\end{Cor}

Note that in this case the ``Howe-dual'' algebra is not described as the enveloping of a Lie algebra of bilinears,  in contrast to the anti-periodic case. It is rather a $\oZ_2$-graded algebra due to presence of the $\gl(1|1)$ subalgebra. It is easier to show for the open case, where the $\TL{N}$ is the enveloping algebra of a non-semisimple Lie algebra (having maximum semisimple Lie subalgebra $\ssp_{2L-2}$ and its radical is smaller than in the $\JTL{N}$ case) and the action of its centralizer $\LQGi$  is generated by the semidirect product $\ssp_2\ltimes \gl(1|1)$ of the Lie algebras $\ssp_2$ and $\gl(1|1)$.

\subsection{Normal ordered basis}\label{sec:nord}
We introduce here a few  formalities which will be useful to analyze  the $N\to\infty$ limits of our algebras. Recall that the excitations over the  space of  ground states
are generated by  the
creation modes $\chi^{\dagger}_{p}$ and $\eta_p$, with $p\in(0,\pi)$, while the annihilation modes are  $\chi_{p}$ and $\eta^{\dagger}_p$. We can thus introduce normal ordering prescriptions like $\nord{\chi_{p}\chi^{\dagger}_{p'}}\equiv-\chi^{\dagger}_{p'}\chi_{p}$ and  $\nord{\eta^{\dagger}_{p}}\eta_{p'}\equiv-\eta_{p'}\eta^{\dagger}_{p}$, {\it etc.} Then, a normal ordered basis in $\gl_{2L}\oplus\oC\one$ (a trivial central extension) is given by $\Egl'_{m,n}=\nord{\Egl_{m,n}}$, where the elementary matrices $\Egl_{m,n}$, for $0\leq m,n\leq 2L-1$, are introduced above in~\eqref{gl-stand-basis}. The normal ordering affects only  one half of the Cartan elements: $\Egl'_{n,n} =  \Egl_{n,n}$ while $\Egl'_{L+n,L+n} =  \Egl_{L+n,L+n} - \one$, where now $0\leq n \leq L-1$ and $\one$ is the identity in $\Endo\, \Hilb_{2L}$.
The defining relations in the normally ordered basis are slightly changed -- they have the central element term:
\begin{equation}\label{rel-gln-nord}
[\Egl'_{m,L+n},\Egl'_{L+k,l}] = \delta_{n,k}\Egl'_{m,l} - \delta_{m,l}\Egl'_{L+k,L+n} - \delta_{n,k}\delta_{m,l}\one,
\end{equation}
while all the other commutators have the same form as in the standard basis~\eqref{rel-gln-stand}.

We also consider central extensions of the  Lie subalgebras $\ssp_{2L-2}\subset\gl_{2L}$ and $\interfin{2L}\subset\gl_{2L}$ by the identity $\one$. So, we introduce the Lie algebras \begin{equation}\label{def-prime-alg}
\ssp'_{2L-2}=\ssp_{2L-2}\oplus\oC\one\qquad \text{and}\qquad \interfin{N}'=\interfin{N}\oplus\oC\one.
\end{equation}
Then, we choose the normal ordered basis in their Cartan subalgebras as
\begin{equation}\label{sp-nord-basis}
\Asp'_{n,n} \equiv \nord{\Asp_{n,n}} = \Egl'_{n,n} - \Egl'_{L+n,L+n} =  \Asp_{n,n} + \one,
\end{equation}
 see notations in~\eqref{sp-stand-basis}.
 Similarly to~\eqref{rel-gln-nord}, the commutators $[\Bsp_{m,n},\Csp_{k,l}]$ have the central element term in the normally ordered basis in $\ssp'_{2L-2}$ while other relations are not changed.

Note finally that the weights of highest-weight $\ssp'_{2L}$-modules are now different due to the different choice of Cartan elements (or simple roots). For example, the vacuum irreducible module $\AIrrTL{1}{1}$  has now  the weight $(0,0,\dots,0)$, the next one $\AIrrTL{2}{-1}$ corresponds to the weight $(1,0,\dots,0)$, {\it etc.}


\section{Scaling limit}\label{sec:scal-lim}

We are now interested in studying how the algebraic properties of the finite spin chain relate with those of its continuum -- or scaling --  limit. Our ultimate purpose is to extract from the $\gl(1|1)$ case lessons that can be used in other, more complicated situations. We shall for this purpose, introduce and discuss in the next section the new concept of interchiral algebra and its relationship with fields in a logarithmic conformal field theoretic setting. For now however, we do not discuss field theory \textit{per se}, and define and discuss the scaling limit of the spin chains using fermionic modes.

In this section, we consider mostly the periodic model and discuss briefly the anti-periodic case only at the end in Sec.~\ref{sec:scal-lim-tw}.
The limit of $\JTL{N}$ algebras is constructed in Sec.~\ref{sec:scaling} and~\ref{commrels}. We discuss  in Sec.~\ref{sec:bimod-scal-lim} the structure of the bimodule over the two commuting algebras, $\LQGodd$ and the scaling limit of $\JTL{N}$ or $U\interfin{N}$ algebras denoted by $U\interinfin$.
In Sec.~\ref{sec:interinfin-sm}, we also describe  simple  modules over $U\interinfin$ that appear in the space $\Hilb$ of scaling states.
Finally, the full generating functions of energy levels are computed in Sec.~\ref{sec:gen-func}.

\newcommand{\GRSscallim}{4.3~}

\subsection{The full scaling limit of the closed $\gl(1|1)$ spin-chains}\label{sec:scaling}
In this section, we recall the construction of the scaling limit~\cite{GRS1} of
 the (anti)periodic spin-chain that gives a Conformal Field Theory model. We then use  the  Lie algebraic reformulation of the (images of) $\JTL{N}$ and $\JTL{N}^{tw}$  to construct the full scaling limit of these algebras. An essential ingredient  in the general definition
 of the scaling limit   is
 the low-lying eigenstates of the Hamiltonian~$H$. In order to study the
 action of JTL elements on these eigenstates in the limit $L\to\infty$
 (recall $N=2L$) we first truncate each $\Hilb_{2L}$, keeping only
 eigenspaces up to an energy level $M$, for each positive number $M$.
 Each such truncated space turns out to be finite-dimensional in the
 limit, {\it i.e.}, it depends on $M$ but not $L$. Then, keeping matrix
 elements of JTL elements that correspond to the action only within
 these truncated spaces of ``scaling'' states, we obtain well-defined
 operators in the limit $L\to\infty$. The corresponding operators
 acting on all scaling states of the CFT can be finally obtained (if
 they exist) in the second limit $M\to\infty$.
A bit more formally, \textit{the scaling limit} denoted
simply by `$\mapsto$' is defined as a limit over graded spaces of coinvariants with
respect to smaller and smaller subalgebras in the creation modes
algebra, see details in  Sec.~\GRSscallim from~\cite{GRS1}. Meanwhile, in the case of our $\gl(1|1)$ spin-chains we can actually give a  clearer definition of the scaling limit (of algebras and their modules) by means of inductive systems, see App.~\AppLim.

\subsubsection{The Clifford algebra scaling limit}\label{Cliff-def}
We first recall~\cite{GRS1} the scaling limit of the Clifford algebra $\Clif{4L}$ generators introduced in Sec.~\ref{sec:ham-spec} -- the $\chi$ and $\eta$
fermions --
and then define a representation of $\glinf$ which contains the full scaling limit of JTL.
The lattice fermions appropriately rescaled coincides in the scaling limit with the symplectic-fermions modes:
\begin{eqnarray}\label{eq:sferm-lat-def}
\begin{split}
\sqrt{m}\,\chi_p\mapsto\sfermm_m,\qquad\sqrt{m}
\,\bar{\eta}^\dagger_p\mapsto\sfermp_m,\qquad\sqrt{m}\,\bar{\chi}_p\mapsto\bsfermm_m,
\qquad-\sqrt{m}\,\eta^\dagger_p\mapsto\bsfermp_m,\\
-\sqrt{m}\,\bar{\eta}_p\mapsto\sfermm_{-m}, \qquad\sqrt{m}\,\chi^\dagger_p\mapsto\sfermp_{-m},
\qquad\sqrt{m}\,\eta_p\mapsto\bsfermm_{-m}, \qquad\sqrt{m}\,\bar{\chi}^\dagger_p\mapsto\bsfermp_{-m},
\end{split}
\end{eqnarray}
where we consider any finite integer $m>0$ and set $p=m\pi/L$ and  $\bbeta^{(\dagger)}_p=\eta^{(\dagger)}_{\pi-p}$, and
$\bbchi^{(\dagger)}_p=\chi^{(\dagger)}_{\pi-p}$, while  the scaling limit for zero modes is
\begin{eqnarray}\label{eq:sferm-zero-def}
\begin{split}
\sqrt{L/\pi}\,\eta_0&\mapsto \sfermm_0=\bsfermm_0,
&\qquad
 \sqrt{L/\pi}\,\chi_0^{\dagger}&\mapsto\sfermp_0=\bsfermp_0,\\
\sqrt{\pi/L}\,\eta^{\dagger}_{0}&\mapsto i\phip_0,& \qquad
\sqrt{\pi/L}\,\chi_{0}&\mapsto -i\phim_0.
\end{split}
\end{eqnarray}
Indeed, having this limit we obtain
 that the scaling limit of the Hamiltonian $H=\sum_{j=1}^N e_j$ gives the zero mode of the stress energy tensor $T(z)+\bar{T}(\bz)$ in the symplectic fermions theory:
\begin{equation}\label{L0-def}
\ffrac{L}{2\pi}\Bigl(H+\ffrac{4L}{\pi}\Bigr)\mapsto  \sum_{m\in\oZ}
\nord{\sfermp_{-m}\sfermm_m}+
\bigl(\psi\to\bar{\psi}\bigr)-\ffrac{c}{12} = L_0+\bar{L}_0-\ffrac{c}{12},
\end{equation}
where we use the  fermionic normal ordering prescription introduced in Sec.~\ref{sec:nord} which takes the form~\cite{Kausch}
\begin{equation}
\nord{\sferm^{\alpha}_{m}\sferm^{\beta}_n}
=\begin{cases}
\sferm^{\alpha}_{m}\sferm^{\beta}_n&\textit{for}\;\; m<0,\\
-\sferm^{\beta}_n\sferm^{\alpha}_{m}&\textit{for}\;\; m>0.
\end{cases}
\end{equation}

We note that the fermionic operators act on the  space $\Hilb$ of scaling states, which are limits of low-energy eigenvectors\footnote{It is better actually to say `root-vectors' because the Hamiltonian is not diagonalizable on $\Hilb_N$ and we have a basis where  relations such as $(H-\lambda\one)^2 v = 0$, for $v\in\Hilb_{N}$, are satisfied.} of the Hamiltonian $H$, in the sense of construction in App.~\AppLim. The ground subspace in $\Hilb$ has the same structure~\eqref{diag:space-gr-st} as in the finite chain $\Hilb_{N}$, where one should replace $\chi_0^{\dagger}$ and $\eta_0$
by the zero modes $\sfermp_0$ and $\sfermm_0$, respectively.
 The module structure (over the infinite-dimensional Clifford algebra) on $\Hilb$  containing all the excitations over the ground states
$\lvac$, $\sfermp_0\lvac$, $\sfermm_0\lvac$ and $\sfermm_0\sfermp_0\lvac=\vac$ is thus generated by the free action of the
creation modes $\sferm^{1,2}_n$ and $\bsferm^{1,2}_m$, for integers $n,m<0$.
The annihilation modes, or positive modes, act on the vacuum states in the usual way.

We have also shown in our first paper~\cite{GRS1} that the limit~\eqref{eq:sferm-lat-def} allows one to obtain all left and right Virasoro modes $L_n$ and $\bar{L}_n$ in the symplectic fermions representation
\begin{eqnarray}\label{Ln-def}
L_n=\sum_{m\in\oZ} \nord{\sfermp_{n-m}\sfermm_m},\qquad
\bar{L}_n=\sum_{m\in\oZ} \nord{\bsfermp_{n-m}\bsfermm_m},\qquad n\in\oZ,
\end{eqnarray}
 as the scaling limit of particular JTL elements, see also more details below in  Sec.~\ref{sec:interch-modes}. A comment is necessary about the infinite sums in the definition of $L_n$'s and $\bar{L}_n$'s. Note that the Hilbert space $\Hilb$ has bi-grading by the pair $(L_0,\bar{L}_0)$, where each homogeneous or root subspace $\Hilb^{n,m}$ for positive $n$ and $m$ is spanned by states of the form $\sferm^{1,2}_{n_1}\dots\sferm^{1,2}_{n_k}\,\bsferm^{1,2}_{m_1}\dots\bsferm^{1,2}_{m_l}$ acting on any ground state such that $\sum_{i=1}^k n_i = -n$ and $\sum_{j=1}^{l}m_l = -m$. The space $\Hilb$ then can be described as the direct sum $\oplus_{n,m}\Hilb^{n,m}$ which means that any state $v\in\Hilb$ can be written as $v=\sum_{n,m}v^{(n,m)}$, with $v^{(n,m)}\in\Hilb^{n,m}$, and there exist positive integers $N$ and $M$ such that $v^{(n,m)}=0$ for any $n>N$ or $m>M$. This is what we call \textit{the finite-energy and finite-spin} (or simply \textit{scaling}) \textit{states space} and it is
 a dense\footnote{See the discussion of completions and topology in Sec.~\ref{sec:interch-interinfin} after~\eqref{eq:cHilb}.} subspace in the the direct product of the root subspaces $\Hilb^{n,m}$.
 Therefore, the formally infinite sums in~\eqref{Ln-def}  actually  reduce to finite sums in the scaling states space $\Hilb$.
  This is standard, and discussed for instance in~\cite{Ottesen}.

\subsubsection{The symplectic fermion representation of $\glinf$}\label{sec:spinf}
We recall that the bilinears in the $\chi$-$\eta$ fermions give a standard basis $\Egl_{m,n}$, the usual elementary matrices, in the Lie algebra
$\gl_{2L}$, see Sec.~\ref{Howe-ss} for details. Due to the fact that low-lying eigenstates are generated by the Clifford generators with their momenta $p$ close to $0$ or $\pi$ (in the large $L$ limit), we keep in the scaling limit all bilinears in the fermions with any of these momenta.
Note that basis elements~\eqref{gl-stand-basis} of $\gl_{2L}$ can be divided into $4$ blocks with the $(m,n)$ indices running from $0$ to $L-1$. Because in our limit we keep only those indices close to $0$ or to $L$, then each block is further divided into $4$ infinite blocks in the scaling limit.
 This  limit of
$\gl_{2L}$ can be schematically described in terms of its standard generators as
\begin{equation}
\Bigl[\; \Egl_{m,n} =  \chid_p\,\chi_q\;\Bigr]_{0\leq p,q\leq L-1} \qquad \xrightarrow{\;L\to\infty\;} \qquad
 \begin{matrix}
 \chid_0\,\chi_0 & \chid_0\chi_{\step} & \ldots & \chid_{0}\,\bbchi_{2\step} & \chid_{0}\,\bbchi_{\step} \\
  \chid_{\step}\,\chi_{0} &\chid_{\step}\,\chi_{\step}   &  \ldots  &   \chid_{\step}\,\bbchi_{2\step} & \chid_{\step}\,\bbchi_{\step}   \\
 \vdots & \vdots & \ddots & \vdots & \vdots\\
   \bbchid_{2\step}\,\chi_0  & \bbchid_{2\step}\,\chi_{\step}  &  \ldots    &   \bbchid_{2\step}\,\bbchi_{2\step}   &  \bbchid_{2\step}\,\bbchi_{\step}   &\\
  \bbchid_{\step}\,\chi_0  & \bbchid_{\step}\,\chi_{\step}  &  \ldots    &   \bbchid_{\step}\,\bbchi_{2\step}   &  \bbchid_{\step}\,\bbchi_{\step}   &
\end{matrix}
\end{equation}
 for the left-top block in the matrix algebra, and similarly for other three blocks.
Here, we set for the momenta variables $p=m\step$, $q=n\step$, with $\step=\pi/L$, and
these momenta on the right-hand side are formal variables -- one should of course make substitutions~\eqref{eq:sferm-lat-def} and~\eqref{eq:sferm-zero-def}.
We emphasize that each item of the  block on the right is an infinite elementary matrix having identity at the corresponding position and zeros otherwise. These elementary matrices (for all four blocks) define a representation of an infinite-dimensional Lie algebra $\glinf$ on the space $\Hilb$ which we call the symplectic-fermion representation of~$\glinf$.

  A comment is necessary about the exact definition of the algebra $\glinf$. Any element of  $\glinf$ is an infinite matrix  with a  \textit{finite} number of non-zero elements,
 {\it i.e.}, it is represented on $\Hilb$ by a finite sum of the generators --  the elementary matrices. The  commutation relations in this algebra are  given by corresponding limits of~\eqref{rel-gln-stand} and they correspond to usual basis and relations~\cite{Kac-book} in $\glinf$ after appropriate rearranging  rows and columns,  see more precise statements in App.~\Appglinf.

  We note that formally the Virasoro generators $L_n$'s and $\bar{L}_n$'s from~\eqref{Ln-def} do not belong to such defined $\glinf$ (recall there exist many versions of $\gl_{\infty}$ algebras~\cite{Kac-book}) but belong, for $n\ne0$, to its completed version $\bglinf$, where there is now
 a possibly infinite number of non-zero elements, but the matrix still has a finite number of non zero diagonals, see details in App.~\Appglinf.
The only problem is with $L_0$ and $\bar{L}_0$, because of  the well-known central anomaly that appears in CFT. In order to get a convenient algebra containing the $L_0$ and $\bar{L}_0$ operators we should actually take the scaling limit of the family of algebras $\{\gl_{2L}\}$ in their normally ordered basis $\Egl'_{m,n}$ introduced in Sec.~\ref{sec:nord}.
 Using then~\eqref{eq:sferm-lat-def}-\eqref{eq:sferm-zero-def},
 we  obtain finally
  the generators (after rescaling by $\pm\sqrt{mn}$) of another infinite-dimensional Lie algebra we call $\glinfc$
  given by the following list of bilinears
\begin{multline}\label{gl-inf-2}
\glinfc =\langle\nord{\sfermp_{m}\sfermm_{n}},\quad \nord{\bsfermp_{m}\sfermm_{n}},\quad \nord{\sfermp_{m}\bsfermm_{n}},\quad \nord{\bsfermp_{m}\bsfermm_{n}},
\quad \nord{\sfermp_{m}\phim_{0}},\quad \nord{\bsfermp_{m}\phim_{0}},\\
  \nord{\phip_0\sfermm_{n}},\quad  \nord{\phip_0\bsfermm_{n}},\quad
 \nord{\phip_{0}\phim_{0}},\quad m,n\in\oZ\rangle,
\end{multline}
which are now all normally ordered (note that some of these bilinears, $\sfermm_{-n}\sfermp_{n}$, are not elementary matrices.) This Lie algebra turns out to be a central extension of the $\glinf$, {\it i.e.}, it is $\glinfc=\glinf\oplus \oC\one$, as a vector space. The central element $\one$ is obtained by the scaling limit of the identity in the (trivial) central extension $\gl_{2L}\oplus\oC\one$ considered in Sec.~\ref{sec:nord}, see also a formal inductive limit construction in  App.~\AppLim. Commutation relations in this algebra are obtained as limits of the relation~\eqref{rel-gln-nord} which fixes the central charge of $\glinfc$.  The  contribution of the central term in commutators of the new generators can be also obtained using the two-cocycle~\cite{Kac-book} $c(\Egl'_{m,n},\Egl'_{m',n'})$ and it coincides with the vacuum expectation of the commutator $[\Egl'_{m,n},\Egl'_{m',n'}]$ as it should. Now, all the left-right Virasoro operators on $\Hilb$ are embedded  into the completed Lie algebra $\glinfco$ including  $L_0$ and $\bar{L}_0$. We will actually work below with this Lie algebra $\glinfco$ containing the Virasoro generators and not with~$\glinf$.

\subsubsection{The scaling limit of the JTL algebras, $\spinf$, and $\interinfin$}\label{sec:interinfin}

The scaling limit of the (spin-chain representations of) $\JTL{N}$ algebras can then be taken using the Lie algebraic description given in Thm.~\ref{Thm:US-JTL-iso}. It was shown there that the spin-chain representation  of $\JTL{N}$ is given by a representation of the enveloping algebra $U \interfin{N}$.  We thus need only to take the scaling limit of the central extension $\interfin{N}'=\interfin{N}\oplus\oC\one$ introduced in Sec.~\ref{sec:nord} which is straightforward since there is an embedding of $\interfin{N}$ into the $\gl_{2L}$, see Dfn.~\ref{interfin-def}.  In particular, the  Lie subalgebra $\ssp'_{2L-2}\subset\interfin{N}'$ from~\eqref{def-prime-alg} generated by normally ordered elements in~\eqref{sp-nord-basis} and~\eqref{sp-stand-basis} has the limit generated by the following quadratic expressions
\begin{equation}\label{spinf-gen}
J_{\alpha\beta}\nord{\sferm^{\alpha}_{m}\sferm^{\beta}_{n}},\quad
J_{\alpha\beta}\nord{\bsferm^{\alpha}_{m}\bsferm^{\beta}_{n}},\quad
 S_{\alpha\beta} \psi^{\alpha}_m \bar{\psi}^{\beta}_n,\qquad m,n\in\oZ/\{0\},
\end{equation}
where $J_{\alpha\beta}$ is such that $J_{\alpha\beta}J^{\beta\gamma}=\delta^{\gamma}_{\alpha}$ and we use  the symplectic form $J^{12}=-J^{21}=1$, and $S_{\alpha\beta}$ is the symmetric form $S_{12}=S_{21}=1$, $S_{11}=S_{22}=0$. We  show below that the vector space spanned by finite linear combinations of these operators is closed under taking commutators.
The infinite-dimensional Lie algebra they generate coincides with the central extension of the Lie algebra $\spinf\subset\glinf$.

Further, we introduce  $\enrg_{m,n}$, $\strt_{m,n}$, and $\bstrt_{m,n}$ as the
operators
\begin{equation}\label{prop:sympl-ferm-interLie}
\enrg_{m,n} =S_{\alpha\beta} \psi^{\alpha}_m \bar{\psi}^{\beta}_n,\quad
\strt_{m,n}=J_{\alpha\beta}{\sferm^{\alpha}_{m}\sferm^{\beta}_{n}},\quad \bstrt_{m,n}=J_{\alpha\beta}{\bsferm^{\alpha}_{m}\bsferm^{\beta}_{n}}, \qquad\text{with}\quad n,m\in\oZ,
\end{equation}
{\it i.e.}, as in~\eqref{spinf-gen} but without normal ordering and including combinations involving the zero fermionic modes $\sfermm_0$ and $\sfermp_0$ (note that there are no bilinears here involving conjugate modes $\phip_0$ and $\phim_0$). Actually, we have the symmetry $\strt_{m,n}=\strt_{n,m}$ and $\bstrt_{m,n}=\bstrt_{n,m}$. The operators in~\eqref{prop:sympl-ferm-interLie} together with the central term $\one$ give now generators of the scaling limit of $\interfin{N}'$ which we denote simply by $\interinfin\subset\glinfc$.  We conclude that  the corresponding action of $U\interinfin$ is  \textit{the scaling limit} of the (spin-chain images of the) JTL algebras. We give in App.~\AppLim~a more formal construction of this scaling-limit algebra as a direct/inductive limit $\varinjlim_L \repgl\bigl(\JTL{2L}\bigr)$ of the finite-dimensional spin-chain representations $\repgl$ of the JTL algebras.

\subsection{Commutation relations in $\interinfin$}\label{commrels}
 We next obtain commutation relations
between the generators  $\enrg_{n,m}$,
 $\strt_{m,n}$ and $\bstrt_{m,n}$ of $\interinfin$.
Using a straightforward calculation and identities like $J_{\alpha\beta}J_{\gamma\delta}J^{\beta\gamma}  = J_{\alpha\delta}$, we first obtain
\begin{multline}\label{interch-comm-rel-TT}
[\strt_{m,n},\strt_{k,l}]=[J_{\alpha\beta}\sferm^{\alpha}_{m}\sferm^{\beta}_{n}, J_{\gamma\delta}\sferm^{\gamma}_{k}\sferm^{\delta}_{l}] \\
= J_{\alpha\beta} n \bigl(\delta_{n+k,0}\, \psi^{\alpha}_m \psi^{\beta}_l +
\delta_{n+l,0} \,\psi^{\alpha}_m \psi^{\beta}_k  \bigr)
+ J_{\alpha\beta} m \bigl(\delta_{m+k,0}\, \psi^{\alpha}_l \psi^{\beta}_n +
\delta_{m+l,0} \,\psi^{\alpha}_k \psi^{\beta}_n  \bigr)
\end{multline}
and a similar expression with $\sferm^{\alpha}_n\to\bsferm^{\alpha}_n$.
We then have
\begin{align}
[\enrg_{m,n} , \strt_{k,l} ] &= S_{\alpha\beta}\, m \bigl(\delta_{k+m,0}\, \sferm^{\alpha}_l \bsferm^{\beta}_n +
\delta_{m+l,0} \,\sferm^{\alpha}_k \bsferm^{\beta}_n  \bigr),\label{interch-comm-rel-1}\\
[\enrg_{m,n} , \bstrt_{k,l} ] &= S_{\alpha\beta}\, n \bigl(\delta_{k+n,0}\, \sferm^{\alpha}_m \bsferm^{\beta}_l +
\delta_{n+l,0} \,\sferm^{\alpha}_m \bsferm^{\beta}_k  \bigr),\label{interch-comm-rel-2}
\end{align}
and the commutators
\begin{equation}\label{interch-comm-rel}
[\enrg_{n,m} , \enrg_{k,l} ] = J_{\alpha\beta}\bigl(l\,\delta_{l+m,0}\, \psi^{\alpha}_n \psi^{\beta}_k +
k\,\delta_{k+n,0} \,\bar{\psi}^{\alpha}_m \bar{\psi}^{\beta}_l  \bigr) - 2kl\,\delta_{k+n,0}\delta_{l+m,0},
\end{equation}
where we use repeatedly the identities $S_{\alpha\beta}J_{\gamma\delta}J^{\alpha\gamma}  = S_{\beta\delta}$ and
$S_{\alpha\beta}J_{\gamma\delta}J^{\alpha\delta}  = -S_{\beta\gamma}$.

Using the definition of $\enrg_{n,m}$, $\strt_{m,n}$ and $\bstrt_{m,n}$ given in~\eqref{prop:sympl-ferm-interLie}, we see that the relations~\eqref{interch-comm-rel-TT} and~\eqref{interch-comm-rel} together with~\eqref{interch-comm-rel-1} and~\eqref{interch-comm-rel-2} prove the following proposition which can be considered as an alternative definition of the $\interinfin$ Lie algebra.
\begin{prop}\label{prop:interch-rel}
The  Lie algebra $\interinfin$ has $\enrg_{n,m}$, $\strt_{k,l}$ and $\bstrt_{r,s}$, with $n,m,k,l,r,s\in\oZ$, as its basis elements,  with the defining relations
\begin{align}
[\enrg_{n,m} , \enrg_{k,l} ] &= l\delta_{l+m,0}\, \strt_{n,k} +
k\,\delta_{k+n,0} \bstrt_{m,l}  - 2kl\,\delta_{k+n,0}\delta_{l+m,0},\\
[\enrg_{m,n} , \strt_{k,l} ] &= m \bigl(\delta_{k+m,0}\, \enrg_{l,n} +
\delta_{m+l,0} \, \enrg_{k,n}  \bigr),\\
[\enrg_{m,n} , \bstrt_{k,l} ] &
= n \bigl(\delta_{k+n,0}\, \enrg_{m,l} +
\delta_{n+l,0} \, \enrg_{m,k}  \bigr),\\
[\strt_{m,n},\strt_{k,l}] &= n \bigl(\delta_{n+k,0}\, \strt_{m,l}+
\delta_{n+l,0} \,\strt_{m,k} \bigr) + m \bigl(\delta_{m+k,0}\, \strt_{l,n} +
\delta_{m+l,0} \,\strt_{k,n}  \bigr),\\
[\bstrt_{m,n},\bstrt_{k,l}] &= n \bigl(\delta_{n+k,0}\, \bstrt_{m,l}+
\delta_{n+l,0} \,\bstrt_{m,k} \bigr) + m \bigl(\delta_{m+k,0}\, \bstrt_{l,n} +
\delta_{m+l,0} \,\bstrt_{k,n}  \bigr),\\
[\strt_{m,n},\bstrt_{k,l}] &= 0.
\end{align}
\end{prop}

We note that it is straightforward to check the Jacobi identitites.

\subsection{Bimodule in the scaling limit and the centralizer of $\interinfin$}\label{sec:bimod-scal-lim}

Using a formal construction of  inverse/projective limits, we define  \textit{the scaling limit of the JTL centralizers} $\centJTL$ in App.~\AppLim. The limit is an infinite dimensional associative algebra which we identify with a quotient of $\LQGodd$ and we denote this quotient by $\rho\bigl(\LQGodd\bigr)$, see~\eqref{app:inv-lim-U}.
Fermionic expressions for the generators in the scaling limit of the centralizers were  computed in our first paper~\cite{GRS1}:
\begin{equation}\label{scal-lim-LQGodd}
  \EE n =
\left[\sum_{m>0}
\Bigl(\ffrac{\sfermp_{m}\sfermp_{-m}}{m} -
\ffrac{\bsfermp_{m}\bsfermp_{-m}}{m}\Bigr)\right]^n\psi_0^2,\qquad
\FF n =
\left[\sum_{m>0}
\Bigl(\ffrac{\sfermm_{m}\sfermm_{-m}}{m} -
\ffrac{\bsfermm_{m}\bsfermm_{-m}}{m}\Bigr)\right]^n\psi_0^1,
\end{equation}
with the representation of the Cartan element $\h$ as
\begin{equation}\label{scal-lim-h}
\h =- i/2 \bigl(\sfermp_0\phim_0 + \sfermm_0\phip_0\bigr) + \sum_{m>0}
\ffrac{1}{m}\bigl(\sfermp_{-m}\sfermm_{m} + \sfermm_{-m}\sfermp_{m} + \bsfermp_{-m}\bsfermm_{m} + \bsfermm_{-m}\bsfermp_{m}\bigr)
\end{equation}
while the generator $\K = (-1)^{2\h}$.

It is obvious that both $\glinfc$ and its subalgebra $\interinfin$ act on the space $\Hilb$ of scaling states.
To study the symmetry algebra of
the Lie algebra $\interinfin$ on this space, we first note that
using   identities like
\begin{equation}
\Bigl[\sfermp_m,\sfermm_{|m|}\sfermm_{-|m|}\Bigr]=|m|\sfermm_{m},\qquad
\Bigl[\sfermm_m,\sfermp_{|m|}\sfermp_{-|m|}\Bigr]=-|m|\sfermp_{m},\qquad
m\in\oZ,
\end{equation}
it is straightforward to check that the generators~\eqref{prop:sympl-ferm-interLie} commute with the $\LQGodd$ action given in~\eqref{scal-lim-LQGodd} and~\eqref{scal-lim-h}.
Moreover, we prove in Thm.~\ref{thm:centr-interLie} that the centralizer of the direct-limit algebra $\interinfin$ equals the inverse limit of the centralizers for each term $\interfin{N}$. In view of the importance of this fact, we repeat it here:
\textit{the centralizer $\cent_{U\interinfin}$ of the enveloping algebra $U\interinfin$ action~\eqref{prop:sympl-ferm-interLie} on $\Hilb$
 is given by the representation of $\LQGodd$ in~\eqref{scal-lim-LQGodd} and~\eqref{scal-lim-h}.}

\medskip

We next describe  bimodule structure of $\Hilb$ over $U\interinfin$ and its centralizer  $\rho\bigl(\LQGodd\bigr)$.


As a consequence of the direct limit construction of $\interinfin$ and of its module $\Hilb$ in Sec.~\ref{app:JTL-lim}, we obtain that the (direct) limit of irreducible representations over $\interfin{N}$ is also irreducible with respect to the action of $\interinfin$. This is proved in Prop.~\ref{prop:lim-simples} and it in particular means that the only simple $\interinfin$-modules that appear in the direct limit space $\Hilb$ are the limits $\interM{j}$.
Further, see a note after Prop.~\ref{prop:lim-simples}, the indecomposable but reducible $\JTL{N}$-modules $\APrTL{j}$ that appear in $\Hilb_N$ as direct summands go in the scaling (or direct) limit to
indecomposable but reducible $U\interinfin$-modules (so, $\APrTL{j}$ do not split on direct sums in the limit) and their subquotient structure has the same pattern as in Fig.~\ref{FF-JTL-mod}, or more precisely in  its  infinite analogue in Fig.~\ref{ind-chain-mod-fig} in Sec.~\ref{sec:indecomp-Smod}.
All this is not surprising as we have essentially the same centralizer $\LQGodd$ in CFT as for any $\JTL{N}$ algebra in the periodic spin-chain,
 see Thm.~\ref{thm:centr-interLie} about the centralizing algebra for $U\interinfin$.
 The only difference from finite chains bimodules described in Sec.~\ref{ind-chain-bimod-subsec} is that the representations of $\LQGodd$ now admit the ``total spin'' (the $\h$ eigenvalues) of any integer value:
 we have a decomposition of $\Hilb$ as a module over $\LQGodd$ onto the indecomposable direct summands $\PPodd{j}$, for any $j>1$, and  with multiplicities given now by graded vector spaces which are the simple $\interM{j}$ modules over $U \interinfin$.
Then, in our space of scaling states where each state is a finite linear combination of basis ones, we can easily extend our analysis from~\cite{GRS2} and obtain the bimodule structure over the pair of commuting algebras $\bigl(U\interinfin,\LQGodd\bigr)$ as in Fig.~\ref{JTL-Uqodd-bimod}.
In the figure for the bimodule in the scaling limit, each node is now a simple $\interinfin$-module $\interM{j}$ and the towers have infinite length.
  The first question is of course about the left-right Virasoro content of these scaling limits, which we discuss below.


\subsection{The content of simple $\interinfin$-modules}\label{sec:interinfin-sm}
We start our analysis by discussing the left-right Virasoro $\VirN(2)=\Vir(2)\boxtimes\overline{\Vir}(2)$ content of the
simple $\JTL{N}$ modules in the scaling limit. By this, we mean more precisely  the Virasoro
representation content of the states that contribute to the direct
limit $\interM{j}$ of the simple $\JTL{N}$-modules $\AIrrTL{j}{(-1)^{j+1}}$ (recall that this limit is simple as well, as a $\interinfin$-module.)
 It is convenient for this purpose to
first evaluate the generating functions~\eqref{charform} of energy-momentum levels in  the
$\JTL{N}$ modules $\APrTL{j}$ (the direct summands of the spin-chain) we have encountered previously. This can be done by reorganizing the fermions such that $\psi^{1,2}$'s correspond to even modes of new fermions and  $\bar{\psi}^{1,2}$'s give odd modes of the fermions, see App.~\Appglinf. Then, we can produce a generating function $F_{j,(-1)^{j+1}}$ for each charged sector $\APrTL{j}$ with $S^z=j$ using the new fermions. Repeating an   exercise similar to what can be found in~\cite{Kac-book}, we obtain the character formula for $F_{j,(-1)^{j+1}}$ in the limit which reads
%
  \begin{equation}
   F_{j,(-1)^{j+1}}=\ffrac{q^{1/12}\bar{q}^{1/12}}{P(q)P(\bar{q})}
     \sum_{n\in \oZ}q^{\frac{(j+n)^2+j+n}{2}}
    \bar{q}^{\frac{n^2+n}{2}}
     \end{equation}
with
\begin{equation}\label{Pq}
\displaystyle P(q) =  \prod_{n=1}^{\infty} (1 - q^n) = q^{-1/24} \eta (q).
\end{equation}
We present a  more general derivation of this  result in App.~\AppChar~which is based on well-know scaling properties of the spectrum for twisted XXZ models, and will be useful in our future studies~\cite{GRS4}.
From the structure of the spin chain modules (see
Sec.~\ref{ind-chain-bimod-subsec} and Fig.~\ref{FF-JTL-mod}), we
deduce then the traces $F_{j,(-1)^{j+1}}^{(0)}$ for the simple $\interinfin$-modules
$\interM{j}$
which read, in terms of Virasoro characters,  (see a calculation in App.~\AppChar)
  \begin{equation}\label{Fs-Vir-content-gen}
  F_{j,(-1)^{j+1}}^{(0)}=\sum_{j_1,j_2> 0}^* \chi_{j_1,1}\overline{\chi}_{j_2,1}
  \end{equation}
  where the sum is done with the following constraints:
  \begin{eqnarray}\label{cond-sum-star}
  |j_1-j_2|+1\leq j,\qquad 
  j_1+j_2-1\geq j, \qquad 
  j_1+j_2 - j = 1\; \text{mod} \; 2
  \end{eqnarray}
 (note this is equivalent to treating $j$ not as a spin but as a degeneracy, {\it i.e.}, setting $j=2s+1$, $j_i=2s_i+1$ and combining $s_1$ and $s_2$ to obtain spin $s$), and we recall  the characters of the simple modules with Kac labels $(j,1)$ over Virasoro at $c=-2$:
\begin{equation}
\chi_{j,1}={q^{(2j-1)^2/8}-q^{(2j+1)^2/8}\over\eta(q)}.
\end{equation}

For instance, we have
 \begin{eqnarray}\label{Fs-Vir-content}
 \begin{split}
 &F_{0,-1}^{(0)}&=&\quad0,\\
  &F_{1,1}^{(0)}&=&\quad\sum_{r=1}^\infty \chi_{r1}\bar{\chi}_{r1},\\
   &F_{2,-1}^{(0)}&=&\quad\sum_{r=1}^\infty \chi_{r1}\left(\bar{\chi}_{r-1,1}+
   \bar{\chi}_{r+1,1}\right),\\
  &F_{3,1}^{(0)}&=&\quad\chi_{11}\bar{\chi}_{31} + \chi_{21}\left(\bar{\chi}_{21}+\bar{\chi}_{41}\right)
  +\sum_{r=3}^\infty \chi_{r1}\left(\bar{\chi}_{r-2,1}+\bar{\chi}_{r1}+
   \bar{\chi}_{r+2,1}\right).
   \end{split}
   \end{eqnarray}
   %

The result~\eqref{Fs-Vir-content-gen} implies that, in contrast with the open
case, simple modules of the $\JTL{N}$ algebra do not become simple
modules over $\VirN(2)$ in the scaling limit. While in general one
might expect that this scaling limit gives rise to non-fully-reducible
modules over $\VirN(2)$, it turns out that one gets in fact a direct
sum of simple modules over the full Virasoro algebra
$\VirN(2)$. This can be checked either by explicitly working out this
\scal limit as  was done in~\cite{GRS1}, or by using a direct
argument. We first note that the left Virasoro generators are
particular endomorphisms of the module structure over the right
Virasoro $\bar{\Vir}(2)$ action. Then, taking into account possible extensions/`glueings' between
 simple $\Vir(2)$-modules -- a module $\VX_{j,1}$ can be extended only
by $\VX_{j\pm1,1}$ into an indecomposable module, -- the left Virasoro
algebra can map vectors from subquotients $\VX_{j,1}\boxtimes\bar{\VX}_{k,1}$
only to $\VX_{j\pm1,1}\boxtimes\bar{\VX}_{k,1}$,  and similarly for the right Virasoro algebra.  Therefore, the
character of a reducible but indecomposable $\VirN(2)$-module has to have at least
one pair of the indices $(j,k)$, say $(j_1,k_1)$ and $(j_2,k_2)$, such that $|j_1-j_2|=1$ and $k_1=k_2$ or $|k_1-k_2|=1$ and $j_1=j_2$. None of these conditions
  is true for the functions $F_{n,(-1)^{n-1}}^{(0)}$ given above. We
thus conclude that $\VirN(2)$-modules with characters
in~\eqref{Fs-Vir-content-gen}, for any fixed $j$, are semisimple. The
foregoing analysis gives thus the left-right Virasoro content for the direct limits of simple
$\JTL{N}$-modules\footnote{
Having the identification of the Hamiltonians $H=H_0(0)$
with the Cartan-subalgebra elements of $\ssp_{2L-2}$, see~\eqref{hamilt-sp} below,  it should be possible using the Weyl character formula for fundamental
  representations of $\ssp_{2L-2}$ to study the character asymptotics and
extract the $\VirN(2)$ content of the limits as in~\eqref{JTL-simple-Vir-sum}. We leave this exercise for  future work.
}:
\begin{equation}\label{JTL-simple-Vir-sum}
\AIrrTL{j}{(-1)^{j+1}} \mapsto \interM{j}|_{\VirN(2)}=\bigoplus_{j_1,j_2> 0}^* \VX_{j_1,1}\boxtimes\bar{\VX}_{j_2,1},
\end{equation}
with the conditions~\eqref{cond-sum-star} on the sum.

Recall now Prop.~\ref{prop:lim-simples} where we found that direct limits of simples are again simple modules,
and thus the only simple $\interinfin$-modules that appear in the direct limit space $\Hilb$ are the limits $\interM{j}$.
 We thus conclude this subsection with an important statement: \textit{the $U\interinfin$ simple modules that appear in $\Hilb$ are the direct limits $\interM{j}$ of the $\JTL{N}$ simples $\AIrrTL{j}{(-1)^{j+1}}$ and considered as modules over $\VirN(2)$ they are the direct sums~\eqref{JTL-simple-Vir-sum} over $\VirN(2)$ simples.}

 \subsection{Scaling limit of the anti-periodic chains}\label{sec:scal-lim-tw}
Of course, we could also consider the antiperiodic model (or the twisted sector) for the
symplectic fermions  where the   modes are half-integer~\cite{Kausch}. This corresponds to  the scaling limit of
the anti-periodic $\gl(1|1)$ spin-chain from
Sec.~\ref{subsec:antiper-mod}.  The
$\JTL{N}$ algebra is  then replaced by $\JTL{N}^{tw}$ while $\LQGodd$ is replaced by $U
s\ell(2)$. The corresponding bimodule is semisimple and is given
in~\eqref{antip-bimod}.

We take the scaling limit of $\JTL{N}^{tw}$ using Thm.~\ref{thm:Lie-JTLtw} where it was shown  that the spin-chain representation  of $\JTL{N}^{tw}$ is isomorphic to a representation of the enveloping algebra $U \ssp_{N}$.  We thus need only to take the scaling limit of the Lie algebra $\ssp_{N}$ which is quite straightforward. Repeating the analysis given in Sec.~\ref{sec:scaling} or following lines in App.~\ref{app:JTL-lim},
we obtain in this case that  $\ssp_{N}$ has the
 direct limit generated by the same  bilinears~\eqref{spinf-gen} as in the periodic case but now with $n,m\in\oZ-\half$. This representation of $\spinf$ gives a symplectic-fermion representation
of the universal enveloping algebra $U\interinfin$
in  the twisted model as well and the generators $\enrg_{m,n}$, $\strt_{m,n}$, and $\bstrt_{m,n}$  have  similar expressions as in~\eqref{prop:sympl-ferm-interLie} but with fermionic modes shifted by one-half:
 \begin{equation}\label{eq:sympl-ferm-tw-interLie}
\enrg_{m,n} =S_{\alpha\beta} \psi^{\alpha}_{m\mp\half} \bar{\psi}^{\beta}_{n\mp\half},\quad
\strt_{m,n}=J_{\alpha\beta}{\sferm^{\alpha}_{m\mp\half}\sferm^{\beta}_{n\mp\half}},\quad \bstrt_{m,n}=J_{\alpha\beta}{\bsferm^{\alpha}_{m\mp\half}\bsferm^{\beta}_{n\mp\half}}, \quad n,m\in\oZ/\{0\},
\end{equation}
where $m\mp\half$ is by definition $m-\half$ for positive $m$ and $m+\half$ for negative values of $m$.
In particular, we see that the radical of $\interinfin$ which is a subalgebra generated by $\enrg_{0,n}$, $\strt_{m,0}$, {\it etc}.,  is trivially represented in this model.

There are no zero modes in the anti-periodic model and the continuum limit of the $U \SU(2)$ generators reads as~\cite{GRS1}
\begin{equation}\label{scal-lim-sl2}
\tilde{Q}^a=d^a_{\alpha\beta}\sum_{n=0}^\infty
\left(\ffrac{\psi^\alpha_{-n-1/2}\psi^\beta_{n+1/2}}{n+1/2} - \ffrac{\bar{\psi}^\alpha_{-n-1/2}\bar{\psi}^\beta_{n+1/2}}{n+1/2}\right),
\end{equation}
with the matrices
\begin{equation}\label{da-mat}
d^0_{\alpha\beta}=
\half
\begin{pmatrix}
-1 & 0\\
0 & -1
\end{pmatrix}, \quad
d^1_{\alpha\beta}=
\half
\begin{pmatrix}
1&0\\
0&-1
\end{pmatrix}, \quad
d^2_{\alpha\beta}=
\half
\begin{pmatrix}
0&-1\\
-1&0\end{pmatrix},
\end{equation}
with $[\tilde{Q}^a,\tilde{Q}^b]=f^{ab}_c\tilde{Q}^c$ and $f^{01}_2=-1$.
We check then  that the action of the Lie algebra $\interinfin$ of  fermionic bilinears commutes with this  action of $U \SU(2)$. Moreover, we repeat the construction in App.~\ref{app:inv-lim} and prove Lem.~\ref{lem:coh} in this context. This allows then to state  an analogue of  Thm.~\ref{thm:centr-interLie} about the centalizer in the half-integer-mode sector.

\begin{thm}\label{prop:interch-centr-tw}
The centralizer $\cent_{U\interinfin}$ of the enveloping algebra $U\interinfin$ action
in the half-integer-mode (twisted) sector, which factorizes to the action of $U\spinf$ in this model,
 is given by the representation of $U \SU(2)$ from~\eqref{scal-lim-sl2}. The bimodule over $\bigl(\cent_{U\interinfin},U\spinf\bigr)$  is semisimple and given by the direct sum~\eqref{antip-bimod} over all $j\geq0$.
\end{thm}

This is confirmed by the  decomposition of  the full partition function to which we now turn.

\subsection{The full generating functions}\label{sec:gen-func}

The generating
function of levels for the periodic $\gl(1|1)$ model -- that is, the left hand
side of (\ref{charform}) reads, from the analysis in
\cite{ReadSaleur01},
\begin{equation}\label{genfct}
Z=F_{0,-1}+2\sum_{ j\text{ even}} F_{j,-1}+2\sum_{j \text{ odd}} F_{j,1}
\end{equation}
(in the reinterpretation as the partition function of the XX model
with appropriate twists, the factors~$2$ arise from the $\oZ_2$
spin-flip symmetry, so the summation index  $j$ is positive).  Elementary algebra using
\begin{equation}
F_{2n,-1}^{(0)}+F_{2n+1,1}^{(0)}=F_{2n,-1}-F_{2n+1,1}
\end{equation}
leads to
\begin{equation}
Z=4\sum_{j=0}^\infty jF_{j,(-1)^{j+1}}^{(0)}.
\end{equation}
This formula is a direct translation of the bimodule stucture
discussed in Sec.~\ref{ind-chain-bimod-subsec}: the same
decomposition holds in fact for a finite system, with $F^{(0)}$ the trace
over the simple modules over $\JTL{N}$.
Note that it is similar to the formula $Z_{\mathrm{op}}=4\sum_{j=0}^\infty j\chi_{j,1}$
 obtained in \cite{ReadSaleur07-2} for the open $\gl(1|1)$ spin chains.
 In the open case, the degeneracy $4j$ arose as dimensions of irreducible representations of the full quantum group, which coincides with the algebra ${\cal A}_{1|1}$. In the periodic case, we lose the full quantum group symmetry and retain only $\LQGodd$ on the lattice. This leads, nevertheless, to the same $4j$ for the
simple modules: $\AIrrTL{j}{(-1)^{j+1}} $ in the periodic case, $\VX_{j,1}$ in the open case.

Finally, we rewrite the full generating function as
\begin{equation}
Z=4\sum_{j=0}^\infty jF_{j,(-1)^{j+1}}^{(0)}=4\left|\sum_{j=0}^\infty j\chi_{j,1}\right|^2=4\left| q^{1/12}\prod_{n=1}^\infty (1+q^n)^2\right|^2.
\end{equation}
The right hand side is now seen to coincide with the generating
function for the level of symplectic fermions with periodic boundary
conditions indeed \cite{Kausch}. Of course, the corresponding
partition function of doubly periodic symplectic fermions on the torus
vanishes exactly due to the $\gl(1|1)$ symmetry; it is thus trivially
modular invariant. The generating function $Z$ in (\ref{genfct}) is not modular invariant, nor does it have to be.

In the anti-periodic case or twisted model,
 the generating function of levels is now
\begin{equation}
Z=F_{0,1}+2\sum_{\text{odd } j>0} F_{j,-1}+2\sum_{\text{even }j>0} F_{j,1}
\end{equation}
corresponding to
   \begin{equation}
   Z= \sum_{j=0}^\infty (j+1)  F_{j,(-1)^j}^{(0)}\,,
    \end{equation}
in accordance with the decomposition~\eqref{antip-bimod}.  The
$(j+1)'s$ are now dimensions of the $U s\ell(2)$-modules; the
$\gl(1|1)$ symmetry is broken. As before, we can rewrite
   \begin{equation}
   Z=\left|\sum_{j=1}j\chi_{j,2}\right|^2=
\left|q^{-1/24}\prod_{n=1}^\infty \left(1+q^{n+1/2}\right)^2\right|^2
    \end{equation}
which is a well known expression in this sector as well.

 \section{Interchiral algebra and local fields}\label{sec:interchiral}

We are now discussing a reinterpretation of our previous results in a way that, we expect, can be generalized to other LCFTs. This requires us to perform a few manipulations whose validity, however likely, we cannot prove at this stage. The present section, while more physical,  is thus of a less mathematical nature, and, ultimately, based on some conjectures we discuss below.

While the structure of the symplectic fermion LCFT  is fully
consistent with the one of the $\gl(1|1)$ spin chain, there is an
obvious difference due to the presence of the $\nSU(2)$ symmetry in
the continuum theory, which, in the periodic case,  does not have a natural analog on the
lattice.  Investigation of more complicated
models~\cite{GRS4} suggests that the most productive way to think of
this difference is to essentially forget the $\nSU(2)$ symmetry, which
seems to be an artifact from the point of view of the $\gl(1|1)$
theory.  Rather, we believe that the bimodule structure of the lattice
model does suggest to us the proper symmetries, and the proper way to
analyze the scaling limit. Put otherwise, the good algebraic object
that will organize the spectrum of most general logarithmic lattice
models should be {\bf the scaling limit of the Jones--Temperley--Lieb
algebra}.  Switching our point of view in this way leads to profound
and maybe not so surprising conclusions, in particular that the
organizing algebra of bulk LCFTs should contain non-chiral objects.
Indeed, while the  focus in the past~\cite{KooSaleur} had been mostly to extract the stress-energy
tensor modes from the
Temperley--Lieb or Jones--Temperley--Lieb algebras, it is well known \cite{ReadSaleur01} that the scaling limit of
some elements in $\JTL{N}$ can lead to other physical observables corresponding
to different bulk scaling fields.  A very important example of
such a field is  what we will call the energy operator, which is  associated
with the staggered sum
\begin{equation}\label{energypotts}
\sum_{i=1}^{N} (-1)^i e_i\; \mapsto\; \int dx\, \Phi_{2,1}\times \overline{\Phi}_{21}(x,\tau=0),
\end{equation}
where the integral is taken over the
  circumference of a cylinder at constant imaginary-time
  $\tau=0$. In the Potts model of statistical mechanics, where  the Temperley Lieb algebra with positive, integer values of $m$ appears, this field is canonically coupled to the temperature.   The field  in~\eqref{energypotts} is the non chiral degenerate field with conformal weights $h=\bar{h}=h_{2,1}$; later,  we will denote it by
  $\enrgcyl(x,\tau)$.  Of course, the introduction of such fields in the organizing
  algebra of a LCFT requires discussion of objects which mix chiral
  and anti-chiral sectors. We shall in this section introduce the new
  concept of an {\bf interchiral algebra}, and discuss its structure
  and role in the case of  the closed $\gl(1|1)$ chains
  and the bulk symplectic fermions.

  It is important to note that in the open case, the continuum limit of  the Temperley--Lieb algebra $\TL{N}$ only leads to the enveloping algebra of the Virasoro algebra, and does not involve
  other conformal fields. A related fact is, for instance, that the boundary energy operator at the ordinary transition in the Potts model coincides with the stress energy tensor \cite{BatchelorCardy}. There are probably cases for open spin chains where the algebra of hamiltonian densities will lead, in the continuum limit, to a bigger algebra than the Virasoro algebra -- for instance, the super Virasoro algebra, $W$-algebras~\cite{GST}, {\it etc.}

From another, more formal point of view, recall we  showed in~\cite[Sec. 5.2]{GRS1} that the \scal limit of the JTL centralizer
$\centJTL$ -- an infinite dimensional representation of the $\LQGodd$ -- gives
an algebra of intertwining operators respecting the left and right
Virasoro algebras. On the other hand, we also showed~\cite{GRS1} that the
centralizer of $\VirN(2)$ contains $\LQGodd$ but is a bigger algebra.
It is thus reasonable to expect that the scaling limit  of the Jones-Temperley-Lieb algebra
 is bigger than the non-chiral Virasoro algebra $\VirN(2)$, and that the
 latter  is only a proper subalgebra.

 To make progress, we turn back to the
 scaling limit analysis of the $\gl(1|1)$ spin-chains discussed
 in~\cite{GRS1} and extract an additional field that generates the
 full scaling limit of the JTL algebras.
  It turns out that the limit
gives modes of a field  of conformal dimension $(1,1)$.
We call this field
the \textit{interchiral field} $\enrg(z,\bar{z})$; it is
expressed in terms of  derivatives of the symplectic fermions as
\begin{equation}\label{interfielddef}
\enrg(z,\bz)=S_{\alpha\beta}\psi^{\alpha}(z)\bar{\psi}^{\beta}(\bz)\qquad \text{with}\quad
\psi^{\alpha}(z)=\der\Phi^{\alpha}(z,\bz),\quad
\bar{\psi}^{\alpha}(\bar{z})=\bder\Phi^{\alpha}(z,\bz),
\end{equation}
where we introduced the symmetric form
\begin{equation}
 S_{12}=S_{21}=1, \quad
S_{11}=S_{22}=0.
\end{equation}
We recall that the derivatives of the symplectic fermions have the mode decomposition~\cite{Kausch}
\begin{equation}
\psi^{\alpha}(z) = \sum_{n\in\oZ}\psi^{\alpha}_n z^{-n-1},\qquad \bar{\psi}^{\alpha}(\bz) = \sum_{n\in\oZ}\bar{\psi}^{\alpha}_n \bz^{-n-1},
\end{equation}
and are primary fields of conformal dimensions $(1,0)$ and $(0,1)$, respectively. The fermionic modes  satisfy
 the anti-commutation relations
\begin{equation}\label{sferm-rel}
\{\psi^{\alpha}_m,\psi^{\beta}_{m'}\} =
mJ^{\alpha\beta}\delta_{m+m',0}\,,\qquad \alpha,\beta\in\{1,2\},\quad m,m'\in\oZ,
\end{equation}
with the symplectic form $J^{12}=-J^{21}=1$, and the same formulas for the antichiral modes $\bar{\psi}^{\alpha}_n$, and $\{\psi^{\alpha}_m,\bar{\psi}^{\beta}_{m'}\} = 0$. We also introduce `constant' modes $\phip_0$ and $\phim_0$ which are conjugate to the zero modes
\begin{equation}\label{sferm-0-rel}
\left\{\phim_0,\sfermp_0\right\}=i,\qquad
\left\{\phip_0,\sfermm_0\right\}=-i.
\end{equation}

\subsection{Modes of local fields and the algebra $\interch$}\label{sec:interch-modes}

While dealing with  bilinears in fermionic modes is quite convenient and provides  the most obvious approach to the scaling limit of JTL algebras, it is also important to understand what happens in more physical terms. From a  (L)CFT point of view indeed, the natural objects are, for instance,  not so much bilinears such as $\nord{\sfermp_{n-m}\sfermm_m}$, but the modes of the stress energy trensor, which are (infinite) sums such as $L_n=\sum_{m\in\oZ} \nord{\sfermp_{n-m}\sfermm_m}$. Also, while  Howe duality (in its non semi-simple version) gave us a powerful and new angle on $\JTL{N}$ and its scaling limit, it is not likely that this will  generalize to other models. It is thus crucial to be able to come up with an alternative, if less pleasant, approach to the scaling limit of JTL algebras.

\subsubsection{Higher Hamiltonians}\label{sec:hham}
We begin by  recalling the lattice analysis given in~\cite{GRS1}
where it was shown how to proceed from the $\JTL{N}$
generators to get  Virasoro modes in the non-chiral logarithmic
conformal field theory of symplectic fermions:
 the combinations
\begin{equation}\label{HPn-def}
H(n) = -\sum_{j=1}^N e^{-iqj} e^{\gl}_j,\qquad
P(n)=\ffrac{i}{2}\sum_{j=1}^N e^{-iqj} [e^{\gl}_j,e^{\gl}_{j+1}], \qquad q=\ffrac{n\pi}{L},
\end{equation}
of the
$\JTL{N}$ generators (recall that $\e^{\gl}_j$ denotes the representation of $e_j$ given in~\eqref{rep-TL-1} and~\eqref{rep-JTL-2}) converge in the scaling limit as $L\to\infty$ to the well-known symplectic fermions representation
of the left and right Virasoro generators
\begin{equation}\label{lim-HP-finite}
\ffrac{L}{2\pi}H(n) \mapsto
L_{n}+\bar{L}_{-n}, \qquad \ffrac{L}{2\pi}P(n) \mapsto
L_{n}-\bar{L}_{-n}
\end{equation}
We note that the limit~\eqref{lim-HP-finite} of $H(n)$ and $P(n)$ is
taken for finite $n$.

Here, we go  further and consider also limits
with $n$ close to $N/2$ (corresponding, physically, to alternating sum of JTL generators, related with the energy operator), or equivalently we consider $H(L-k)$ and
$P(L-k)$ with $L\to\infty$ and finite $k$. It turns out that this limit
gives modes of
a field of conformal dimension $(1,1)$.
To show this,
we introduce a family  of operators $H_l(n)$ which are
Fourier images of all the higher Hamiltonians:
\begin{equation}\label{hham-def}
H_l(n) = -\half e^{-il\frac{\pi}{2}}\sum_{j=1}^Ne^{-iqj}E_{j,l},\qquad q=\ffrac{n\pi}{L} \quad\text{and}\quad l\in\oN,\quad n\in\oZ,
\end{equation}
where we used the following notation for  multiple
commutators of the $\JTL{N}$-generators $e_j$
\begin{equation}
E_{j,l} = \Bigl[e_j,\bigl[e_{j+1},\dots
 e_{j+l-2},[e_{j+l-1},e_{j+l}]\dots\bigl]\Bigl], \qquad 1\leq j\leq N.
\end{equation}
We call these Fourier images  generalized higher Hamiltonians. Note that $P(n)=H_1(n)$ and we set $H_0(n)=H(n)$.

The fermionic expressions for the $H_l(n)$'s were obtained in our first paper~\cite[Sec.~4.5.1]{GRS1} and
it turns out that the $H_l(n)$'s can be written in the basis~\eqref{interfin-def} of the Lie algebra $\interfin{N}$, for $l\geq0$
and $0\leq n\leq L-1$, as
\begin{multline}\label{Hln-pos}
H_l(n) =
2e^{iq\frac{l+1}{2}}\Biggl( \sum_{m=1}^{L-n-1}
\cos{l\bigl(p+\ffrac{q}{2}\bigr)}\sqrt{\sin{(p)}\sin{(q+p)}}\, \Asp_{m,m+n}\\
+\half\sum_{m=1}^{n-1}
\cos{l\bigl(p-\ffrac{q}{2}\bigr)}\sqrt{\sin{(p)}\sin{(q-p)}}\bigl(\Csp_{m,n-m}
+(-1)^l \Bsp_{L-m,L+m-n}\bigr)\\
+\cos{\ffrac{lq}{2}}\sqrt{\sin{q}}\bigl(\Egl_{0,n} +(-1)^l \Egl_{0,2L-n} + \Egl_{L+n,L} +
 (-1)^l\Egl_{L-n,L} \bigr) +
2\delta_{n,0}\bigl(1+(-1)^l\bigr)\Egl_{0,L} \Biggr),
\end{multline}
where we set $p=m\pi/L$ and $q=n\pi/L$ as usually. The expressions for negative values of $n$ are given in~\eqref{Hln-neg}.
Note that all zero modes $H_l(0)$ and the lattice translation
generator $u^2$  have very simple expressions:
\begin{equation}\label{hamilt-sp}
H_l(0) = 2\sum_{m=1}^{L-1}
\cos{(lp)}\sin{(p)}\, \Asp_{m,m} + 4\bigl(1+(-1)^l\bigr)\Egl_{0,L}
\end{equation}
and
\begin{equation}
u^2 = \exp\left(-\ffrac{2i\pi}{L}\sum_{n=1}^{L-1}n\Asp_{n,n}\right)
\end{equation}
and they belong, for odd values of $l$, to the Cartan
subalgebra of $\ssp_{2L-2}$.

We check  then that the element $H_0(L)$ is the energy operator indeed $\sum_{j=1}^N(-1)^j e_j$ on the lattice (which differs from the Hamiltonian $H$), which reads
\begin{equation}
\sum_{j=1}^N(-1)^j e_j =2i \sum_{\substack{p=\step\\\text{step}=\step}}^{\pi-\step}\sin{p}\bigl(\chi^{\dagger}_p\eta_{p}
+ \eta^{\dagger}_p\chi_{p}\bigr).
\end{equation}
The  \energy operator rescaled by the factor $\frac{L}{2\pi i}$ has then the limit
\begin{equation*}
\ffrac{L}{2\pi i}\sum_{j=1}^N(-1)^j e_j \mapsto \sum_{m\in\oZ} \bigl(\sfermp_{m}\bsfermm_m + \sfermm_{m}\bsfermp_m \bigr)\equiv \enrg_0.
\end{equation*}
We denote this limit by $\enrg_0$ and show below that it corresponds to the zero mode of a primary field of dimension $(1,1)$. First however, we discuss what will correspond to higher modes $\enrg_n$ of this primary field. We refer to our calculations in App.~\AppHln~where we obtain in
the scaling limit for $H_0(L-k)$ (keeping only the leading order in $1/N$), with finite integer
$k$,
\begin{equation}\label{scal-lim-Hn}
\ffrac{L}{2\pi i}H_0(L-k) \mapsto  S_{\alpha \beta}\sum_{m\in\oZ}\sferm^{\alpha}_{m}\bsferm^{\beta}_{m+k}  \equiv
\enrg^{(0)}_k,\qquad k\in\oZ,
\end{equation}
and the scaling limit for $H_1(L-k)$ is
\begin{equation}
-\ffrac{L^2}{\pi^2}H_1(L-k) \mapsto S_{\alpha
 \beta}\sum_{m\in\oZ}(2m+k)\sferm^{\alpha}_{m}\bsferm^{\beta}_{m+k}\equiv
 \enrg^{(1)}_k,\qquad k\in\oZ,
\end{equation}
where we introduce, for $l\geq 0$ and $k\in\oZ$, the operators
\begin{equation}\label{enrgl-def}
\enrg^{(l)}_k=S_{\alpha \beta}\sum_{m\in\oZ}(2m+k)^l\sferm^{\alpha}_{m}\bsferm^{\beta}_{m+k}.
\end{equation}

The whole family of operators $\enrg^{(l)}_k$, with $l\geq0$ and $k\in\oZ$, can be extracted from the scaling limit of
the operators $H_{l'}(L-k)$ keeping not only the leading terms in their formal expansion in $1/L$ but also all sub-leading terms in the
expansion. Put differently, we can formally expand $H_{l'}(L-k)$ as Laurent series in $N$ where coefficients are fermionic bilinears of the form  $\enrg^{(l)}_k$. For simplicity, we refrain from discussing this any further here.

\subsubsection{Fields}\label{sec:interch-field}

The foregoing family of operators
$\enrg^{(l)}_k$, with $l\geq0$ and $k\in\oZ$, can be obtained as modes
of a primary field of conformal dimension $(1,1)$, and of its
descendants. We denote this primary
field defined on the complex plane as
\begin{equation}\label{enrg-field-pl}
\enrg_{\text{pl.}}(z,\bz)=S_{\alpha\beta}\psi^{\alpha}(z)\bar{\psi}^{\beta}(\bz),
\end{equation}
where  $S_{\alpha\beta}$ is the symmetric form, {\it i.e.}, $S_{12}=S_{21}=1$, $S_{11}=S_{22}=0$. We  call this field the \textit{inter}chiral field; note that it is of course  neither chiral nor anti-chiral.

Recall first the mode decomposition of the fermion fields
\begin{equation}
\psi^{\alpha}(z) = \sum_{n\in\oZ}\psi^{\alpha}_n z^{-n-1},\qquad \bar{\psi}^{\alpha}(\bz) = \sum_{n\in\oZ}\bar{\psi}^{\alpha}_n \bz^{-n-1}.
\end{equation}
Then,  $\enrgpl(z,\bz)$ has a formal expansion
\begin{equation}\label{enrgpl-mod-exp}
\enrgpl(z,\bz) =  S_{\alpha\beta} \sum_{n,m\in\oZ}\psi^{\alpha}_n \bar{\psi}^{\beta}_m  z^{-n-1}\bz^{-m-1}.
\end{equation}
Interpreting  $z$ and $\bz$  as independent complex coordinates allows one to extract $\psi^{\alpha}_n \bar{\psi}^{\beta}_m$ by a double integration over $z,\bz$. This is not too pleasant however, since we rather would like to deal with modes obtained by  integration in the physical theory, where $z,\bz$ are not independent. Note also that in the Virasoro case, there is no $\bz$ we can integrate over, so there is no integration that can give rise to the single product $\psi^{\alpha}_n \psi^{\beta}_m$ starting from the stress energy tensor and its descendents.

To make contact with the hamiltonian formulation we used so far, it is now convenient to perform a conformal transformation $z\to w(z)$, $\bz\to w^*(\bz)$
onto a cylinder of   circumference $\LL$, with  $w(z)=\frac{\LL}{2\pi}\log(z)$.
We set
$w=\tau -i x$ and now $w^*=\tau +i x$, we then obtain
\begin{equation}
\enrgcyl(x,\tau) = \Bigl(\ffrac{2\pi}{\LL}\Bigr)^2z\bz\enrgpl(z,\bz)
=\Bigl(\ffrac{2\pi}{\LL}\Bigr)^2  S_{\alpha\beta}
\sum_{n,m\in\oZ}\psi^{\alpha}_n \bar{\psi}^{\beta}_m  \exp\Bigl[-\ffrac{2\pi}{\LL}\bigl(\tau(m+n)+ix(m-n)\bigr)\Bigr].
\end{equation}
Introducing the mode expansion of the field $\enrgcyl(x,\tau)$ on the
cylinder in the Heisenberg presentation, where the modes depend on the
imaginary time $\tau$,
\begin{equation}\label{enrg-modes-exp}
\enrgcyl(x,\tau) = \Bigl(\ffrac{2\pi}{\LL}\Bigr)^2
\sum_{k\in\oZ}\enrg_k(\tau)e^{2\pi ik x/\LL},
\end{equation}
we then obtain
\begin{equation}
\enrg_k(\tau) = \ffrac{\LL}{(2\pi)^2}\int_{0}^{\LL}dx e^{-2\pi i
  k x/\LL} \enrgcyl(x,\tau)
\end{equation}
which is evaluated as
\begin{equation}
\enrg_k(\tau) = S_{\alpha\beta} \sum_{n\in\oZ}\psi^{\alpha}_{n+k}
\bar{\psi}^{\beta}_n e^{-2\pi \tau(2n+k)/\LL}.
\end{equation}
Finally, we get the expansion
\begin{equation}\label{enrg-modes}
\enrg_k(\tau) = \sum_{l\geq0} \ffrac{(-2\pi \tau)^l}{l!(\LL)^l} \enrg^{(l)}_{-k}
\end{equation}
and in particular
\begin{equation}
 \enrg_{k}(0) = \enrg^{(0)}_{-k} ,\qquad k\in\oZ,
\end{equation}
where $\enrg^{(l)}_{k}$ were introduced above in~\eqref{enrgl-def}.

It is important to note that  applying the zero mode $\enrg_0$ of the local field $\enrgcyl(x,\tau)$  on the vacuum state, we get then an infinite sum
$\sum_{m<0} \bigl(\sfermp_{m}\bsfermm_m + \sfermm_{m}\bsfermp_m \bigr)\vac$. All other modes $\enrg_k$ have similar formally divergent action. These formal divergences are not very surprising as a local operator in general is able to excite states of arbitrarily high energy (the spin is fixed in the action of our local operators). Note nevertheless that the  modes $\enrg_k(\tau)$ in the Heisenberg picture are well-defined as they have the convergence factor $e^{-2\pi \tau(2n+k)/\LL}$ which eliminates high excitations. We could thus regularize the action of our local operators by introducing convergence factors in the formally divergent sums but this is not very convenient for us. In order to describe representations of the algebra generated by these local operators, we prefer to choose another regularization in which a locally defined field at a point $w=\tau-ix$ becomes a bilocal field at points $z$ and $\bz$ (hence formally ``point splitting'' by going to $\oC\times \oC$) .
So, we interpret the field $S(z,\bz)$ as  a regularization of the local interchiral field on the cylinder
where the variables $z$ and $\bz$ might be considered independent. After such a regularization the modes of $S(z,\bz)$ given in~\eqref{enrgpl-mod-exp} now indeed generate states of finite energy only.

Recall now that the modes of the regularized field, $S(z,\bz)$, are elements of the Lie algebra $\interinfin$ which is a subalgebra in the algebra $\glinf$ of infinite matrices with finite number of non-zero elements. It turns out that the modes of the local, non-regularized, field  $\enrgcyl(x,\tau)$ belong to the algebraic completion $\bglinf$ of this $\glinf$. The Lie algebra $\bglinf$ was already mentioned above in Sec.~\ref{sec:spinf} and it consists of infinite matrices with finite number of diagonals having (possibly infinite number of) non-zero elements, see also App.~\Appglinf.
Working now in the completion $\bglinf$, we can ask about the Lie algebra generated by the local operators.
We show below that all the other operators $\enrg^{(l)}_{k}$ that appear in
the scaling limit of the $H_{l'}(n)$ operators
 are produced from
$\enrg^{(0)}_{k}$ by the conformal generators $L_n$ and $\bar{L}_n$. It is  of course interesting
to study, more generally,   commutation relations involving these operators, that is, to study
the Lie algebraic structure generated from the $\enrg^{(l)}_k$.

\subsubsection{Commutation relations and the Lie algebra $\interch$}
\label{sec:comm-rel-interch}
We first find commutation relations among the modes $\enrg^{(0)}_k$
with the result
\begin{equation}\label{comm-rel-enrg-zero}
[\enrg^{(0)}_r,\enrg^{(0)}_s] =
J_{\alpha\beta}\sum_{n\in\oZ}\bigl((n-s)\bsferm^{\alpha}_{n}\bsferm^{\beta}_{r+s-n}
 - (n+r)\sferm^{\alpha}_{n}\sferm^{\beta}_{-(r+s)-n}\bigr).
\end{equation}
Recall that the modes of the chiral stress-energy
tensor can be written as $L_k=\half
J_{\alpha\beta}\sum_{n\in\oZ}\nord{\sferm^{\alpha}_{n}\sferm^{\beta}_{k-n}}$,
and similarly for the anti-chiral modes. We introduce also
\begin{equation}\label{virm-def}
\virm{l}_k =
\half J_{\alpha\beta}\sum_{n\in\oZ} n^l
\nord{\sferm^{\alpha}_{n}\sferm^{\beta}_{k-n}}, \qquad
\bvirm{l}_k =
\half J_{\alpha\beta}\sum_{n\in\oZ} n^l \nord{\bsferm^{\alpha}_{n}\bsferm^{\beta}_{k-n}},\qquad l\geq0,\; k\in\oZ,
\end{equation}
with $\virm{0}_k=L_k$ and $\bvirm{0}_k=\bar{L}_k$.
Then,~\eqref{comm-rel-enrg-zero} has the form, for $r\ne-s$,
\begin{equation}\label{enrg-enrg}
[\enrg^{(0)}_r,\enrg^{(0)}_s] = 2\bvirm{1}_{r+s} - 2 \virm{1}_{-(r+s)} - 2s
\bvirm{0}_{r+s} - 2r \virm{0}_{-(r+s)}.
\end{equation}
Note that the new modes $\virm{1}_k$ and $\bvirm{1}_k$
together with $\virm{0}_n$ and $\bvirm{0}_n$ generate a subalgebra:
\begin{align}
2[\virm{1}_n, \virm{1}_m] &= m(n-2m) \virm{1}_{n+m} + m^3 \virm{0}_{n+m},\\
2[\bvirm{1}_n,\bvirm{1}_m] &= m(n-2m) \bvirm{1}_{n+m} + m^3 \bvirm{0}_{n+m},\\
[\virm{1}_n,\bvirm{1}_m]&=0,\quad [\virm{1}_n,\bvirm{0}_m]=0,\quad [\bvirm{1}_n,\virm{0}_m]=0,\\
[\virm{1}_n,\virm{0}_m] &= (n-2m) \virm{1}_{n+m} + m^2 \virm{0}_{n+m},\\
[\bvirm{1}_n,\bvirm{0}_m] &= (n-2m) \bvirm{1}_{n+m} + m^2 \bvirm{0}_{n+m},
\end{align}
where we  give relations for $n\ne-m$ (at the case $n=m$ there is a central charge term in all these commutators, which is not important for our purposes).
Note also that $2[\virm{1}_n,\virm{1}_m] = m[\virm{1}_n,\virm{0}_m]$.

In general, we have
\begin{equation}\label{enrg-enrg-gen}
[\enrg^{(l)}_r,\enrg^{(l')}_s] = \sum_{k=0}^{l+l'+1}\bigl(a_k\virm{k}_{-r-s} +\bar{a}_k\bvirm{k}_{r+s}\bigr),
\end{equation}
where $a_k$ and  $\bar{a}_k$ are coefficients in the expansions
\begin{gather}
\sum_{k=0}^{l+l'+1}a_k x^k = (-1)^{l'+1}(2x+r)^l(x+r)(2x+2r+s)^{l'}\\
\sum_{k=0}^{l+l'+1}\bar{a}_k x^k = (-1)^{l+1}(2x-s)^{l'}(-x+s)(2x-2s-r)^{l},
\end{gather}
and we again suppose that $r\ne-s$.

We compute next  for all higher modes the action of the left  and right Virasoro, {\it i.e.}, commutation relations between $\enrg^{(l)}_r$ and
$\virm{0}_s$, $\bvirm{0}_s$:
\begin{equation}\label{comm-rel-enrg-Ln}
\begin{split}
2[\enrg^{(l)}_r,\virm{0}_s] &= \enrg^{(l+1)}_{r-s} - \sum_{k=0}^l (-s)^{l-k}\left[s\binom{l+1}{k}+r\binom{l}{k}\right] \enrg^{(k)}_{r-s},\\
2[\enrg^{(l)}_r,\bvirm{0}_s] &= \enrg^{(l+1)}_{r+s} - \sum_{k=0}^l (-s)^{l-k}\left[s\binom{l+1}{k}-r\binom{l}{k}\right] \enrg^{(k)}_{r+s},
\end{split}
\qquad r,s\in\oZ.
\end{equation}

In general, we have
\begin{multline}
[\enrg^{(l)}_r,\virm{l'}_s] = \ffrac{-1}{2^{l'+2}} \sum_{k=0}^{l+1}\sum_{k'=0}^{l'+1} (-s)^{l-k}
\Biggl[s(s-r)^{l'-k'}\binom{l+1}{k}\binom{l'}{k'} + r(s-r)^{l'-k'}\binom{l}{k}\binom{l'}{k'}\\
+(-1)^{l'}(r+s)(-s-r)^{l'-k'}\binom{l}{k}\binom{l'+1}{k'}\Biggr] \enrg^{(k+k')}_{r-s}
\end{multline}
and a similar formula for $[\enrg^{(l)}_r,\bvirm{l'}_s]$, with  $r,s\in\oZ$.

Although it is obvious that the commutator $[\virm{l}_n,\virm{l'}_m]$ is expressed as a linear combination of $\virm{l''}_{n+m}$ with $l''\leq l+l'$, its precise expression  looks more complicated and we do not give it for simplicity. We note also that the commutators from~\eqref{enrg-enrg}-\eqref{enrg-enrg-gen} have central charge terms which are the vacuum expectations of these commutators (we do not give explicit computations for brevity).  Finally, we  conclude that the vector space with the basis $\enrg^{(l)}_n$, $\virm{l}_n$ and $\bvirm{l}_n$, with $l\geq0$ and $n\in\oZ$, has a Lie algebra structure. We will denote  \textit{this Lie algebra as $\interch$}.
It is noteworthy to mention that all operators that appear in
 repeated commutators of  $\enrg^{(0)}_n$, including all
 $\enrg^{(l)}_n$, do commute
 with the action of $\LQGodd$. We thus have that the action of the enveloping algebra $U\interch$ commutes with $\LQGodd$.

\subsection{The relation between the $\interch$ and $\interinfin$ algebras and completions}\label{sec:interch-interinfin}
We have so far obtained two different Lie algebras. One, $\interinfin$, is generated by bilinears in fermion modes, while the other, $\interch$, is generated by
the modes of bilinears in fermionic fields, which are local operators. The corresponding objects such as $\enrg^{(l)}_n$ and so on expand on an infinite sum of bilinears in fermion modes, and we now have to face the question of the equivalence of these two descriptions. We note first that the  generators of the second algebra, $\interch$,  are formally infinite sums of the generators of $\interinfin$. On the other hand, any element in the first algebra  $\interinfin$, is expressed (by definition) as a {\sl finite} linear combination of  fermion bilinears.  Therefore, in order to  compare we should
 first take an algebraic completion $\interinfinco$ of the Lie algebra $\interinfin$, that now admits  infinite sums like those for  $\enrg^{(l)}_n$. The completion $\interinfinco$ is defined as a Lie subalgebra in $\bglinf'$ (the completion of $\glinf$ introduced in Sec.~\ref{sec:spinf} and App.~\Appglinf)
that contains any (possibly infinite) linear combinations of the bilinears
\begin{equation}\label{interinfinco-def}
\interinfinco:\qquad\qquad X_{\alpha,\beta}\nord{\nfer^{\alpha}_m\nfer^{\beta}_{-n}}\qquad \text{for a fixed value of}\;\; (m-n),
\end{equation}
where the fermions $\nfer^{\alpha}$ are defined in~\eqref{nfer-sferm} and $X_{\alpha,\beta}=S_{\alpha,\beta}$ for $m$ and $n$ of different parity and $X_{\alpha,\beta}=J_{\alpha,\beta}$ for $m$ and $n$ of same parity.
 Then, recalling  the expressions~\eqref{enrgl-def} and~\eqref{virm-def} we obviously have inclusions
 \begin{equation*}
 \interch\subset\interinfinco\subset\bglinf'.
 \end{equation*}
 So, considering the completion of the Lie algebra appearing in the scaling limit of our spin-chains is crucial if we want to describe local operators.

Strictly speaking, the algebra $\interinfinco$, as well as its subalgebra $\interch$,  does not act on the space $\Hilb$ of the scaling states. It acts
on the completion
\begin{equation}\label{eq:cHilb}
\cHilb=\prod_{n\geq0} \Hilb^{(n)},\qquad
 \text{where}\; \Hilb^{(n)} \;\text{are root/eigen-spaces for}\; E=L_0+\bar{L}_0,
\end{equation}
 constructed initially as the projective limit of our spin-chains in App.~\ref{app:inv-lim} and described in details in the paragraph preceding Cor.~\ref{cor:coh}.
This is the completion in the so-called formal topology~\cite{Kac-book}: all subspaces $\bigoplus_{\oN/M}\Hilb^{(n)}\subset\Hilb$, for a finite subset $M\subset\oN$, are declared to be the fundamental system of neighbourhoods of zero. Then, the completion of $\Hilb$ in this  topology is the direct product $\prod_{\oN}\Hilb^{(n)}$ of its homogeneous, or fixed energy, subspaces $\Hilb^{(n)}$, and $\Hilb\subset\cHilb$ is a dense subspace.
We will denote the representation of the universal enveloping algebra  of the Lie algebra $\interinfinco$ on the completed space by $\rep :U\interinfinco \to \Endo\,\cHilb$ and for simplicity  denote the image $\rep(\interinfinco)$ of the abstract Lie algebra just by $\interinfinco$ in what follows.

 It is interesting to go back for a while to the finite lattice and to mention that there is
an isomorphism between the Lie algebra generated by the generalized higher Hamiltonians $H_l(n)$, from which we extracted operators like $\enrg^{(l)}_n$, and the Lie algebra $\interfin{N}$ introduced in  Sec.~\ref{Howe} in order to give  special  generators of $\repgl\bigl(\JTL{N}\bigr)$ as bilinear monoms in the fermionic modes.
The existence of such an isomorphism is discussed  in App.~\AppHln.1.
The $\interfin{N}$ algebras in the scaling limit give the $\interinfin$, while
 the Lie algebra of $H_l(n)$'s gives in the  limit the Lie algebra $\interch$.
It is natural to  expect that the isomorphism between the two Lie algebras for finite spin chains carries over to their scaling limit, after taking a proper completion.
The completion $\interchco$ (of the image of the representation) of $\interch$ acting in $\cHilb$ is defined as the Lie algebra with the basis of $\interch$ but containing also  infinite linear combinations of $\enrg^{(l)}_n$ or $L_n^{(l)}$, or $\bar{L}_n^{(l)}$'s for any fixed $n$ such that
\begin{equation}\label{interchco-def}
\interchco:\qquad\quad \sum_{l\geq 0} C_l \enrg^{(l)}_n,\qquad \sum_{l\geq 0} C_l L^{(l)}_n,\qquad \sum_{l\geq 0} C_l \bar{L}^{(l)}_n,\qquad \text{with}
\quad \lim_{l\to\infty}C_l=0.
\end{equation}
Actually, we will require a stronger condition: for any (infinite) sequence of the constants $C_l$, the series $\sum_{l=0}^{\infty}C_l$ has to be convergent. Note that this way defined completion $\interchco$ is a subalgebra in $\bglinf'$.
We can show that the infinite combinations in~\eqref{interchco-def}
with the coefficients series $\sum_{l=0}^{\infty}C_l$ equal $0$ or $1$ converge on $\cHilb$ to the generators of $\interinfin$.
We first note that the $\interch$ generators are obtained from those of $\interinfin$ by infinite Vandermonde-type matrices:
\begin{equation}\label{interch-gen}
2L_k^{(l)} = \sum_{n\in\oZ}n^l \nord{\strt_{n,k-n}},\qquad
2\bar{L}_k^{(l)} = \sum_{n\in\oZ}n^l \nord{\bstrt_{n,k-n}},\qquad
\enrg^{(l)}_k = \sum_{n\in\oZ}(2n+k)^l \enrg_{n,k+n},\qquad l\geq0
\end{equation}
where we set $n^0=1$, for any $n$.   Let us consider for simplicity only the case $k=0$ in details. We start with $L_0^{(l)}$, which are non-zero only for even values of $l$ (recall that $\strt_{n,-n}=\strt_{-n,n}$). Then, we can write
\begin{equation}
L_0^{(2l)} = \half\delta_{l,0}\strt_{0,0} + \sum_{n>0} n^{2l}\strt_{n,-n},
\end{equation}
or introducing infinite vectors $\Lv=(L_0^{(0)},L_0^{(2)},L_0^{(4)},\dots)^T$ and $\Tv=(\half\strt_{0,0},\strt_{1,-1},\strt_{2,-2},\dots)^T$, we have
\begin{equation}
\Lv=\vand \cdot \Tv,
\end{equation}
where $\vand$ is the transposed of the
 classical infinite Vandermonde matrix
\begin{equation}
\vandT =
\begin{pmatrix}
1 & x_0 & x_0^2 & \cdots & x_0^n & \cdots\\
1 & x_1 & x_1^2 & \cdots &  x_1^n & \cdots\\
\vdots & \vdots & \ddots & \vdots & \vdots &\vdots\\
1 & x_m & x_m^2 & \cdots &  x_m^n & \cdots\\
\vdots & \vdots & \vdots & \vdots &\ddots &\vdots
\end{pmatrix},
\qquad \text{with}\quad x_j=j^2.
\end{equation}
Now, we face a problem of finding the
 inverse  $\vand^{-1}=(C_{k,n})_{k,n\geq0}$. In general, it is a non-trivial problem but our Vandermonde matrix is special: it has one row and one column of units $1$'s ($x_1=1$) and this property allows at least to prove that the series $\sum_{n\geq0}C_{k,n}$ associated with each row of $\vand^{-1}$ converge. Indeed, we have by definition of the inversed matrix
 \begin{equation}\label{yy}
 \sum_{n\geq0} C_{k,n} x_m^n  = \delta_{m,k}.
 \end{equation}
Introduce then formal power series in $x$ as
\begin{equation}\label{pow-ser}
P_k(x) = \sum_{n\geq0} C_{k,n} x^n
\end{equation}
and using~\eqref{yy} they take the following values at  points in the sequence $\{x_0,x_1,x_2,\ldots\}$:
\begin{equation}
P_k(x_m)=\delta_{m,k}.
\end{equation}
It is now obvious that
\begin{equation}
\sum_{n\geq0}C_{k,n} = P_k(1) = P_k(x_1) = \delta_{k,1}
\end{equation}
and in particular $\lim_{n\to\infty}C_{k,n}=0$.
We have thus shown that
\begin{equation}
\strt_{0,0} = 2\sum_{l\geq0}C_{0,l} L_0^{(2l)},\qquad\qquad
\strt_{k,-k} = \sum_{l\geq0}C_{k,l} L_0^{(2l)}
\end{equation}
are well defined and belong to our completion $\interchco$.
Now, we turn to the zero modes $\enrg^{(l)}_0$ and
introduce an infinite vector $\Lv = \bigl(\frac{1}{2^l}\enrg^{(l)}_0\bigr)_{l\geq0}$. It can be expressed again as a product of a new Vandermonde matrix
  $\vand=(x_m^n)_{n\geq0,m\in\oZ}$, where $x_m=m$, with the vector $\Tv=(\enrg_{k,k})_{k\in\oZ}$. Note that $\vand$ has again one column and one row of units.
Proceeding as above, we have again that the inverse  $\vand^{-1}=(C_{k,n})_{k\in\oZ,n\geq0}$ has rows with the properties that the series $\sum_{n\geq0}C_{k,n}$ converge to $0$ or $1$. Therefore, the elements $\enrg_{k,k}=\sum_{l\geq0}C_{k,l}\enrg_0^{(l)}$ are in the completion $\interchco$.
For the non-zero modes $k$, we divide the generators of $\interch$ by appropriate polynomials in $k$ of order $l$ and proceed similarly to establish the desired properties for the coefficients $C_{k,l}$.

Summarizing the previous arguments, we have thus established an existence of an isomorphism
\begin{equation}\label{interch-iso}
\interchco\cong\interinfinco
\end{equation}
of Lie algebras. Though the explicit transformation requires computation of the coefficients $C_{k,n}$ in the power series~\eqref{pow-ser}, it is not necessary for our purposes.

We can justify the isomorphism in~\eqref{interch-iso} by another, heuristic  but more physical, argument.
Recall  that the modes $\enrg_{k}=\enrg^{(0)}_{k}$ were obtained by expanding the
field $\enrgpl(z,\bz)$ defined in~\eqref{enrg-field-pl} on the unit
circle in the complex plane (or at  time $\tau=0$ in the cylinder geometry). The higher modes
$\enrg^{(l)}_{k}$ can be obtained by expansion at
different moments $\tau$. Indeed,
we see from~\eqref{comm-rel-enrg-Ln} that the dilatation operator
$L_0+\bar{L}_0$ generates $\enrg^{(1)}_{k}$ from $\enrg^{(0)}_{k}$, and so on.
This action can be interpreted as an expansion of $\enrgpl(z,\bz)$ on
a different (non-unit) circle with the same centre.
Transforming the unit-circle expansion by the conformal generators $L_n$
and $\bar{L}_n$ we generate all higher modes $\enrg^{(l)}_n$ and cover an expansion of the interchiral field
$\enrgpl(z,\bz)$ on the whole complex plane.
On the other hand,
this expansion is given by the quadratic monoms $\enrg_{n,m}$
in the  fermions modes. This suggest strongly  that  both   Lie algebras -- the one generated by $\enrg^{(l)}_n$ and the other, $\interinfin$,  generated by  $\enrg_{n,m}$ -- are  isomorphic, after taking proper completions.

\subsection{The interchiral algebra $\interchalg$}\label{sec:interdef}
We define {\it the interchiral algebra} for the $\gl(1|1)$ models, denoted by  $\interchalg$, as follows.
Consider the completion $\interinfinco$ of the Lie algebra $\interinfin$ generated by the modes of the regularized interchiral field  $\enrg(z,\bz)$ in its double mode expansion. The interchiral algebra  $\interchalg$ is then defined as the associative algebra generated by $\interinfinco$, {\it i.e.}, it is  the homomorphic image
of $U \interinfinco$:
\begin{equation}\label{Uinterch-hom-2}
 U\interinfinco\longrightarrow\interchalg.
\end{equation}
The defining relations in $\interinfin$ and thus in  the universal enveloping algebra $U\interinfinco$ are computed in Sec.~\ref{commrels}. The point is, there might be additional relations in the fermionic representation of the abstract algebra $U\interinfinco$, and the particular modules realized in $\cHilb$, as  also happens say for the Virasoro Lie algebra and its enveloping algebra $\Vir(2)$. This is why the map from~\eqref{Uinterch-hom-2} is not an isomorphism but a covering homomorphism of associative algebras.

 We also recall the discussion in the previous subsection about relations between both  Lie algebras $\interch$ and $\interinfin$. In view of existence of an isomorphism~\eqref{interch-iso} of their completions, we can give another, more physical, definition of the interchiral algebra.
The interchiral algebra can be considered as the associative algebra  generated by the
modes of the following {\it local} fields: the stress-energy tensor fields $T(x,\tau)$ and $\bar{T}(x,\tau)$ and the interchiral field $\enrgcyl(x,\tau)$, see the mode expansion in~\eqref{enrg-modes-exp} and in~\eqref{enrg-modes}, where $x$ and $\tau$ are the coordinates on the cylinder.  The Lie algebra generated by these  modes requires of course the completion, $\interchco$ in the sense of~\eqref{interchco-def}, in order to get indeed a definition of $\interchalg$ equivalent to the first one.
In the symplectic-fermion theory, the interchiral algebra  is  then the homomorphic image of the (representation of the) enveloping algebra $U \interchco$.  The Lie algebra $\interch$, and its completion as well, has
 $\enrg^{(l)}_n$, $\virm{l}_n$ and $\bvirm{l}_n$, with $l\geq0$ and $n\in\oZ$, as its basis. The commutation relations in this basis were computed also in Sec.~\ref{sec:comm-rel-interch}, which give  relations in $\interchalg$, but probably not all defining relations, because there should be more  in the representation of $U\interch$. We  write therefore
\begin{equation}\label{Uinterch-hom}
 U\interinfinco \cong  U\interchco\longrightarrow\interchalg
\end{equation}
where the arrow is again a covering homomorphism of associative algebras and $\interchalg$ is the image of $U\interchco$ under this homomorphism.

Note that while the second definition of the interchiral algebra is probably better for further generalizations of the concept of interchiral algebras for other models, we will use the first, more technical (and less physically pleasant), definition of the interchiral algebra in studying its simple modules below.
We note
 further
 that for more complicated logarithmic theories like those describing scaling limit of $\gl(n|n)$ ($n>1$) periodic spin-chains, we will have a different and now faithful representation of $\JTL{N}$, and a description in terms of Lie algebras will probably not be available. Nevertheless, a description in terms of a (properly defined)  interchiral algebra should still exist. This will be discussed in further work.

\subsection{OPEs}

An important advantage in using the interchiral algebra is that it gives a convenient ``vertex-operator algebra'' framework\footnote{Our interchiral algebra strictly speaking is not a vertex-operator algebra because of the neither-chiral-nor-antichiral fields. A proper generalization of the vertex-operator algebras for our context is still required.} where one introduces operator-valued generating functions of formal variables $z$ and $\bz$ and their OPE in order to define an algebraic structure.
It is worth spending some time discussing this approach in the context of the bulk theory.
We recall first the OPE of derivatives $\psi^{\alpha}(z)$ and $\bar{\psi}^{\alpha}(\bz)$ of the symplectic fermions
\begin{equation}\label{sympl-fer-OPE}
\psi^{\alpha}(z)\psi^{\beta}(w) = J^{\alpha \beta} \Bigl(\ffrac{\one}{(z-w)^2} - T(w) + \text{reg.}\Bigr),\qquad
\bar{\psi}^{\alpha}(\bz) \bar{\psi}^{\beta}(\bw) = J^{\alpha \beta}  \Bigl(\ffrac{\one}{(\bz-\bw)^2} - \bar{T}(\bw) + \text{reg.}\Bigr),
\end{equation}
where we use the symplectic form $J^{\alpha \beta}$, with $J^{12}=-J^{21}=1$, and computed the coefficient in front of the stress tensor $T(w)$ using usual conformal-invariance arguments.
The stress tensor $T(z)$ is given in fermions as $T(z)=\half J_{\alpha\beta}:\psi^{\alpha}(z)\psi^{\beta}(z):$ with $J_{\alpha\beta}$ such that $J_{\alpha\beta}J^{\beta\gamma}=\delta^{\gamma}_{\alpha}$.

Then, OPE of two $\enrg(z,\bz)$ fields can be written using~\eqref{sympl-fer-OPE}  as
\begin{multline}
\enrg(z,\bz)\enrg(w,\bw) = S_{\alpha\beta}\psi^{\alpha}(z)\bar{\psi}^{\beta}(\bz)S_{\gamma\delta}\psi^{\gamma}(w)\bar{\psi}^{\delta}(\bw)\\
= -  S_{\alpha\beta}S_{\gamma\delta} \Biggl(J^{\alpha\gamma}\Bigl(\ffrac{\one}{(z-w)^2} - T(w) + \twr \Bigr)
J^{\beta\delta}\Bigl(\ffrac{\one}{(\bz-\bw)^2} - \bar{T}(\bw) + \btwr\Bigr) \\
+ J^{\alpha\gamma}\ffrac{\one}{(z-w)^2}\bigl(:\bar{\psi}^{\beta}(\bw)\bar{\psi}^{\delta}(\bw):
+(\bz-\bw):\bder\bar{\psi}^{\beta}(\bw)\bar{\psi}^{\delta}(\bw): + \dots \bigr)\\
+ J^{\beta\delta}\ffrac{\one}{(\bz-\bw)^2}\bigl(:\psi^{\alpha}(w)\psi^{\gamma}(w):
+(z-w):\der\psi^{\alpha}(w)\psi^{\gamma}(w): + \dots \bigr) + \text{reg.}
\Biggr)
\end{multline}
where we take into account contributions of all possible single and double contractions, and denote by $\twr$ and $\btwr$ contributions of descendants from levels $(n,0)$ and $(0,n)$ with $n>2$, respectively. We note that $ S_{\alpha\beta}S_{\gamma\delta}J^{\alpha\gamma}J^{\beta\delta} =-2$ and $ S_{\alpha\beta}S_{\gamma\delta}J^{\alpha\gamma}=J_{\beta\delta}$. Therefore, the OPE is
\begin{multline}\label{interch-OPE}
\enrg(z,\bz)\enrg(w,\bw) = 2 \Bigl(\ffrac{\one}{(z-w)^2(\bz-\bw)^2} - \ffrac{2T(w)}{(\bz-\bw)^2} - \ffrac{2\bar{T}(\bw)}{(z-w)^2}
+ \ffrac{\twr}{(\bz-\bw)^2} + \ffrac{\btwr}{(z-w)^2} \Bigr)\\
- \ffrac{(z-w)\der T(w)+\dots}{(\bz-\bw)^2}
- \ffrac{(\bz-\bw)\bder \bar{T}(\bw)+\dots}{(z-w)^2} \; + \; \text{reg.}
\end{multline}
We note that the right-hand side in the relation~\eqref{enrg-enrg} is consistent with the
OPE~\eqref{interch-OPE} in a sense that commuting two $\enrg(z,\bz)$
fields we can get only the identity and a linear combination of modes of
$T(z)$ and $\bar{T}(\bz)$ and of their descendants, due to presence of
the symplectic form $ J_{\alpha\beta}$ in~\eqref{interch-comm-rel}.

The OPE of $\enrg(z,\bz)$ with $T(w)$ and $\bar{T}(\bw)$ is quite
obvious because the field $\enrg(z,\bz)$ is primary for both chiral
and anti-chiral stress tensors.
We describe below in Sec.~\ref{sec:interch-simple-2} the vacuum
module over the interchiral algebra $\interchalg$ and some other its simple modules explicitly using this ``vertex-operator'' formulation.

\subsection{The interchiral algebra in the twisted model}\label{sec:interchalg-tw}

It turns out that we have also a representation of the interchiral algebra  in the twisted sector of the
symplectic fermions,  where the fermionic  modes are half-integer~\cite{Kausch}. As we have seen in Sec.~\ref{sec:scal-lim-tw}, this sector corresponds to  the scaling limit of
the anti-periodic $\gl(1|1)$ spin-chain from
Sec.~\ref{subsec:antiper-mod}.
The symplectic-fermion representation
of the universal enveloping algebra $U\interinfin$
in  the twisted model is given by the expressions~\eqref{eq:sympl-ferm-tw-interLie} for the generators $\enrg_{m,n}$, $\strt_{m,n}$, and $\bstrt_{m,n}$.
Using the local operators or modes of the local fields $\enrgcyl(x,\tau)$ and the energy-momentum tensors in the half-integer sector, we obtain the corresponding representation of the interchiral algebra in the anti-periodic model. And once again, taking then the completion  $\overline{\ssp}_{\infty}$ of $\spinf$ as in~\eqref{interinfinco-def}, in this half-integer sector, we obtain the representation of $\interchalg$ as the representation of $U\overline{\ssp}_{\infty}$, which is also a quotient of $U\interinfinco$.

\subsection{Modules over the interchiral algebra $\interchalg$} \label{sec:simpl-JTL-Vir}

In this new  section, we describe the structure of simple and indecomposable modules over the interchiral algebra~$\interchalg$ that is defined as the representation of the completion $\interinfinco$.
Recall that we already described all simple modules $\interM{j}$ over $U\interinfin$ that appear in the whole space $\Hilb$ of the scaling states in Sec.~\ref{sec:bimod-scal-lim} and~\ref{sec:interinfin-sm},
 together with the bimodule structure of $\Hilb$ over $U\interinfin$ and the centralizer $\LQGodd$.
 Here, we give  in Sec.~\ref{sec:interch-simple} an explicit construction of the module $\interM{1}$ and then description of the vacuum module over the interchiral algebra $\interchalg$, and the twisted vacuum module in Sec.~\ref{sec:vac-mod-tw}.  Finally,  indecomposable modules over the interchiral algebra are analyzed in Sec.~\ref{sec:indecompmod}.


\subsubsection{The module $\interM{1}$}\label{sec:interch-simple}
Recall that $\interM{j}$ denotes the module over
$U\interinfin$ which is obtained in the scaling (direct) limit of the JTL simple modules $\AIrrTL{j}{(-1)^{j+1}}$. Suppose we wish to  check directly that $\interM{j}$ are simple  modules over the algebra $U\interinfin$ and that their completions $\interMco{j}$ in the sense of the formal topology, see~\eqref{eq:cHilb}, are simple modules over~$\interchalg$. Our strategy consists in a few steps: (i) we
state that the  $\strt_{m,n}$  and $\bstrt_{m,n}$ basis elements of $\interinfin$ generate the algebras of endomorphims of simple modules (as graded vector spaces) over the left and right Virasoro algebras, respectively; (ii) we then recall that the generators  $\enrg_{m,n}$ are modes of the interchiral field $\enrg(z,\bz)$ of dimension $(1,1)$  and (iii) finally we use the operator-state  correspondence in CFT and identify  composite fields in $\enrg(z,\bz)$ and its derivatives  that produce the direct summands in~\eqref{JTL-simple-Vir-sum} for a few explicit examples.

In our analysis, we use the $s\ell(2)$ action or ``global $SU(2)$ symmetry'' of the symplectic fermions commonly known as Kausch's $s\ell(2)$ action~\cite{Kausch} which we denote by  $\nSU(2)$. The generators of the  $\nSU(2)$ in the non-chiral symplectic-fermion theory are
\begin{equation}\label{eq:Kausch-SU}
Q^a=d^a_{\alpha\beta}\left\{i\phi^\alpha_0\psi^\beta_0+\sum_{n=1}^\infty
\left(\ffrac{\psi^\alpha_{-n}\psi^\beta_n}{n} + \ffrac{\bar{\psi}^\alpha_{-n}\bar{\psi}^\beta_n}{n}\right)\right\}
\end{equation}
where $d^a_{\alpha\beta}$ are defined in~\eqref{da-mat}, and with $[Q^a,Q^b]=f^{ab}_cQ^c$ and $f^{01}_2=-1$.
We note  that the generators  $\strt_{m,n}$ and  $\bstrt_{m,n}$ commute with this $\nSU(2)$ action. Meanwhile, it is  important to note   that the generators $\enrg_{m,n}$ \textit{do not} commute with the $\nSU(2)$ and belong to the zero isospin-projection in the triplet $\nSU(2)$-module.

First, we can consider chiral and anti-chiral sectors separately to show
that simples $\VX_{j,1}$ over the Virasoro Lie algebra with $c=-2$ are also simples over the Lie algebra of $\strt_{m,n}$'s, and similarly for the anti-chiral part. To analyze the chiral sector, it is technically easier to work in a smaller space -- in the chiral symplectic-fermions theory. Recall
that the scaling limit of the open $\gl(1|1)$ spin-chains described in~\cite{ReadSaleur07-2} and in Sec.~\ref{subsec:bimodopen} gives this chiral  LCFT. The symmetry algebra of this theory, which is now the centralizer of the Virasoro algebra $\Vir(2)$, is a representation of the full quantum group
 $\LQG$ at $\q=i$. We recall that this centralizer can be equivalently described as the semi-direct product of the $\gl(1|1)$ and
 the enveloping algebra $U s\ell(2)$. The generators of the $U s\ell(2)$ are the divided powers of the quantum-group generators and their action coincides with the  action of Kausch's global $s\ell(2)$,  while the generators $\psi^{1,2}_0$ of $\gl(1|1)$ correspond to horizontal arrows in the bimodule diagram in Fig.~\ref{openbimodule-cont}.
  On the other hand, the centralizer of $\LQGi$ is an associative algebra of endomorphisms of graded vector spaces $\Endo(\Hilb)=\bigoplus_{n,m\in\oN}\Hom(\Hilb^{(n)},\Hilb^{(m)})$, where $\Hilb^{(n)}$ are finite-dimensional homogeneous subspaces of the $n$th energy level, and such that they commute with $\LQGi$. This centralizer obviously contains  Virasoro $\Vir(2)$ but it is a bigger algebra -- a completion\footnote{The completion of the Virasoro contains in particular
  projectors onto a fixed, $k$th, energy subspace. Such operators are obviously in the centralizer of $\LQGi$ but they are not in the universal enveloping of the Virasoro Lie algebra.  Of course, there is no double-centralizing property in infinite-dimensional spaces and this is why an appropriate completion is necessary for describing centralizers of  quantum groups in the CFTs.}
of $\Vir(2)$.

The chiral space  $\Hilb$  is the direct sum $\WP^{+} \oplus\WP^{-}$ of bosonic and fermionic degrees of freedom, and it is  decomposed
 onto modules over the  product $U \SU(2)\boxtimes\Vir(2)$  as in~\eqref{W-proj-ch-decomp}.  The full bimodule structure is given in Fig.~\ref{openbimodule-cont}.
 We note  that the generators $\strt_{m,n}$ are well-defined operators from $\Endo(\Hilb)$ just introduced  and they commute with the $\LQGi$ action. They are therefore in the centralizer of $\LQGi$. The socle of $\Hilb$ is the intersection of the kernels of the fermionic generators of the $\gl(1|1)$ part of $\LQGi$. We can thus restrict the action of all $\strt_{m,n}$'s onto the socle of the modules in~\eqref{W-proj-ch-decomp}. Using then the $U s\ell(2)$ symmetry, we obtain that simple modules over the Lie algebra generated by $\strt_{m,n}$'s are the same, as graded vector spaces, as the simple modules over the Virasoro. Indeed, the vacuum modules $\VX_{1,1}$ appears with multiplicity one and it is the $U s\ell(2)$ invariant but all $\strt_{m,n}$'s are also $Us\ell(2)$ invariants and they thus generate the same $\VX_{1,1}$ from the vacuum state. Further, we take the highest-weight state of each higher multiplet for the $Us\ell(2)$ and see that all $T_{m,n}$'s generate from it a  module identified with (isomorphic as a graded vector space to) the simple Virasoro module $\VX_{j,1}$. This proves our claim that all $T_{m,n}$, for $m,n\in\oZ$, are graded-vector spaces endomorphisms of the simple Virasoro modules $\VX_{j,1}$, for $j\geq1$.
  Moreover, we also check  that the fermion bilinears $\strt_{m,n}$ are the only bilinears that commute with $\LQGi$.  They are therefore generators of the centralizer of $\LQGi$.
 This means that simple modules over Virasoro are also simple modules over the Lie algebra with the basis given by $\strt_{m,n}$.
 We use the same arguments to show that   the Lie algebra with the basis given by $\bstrt_{m,n}$
 has simple modules $\bar{\VX}_{j,1}$.

In the full non-chiral theory, the left (or right) Virasoro algebra has the same isomorphism classes of simple and indecomposable modules as in the chiral theory, see a discussion in Sec.~\ref{sec:non-chiral-sympl-fer} below.
We thus have essentially the same representation of $\strt_{m,n}$ -- the same expression for the generators as in the chiral theory but acting in a much bigger space, so a difference is only  in `multiplicities' which are infinite now and correspond to modules over the right Virasoro $\overline{\Vir}(2)$. Summarizing,
we give the following statement. The simple modules $\VX_{j,1}\boxtimes\bar{\VX}_{k,1}$ over the left-right Virasoro algebra $\VirN(2)$   are also simple modules over the Lie algebra with the basis given by  $\strt_{m,n}$ and  $\bstrt_{m,n}$, with $m,n\in\oZ$.

\medskip

We then continue by analyzing the vacuum module  over $\interchalg$ which  contains the identity field. By definition, the vacuum module is a vector space generated from one state -- the vacuum $\vac$ -- by those  $\enrg_{n,m}$, $\strt_{n,m}$ and $\bstrt_{n,m}$  that have  negative indexes $n$ and $m$ (note that we used here our first definition of the interchiral algebra from Sec.~\ref{sec:interdef}).
It turns out that this vector space as a $\interinfin$-module is decomposed over the left-right Virasoro exactly like $\interM{1}$. We stated just above that the  $\strt_{n,m}$ and $\bstrt_{n,m}$ generators  of $\interinfin$ generate the endomorphism algebras of a simple Virasoro module $\VX_{j,1}$ or $\bar{\VX}_{j,1}$, respectively.
Therefore,
to extract the Virasoro content of the vacuum module and compare it with $\interM{1}$ we only need to identify highest-weight vectors for $\VirN(2)$ which are generated from the vacuum state by  the $\enrg_{n,m}$ basis elements. Equivalently, using the operator-state  correspondence, we are going to identify left-right Virasoro primary fields  generated from the identity   by a field having these  $\enrg_{n,m}$ as its modes. Such a field was identified in Sec.~\ref{sec:interch-field} with
the interchiral field $\enrg(z,\bz)=S_{\alpha\beta}\psi^{\alpha}(z)\bar{\psi}^{\beta}(\bz)$.
The  $\enrg(z,\bz)$  is a primary field and corresponds to the state
\begin{equation*}
|S\rangle = \lim_{z,\bz\to0}\enrg(z,\bz)\vac
\end{equation*}
which belongs to the subspace $\VX_{2,1}\boxtimes\bar{\VX}_{2,1}$.  Note that this subspace is in the triplet sector with respect to $\nSU(2)$. Then,
$\strt_{n,m}$ and  $\bstrt_{n,m}$  generate this $\VirN(2)$-module from the highest-weight state $|S\rangle$. We thus obtained first two terms in the   decomposition onto  left-right Virasoro modules
\begin{equation}\label{interch-vac-mod}
\interM{1}|_{\VirN(2)}  =  \bigoplus_{j\geq1} \VX_{j,1}\boxtimes\bar{\VX}_{j,1},
\end{equation}
which is the direct sum in~\eqref{JTL-simple-Vir-sum} for $j=1$.
Highest-weight vectors in other direct summands in~\eqref{interch-vac-mod} can be constructed by taking appropriate composite fields in $\enrg(z,\bz)$ and in its derivatives $\der^j\bar{\der}^j\enrg(z,\bz)$ applied to the vacuum $\vac$ in the limit $z,\bz\to0$. So, the next primary field from the decomposition~\eqref{interch-vac-mod} should have conformal dimensions $(3,3)$ and it might be identified with the composite field  $\nord{\der\bar{\der}\enrg(z,\bz)\enrg(z,\bz)}$, up to descendants of $\enrg(z,\bz)$ and of the identity field on the level $(3,3)$. We note that this highest-weight  state belongs to a $5$-dimensional $\nSU(2)$-module.  This result can be obtained by a direct calculation using the double mode expansion~\eqref{enrgpl-mod-exp} together with normal-ordering prescriptions or using OPE formulas given in~\eqref{interch-OPE}.

 The analysis can be continued to construct in a similar way states belonging to higher $\nSU(2)$-multiplets and contributing thus to new primary fields. Using~\eqref{eq:Kausch-SU}, we obtain that the composite field
 \begin{equation}\label{comp-f}
 \nord{\prod_{j=0}^n\der^j\bar{\der}^j\enrg(z,\bz)}
 \end{equation}
  has the conformal weight $(\Delta_{n+2,1},\Delta_{n+2,1})$ and  it is in a $(n+1)$-dimensional $\nSU(2)$-module, up to contributions from lower $\nSU(2)$-multiplets\footnote{Strictly speaking, the composite field constructed belongs actually to a direct sum $\oplus_{j=0}^{n/2}\SUrep{j}$ where $\SUrep{j}$ denotes the $(2j+1)$-dimensional $\nSU(2)$-module.}. This state therefore belongs  to the direct sum $\oplus_{j=1}^{n+2} \VX_{j,1}\boxtimes\bar{\VX}_{j,1}$. By induction, this finally gives the decomposition~\eqref{interch-vac-mod}.

We next observe that the state
\begin{equation*}
|S^2\rangle = \lim_{z,\bz\to0}\nord{\enrg^2(z,\bz)}\vac
\end{equation*}
 is on the level $(2,2)$ of the direct sum $\VX_{2,1}\boxtimes\bar{\VX}_{2,1}\oplus \VX_{1,1}\boxtimes\bar{\VX}_{1,1}$. This means that the square of the interchiral field maps back to the identity module over $\VirN(2)$. A similar analysis for states~\eqref{comp-f} living in higher $\nSU(2)$ multiplets shows that our vector space $\interM{1}$ generated by composite fields in $\enrg(z,\bz)$ and its derivatives is indeed a simple module over $\interinfin$.

\subsubsection{Vacuum module over $\interchalg$}\label{sec:interch-simple-2}
Recall then our discussion about the isomorphism of two different Lie algebras, $\interinfinco$ and $\interchco$, given in Sec.~\ref{sec:interch-interinfin}. The completed algebras, as well as the interchiral algebra $\interchalg$, act on the completed space $\cHilb$ defined in~\eqref{eq:cHilb}.
It is natural then to take the completion of  $\interM{1}=\oplus_n \interM{1}^{(n)}$ as the direct product $\interMco{1}=\prod_n \interM{1}^{(n)}$ of the eigenspaces of $E=L_0+\bar{L}_0$ and declare it as \textit{the vacuum module of the interchiral algebra}. Note that this module has the  basis of $\interM{1}$ and the latter is its dense subspace in the formal topology, see Sec.~\ref{sec:interch-interinfin}.
The first question is now whether the vacuum module $\interMco{1}$ is simple or not: making the completion we add many new vectors which could generate invariant subspaces.
However, it turns out that having $v\in\interMco{1}$ written as an infinite combination of the basis elements, i.e., if $v$ is not in the subspace $\interM{1}$ we can find a word in the generators of $\interchalg$ such that the image of $v$ under the action of the word belongs again to $\interM{1}$. For example, for
\begin{equation*}
v=\enrg^{(0)}_0\vac = S_{\alpha,\beta}\sum_{m<0}\sferm^{\alpha}_m\bsferm^{\beta}_m\vac\in\cHilb
\end{equation*}
we can write its image under $(L_0-2)(\bar{L}_0-2)-L_{-2}\bar{L}_{-2}$ as
\begin{equation*}
\bigl((L_0-2)(\bar{L}_0-2)-L_{-2}\bar{L}_{-2}\bigr)v=S_{\alpha,\beta}\sferm^{\alpha}_{-1}\bsferm^{\beta}_{-1}\vac\in\Hilb
\end{equation*}
and the image is just the Virasoro highest-weight state corresponding to the primary field $\enrg(z,\bz)$.


 In general, we formulate the following conjecture based on many explicit checks.
 \begin{conj}\label{conj:simples}
 Taking the completions $\interMco{j}$ (those in the formal topology) of the simple $U\interinfin$-modules $\interM{j}$
  gives simple modules over the interchiral algebra $\interchalg$.
\end{conj}

\subsubsection{The twisted vacuum module}\label{sec:vac-mod-tw}
We finally discuss an example of a simple module over the interchiral algebra $\interchalg$ in the sector with  half-integer fermionic modes.
Recall that we have found in Sec.~\ref{sec:interchalg-tw} an action of the interchiral algebra~$\interchalg$ in the twisted model of symplectic fermions. The full symmetry algebra in this case is the $U s\ell(2)$ as it was stated in Prop.~\ref{prop:interch-centr-tw}. We define then the twisted vacuum module over the interchiral algebra as the space of $s\ell(2)$ invariants.
 The structure of the twisted vacuum
module  can be described in a  way parallel to Sec.~\ref{sec:interinfin-sm} and~\ref{sec:interch-simple}. The generating function of levels obtained as in~\cite{ReadSaleur01}
now involves the characters
\begin{equation}
\chi_{j,2}={q^{(j-1)^2/2}-q^{(j+1)^2/2}\over\eta(q)}
   \end{equation}
   and we have the identity for the Kac character, see App.~\AppChar,
   \begin{equation}
   K_{r,2j+2}=\sum_{s=0}^j \chi_{r-j+2s,2}.
   \end{equation}
   This time, the simple modules of $\JTL{N}^{tw}$ are obtained with  a
   single subtraction according to~\eqref{dimIrrATL-anti}, so their
   left-right Virasoro $\VirN(2)$ content is
   \begin{equation}
   F_{j,(-1)^j}^{(0)}=F_{j,(-1)^j}- F_{j+2,(-1)^j}=\sum_{r=1}^\infty \chi_{r,2}\left(\overline{\chi}_{r-j,2}+\overline{\chi}_{r-j+2,2}+\ldots+\overline{\chi}_{r+j,2}\right).
   \end{equation}
 The  $\JTL{N}^{tw}$ modules are all semi-simple in this model, and so are the
 $\VirN(2)$  modules, according to the decomposition~\eqref{antip-bimod}.
We  thus obtain the left-right Virasoro content in the scaling limit
\begin{equation}\label{JTL-simple-Vir-sum-tw}
\AIrrTL{j}{(-1)^{j}} \mapsto \bigoplus_{r\geq1} \VX_{r,2}\boxtimes\bigl(\bar{\VX}_{r-j,2}\oplus\bar{\VX}_{r-j+2,2}\oplus\ldots\oplus\bar{\VX}_{r+j,2}\bigr),
\end{equation}
which is again a direct sum of infinite number of simple modules over $\VirN(2)$.

Let $\interM{j}^{tw}$ denotes the module over
$U\interinfin$ (or $U\spinf$) which is obtained in the scaling limit of the $\JTL{}^{tw}$ simple modules $\AIrrTL{j}{(-1)^{j}}$. With an analysis similar to one in Sec.~\ref{sec:interch-simple}, we generate the vacuum module applying by composite fields in $\enrg(z,\bz)$ and its derivatives on the vacuum state, the highest-weight state in $\VX_{1,2}\boxtimes\bar{\VX}_{1,2}$.  The vacuum module $\interM{0}^{tw}$
 decomposition onto $\VirN(2)$-modules is
\begin{equation}\label{interch-vac-mod-tw}
\interM{0}^{tw}|_{\VirN(2)}  =  \bigoplus_{j\geq1} \VX_{j,2}\boxtimes\bar{\VX}_{j,2}
\end{equation}
and it is an irreducible module over $U\spinf$.
We do not give details of calculations as their essentially repeat the previous. Correspondingly, its completion, see a discussion in Sec.~\ref{sec:interch-simple-2}, describes the vacuum module over $\interchalg$ in the twisted sector.

\subsection{Indecomposable modules over $\interchalg$}\label{sec:indecompmod}

To proceed, we must now discuss the scaling limit of the
indecomposable $\JTL{N}$-modules $\APrTL{j}$, which are indecomposable $U\interinfin$-modules denoted by $\interP{j}$, and compare it with
potentially similar structure in the symplectic fermion theory. This
will be facilitated by a preliminary discussion of the latter.

\subsubsection{The structure of  the Virasoro representations in symplectic fermions}\label{sec:non-chiral-sympl-fer}

We already reminded the reader of the  known results about the chiral conformal
field theory of symplectic fermions in
Sec.~\ref{subsec:bimodopen}. We now turn to a similar analysis of the non-chiral theory and its  decomposition  over the left-right Virasoro $\VirN(2)$.
 In this section, we will denote our space of scaling states in the bulk theory by $\bimnch$ in order to not be confused with the decomposition~\eqref{W-proj-ch-decomp} of the chiral theory.

In the integer-mode sector, the space of states $\bimnch$ for the non chiral theory  is  decomposed
into a bosonic sector $\bimnbos$ and the fermionic one $\bimnfer$ with
the $\nSU(2)$ and left-right Virasoro $\VirN(2)$ content~\cite{tobepublished}
\begin{equation}\label{nsf-sl-Vir}
\bimnbos|_{\nSU(2)\boxtimes\VirN(2)} =
\bigoplus_{k\in\oN_0}\SUrep{k}\boxtimes \bimnbos_k, \qquad \bimnfer|_{\nSU(2)\boxtimes\VirN(2)} =
\bigoplus_{k\in\oN-\half}\SUrep{k}\boxtimes \bimnfer_k,
\end{equation}
 where $\SUrep{n}$ denotes now a $(2n+1)$-dimensional $\nSU(2)$-module of the
isospin $n$, and  $\bimn_k$ are $\VirN(2)$-modules which we describe
below.

With the left and right $\Vir(2)$-modules $\bimnl_k$ and
$\bimnr_k$ denoting restriction of $\bimn_k$ on $\Vir(2)$ and
$\overline{\Vir}(2)$, respectively, the bosonic components
$\bimnbos_k$ have the decompositions, with non-negative integer $k$,
\begin{eqnarray}
\bimnbosl_k =
\bigoplus_{a=-k}^k\;\bigoplus_{n> k-a}\VP_{n,1}\tensore\bar{\VX}_{n+2a,1},\qquad
\bimnbosr_k =
\bigoplus_{a=-k}^k\;\bigoplus_{n>
  k-a}\VX_{n,1}\tensore\bar{\VP}_{n+2a,1},\qquad k\in\oN_0,\label{nsf-bos-sl-lrVir-dec2}
\end{eqnarray}
and the fermionic components $\bimnfer_k$ decompose as, with a
positive half-integer $k$,
\begin{equation}
\bimnferl_k =
\bigoplus_{a=-k}^k\;\bigoplus_{n>
  k-a}\VP_{n,1}\tensore\bar{\VX}_{n+2a,1},\qquad
\bimnferr_k =
\bigoplus_{a=-k}^k\;\bigoplus_{n>
  k-a}\VX_{n,1}\tensore\bar{\VP}_{n+2a,1},\qquad k\in\oN-\half,\label{nsf-fer-sl-lrVir-dec2}
\end{equation}
where we use the tensor product $\tensore$ for two Virasoro
modules in order to show that `multiplicities'  of left-
or right-Virasoro staggered modules in the decompositions are simple subquotients over right
or left Virasoro, respectively. We stress that direct summands
in~\eqref{nsf-bos-sl-lrVir-dec2} and~\eqref{nsf-fer-sl-lrVir-dec2} are
not $\VirN(2)$-modules and the rest of this section is devoted to
describing a subquotient structure for the product of the two Virasoro algebras.

The sector with the trivial $\nSU(2)$ action ($k=0$) has a particular
interest because it contains the vacuum state $\vac$ and its
logarithmic partner $\lvac$ ($L_0\lvac = \bar{L}_0\lvac = \vac$),
\begin{equation}\label{nsf-vac-lrVir}
\bimnbosl_0 =  \bigoplus_{n\geq
  1}\VP_{n,1}\tensore\bar{\VX}_{n,1},
  \qquad
  \bimnbosr_0 =  \bigoplus_{n\geq
  1}\VX_{n,1}\tensore\bar{\VP}_{n,1}, \qquad\vac,\lvac\in\VP_{1,1}\tensore\bar{\VX}_{1,1}.
\end{equation}
These two states are the only bosonic states with conformal dimension
$(0,0)$. There are two other states of the same conformal dimension in
the fermionic sector $\bimnfer_{1/2}$ in agreement with~\cite{GK}.

We use the decompositions~\eqref{nsf-bos-sl-lrVir-dec2}
and~\eqref{nsf-fer-sl-lrVir-dec2} over the left and right Virasoro in
studying the subquotient structure over their product $\VirN(2)$.  We
first note an obvious fact that right-Virasoro algebra elements are
represented on each $\bimn_k$ as  intertwining operators for the left-Virasoro action,
{\it i.e.}, there is a homomorphism from the universal enveloping algebra $U
\bigl(\overline{\Vir}(2)\bigr)$ to $\Endo_{\Vir(2)}(\bimn_k)$. Therefore, we should begin with a  description of the space
$\HomVir(\VP_{n,1},\VP_{n',1})$. From  the subquotient structure~\eqref{sf-chiral-stagg-pic} of
the staggered modules, we deduce that the $\Hom$-space is one-dimensional
only for $n=n'$ or $n=n'\pm1$ and zero-dimensional otherwise. In the case
$n=n'$, the image of a basis element in the $\Hom$-space is the simple
submodule $\VX_{n,1}$ while the image in the case $n=n'+1$ is the indecomposable Kac
module $\VX_{n,1}\to\VX_{n+1,1}$ and the image in the case $n=n'-1$ is the contragredient Kac
module $\VX_{n,1}\to\VX_{n-1,1}$. Ananlysing then all possible
endomorphisms on $\bimnl_k$ respecting the left Virasoro and combining with the
decomposition of $\bimnr_k$ over the right Virasoro, we end up in
diagrams depicting a subquotient structure over the
product $\VirN(2)$ of two Virasoro algebras for each~$\bimn_k$.

\begin{figure}\centering
{\scriptsize
  \def\svgwidth{450pt}
    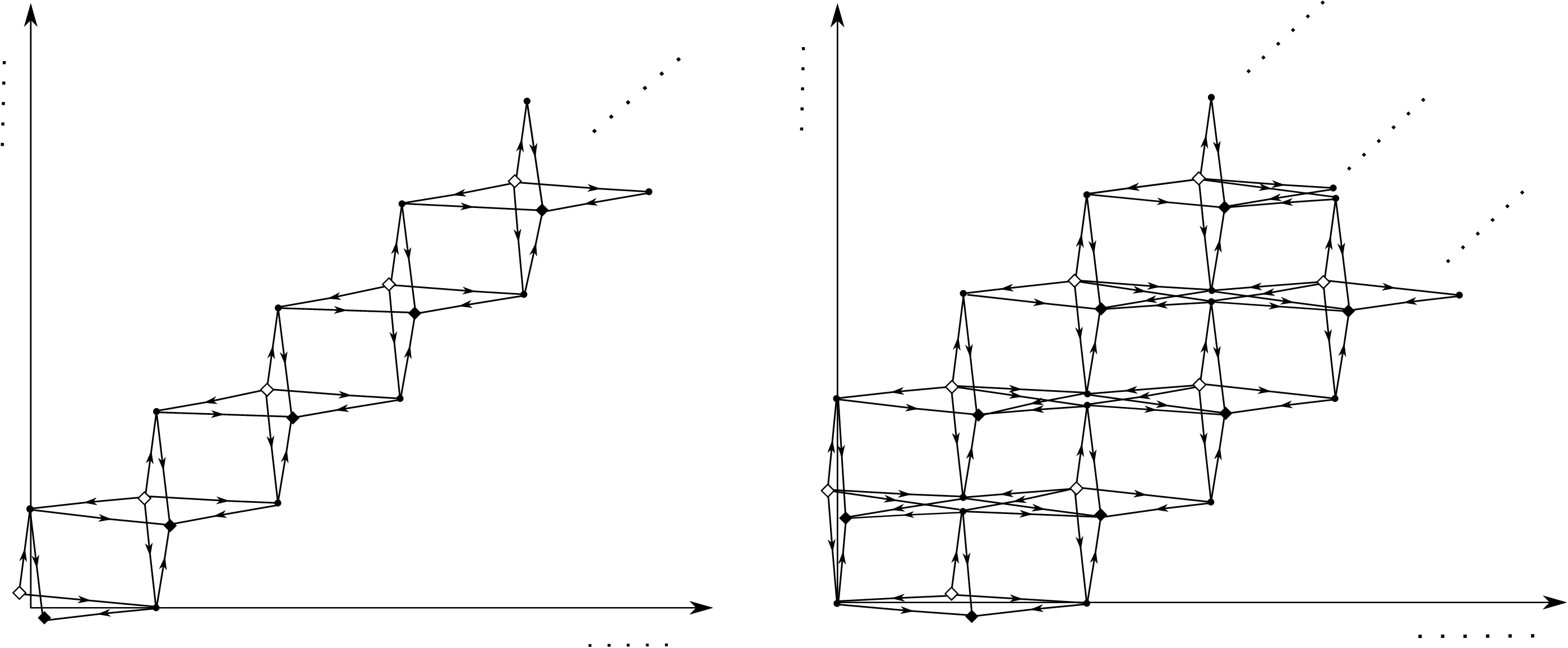
}
    \caption{Module structure over $\VirN(2)=\Vir(2)\boxtimes\overline{\Vir}(2)$
    for the vacuum sector $\bimnbos_0$ (with zero $\nSU(2)$-isospin) on the left
    diagram while the right one is for the doublet-sector $\bimnfer_{1/2}$.  Each node with a coordinate
     $(\bar{n},n')$ is a simple subquotient over
     $\VirN(2)$ with the conformal
     weight $(\Delta_{n',1},\bar{\Delta}_{n,1})$. Vertical arrows represent
    the action of the left Virasoro $\Vir(2)$ and horizontal arrows of the
    right Virasoro $\overline{\Vir}(2)$.}
    \label{sfnonchiralvirvac}
    \end{figure}
\begin{figure}\centering
{\small
\mbox{}\bigskip
  \def\svgwidth{350pt}
    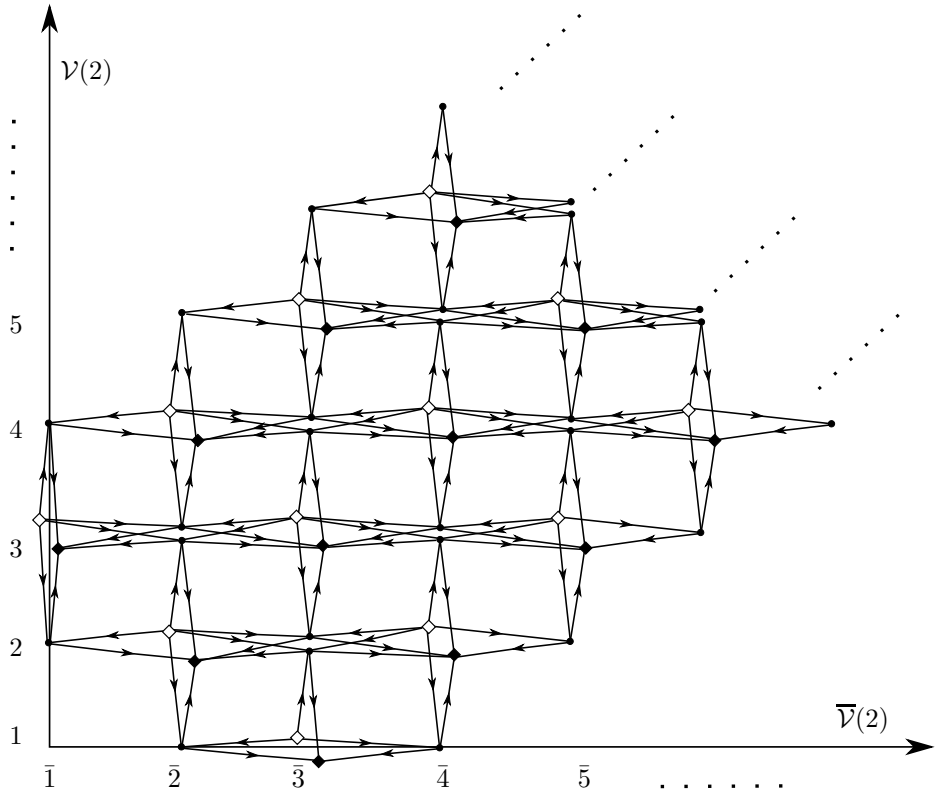
}
   \caption{$\VirN(2)$-module subquotient structure for the
     triplet-sector $\bimnbos_1$.  Each node with a coordinate
     $(\bar{n},n')$ is a simple subquotient over
     $\VirN(2)=\Vir(2)\boxtimes\overline{\Vir}(2)$ with the conformal
     weight $(\Delta_{n',1},\bar{\Delta}_{n,1})$. Vertical arrows
     represent the action of the left Virasoro $\Vir(2)$ and
     horizontal arrows of the right Virasoro $\overline{\Vir}(2)$.
   }
   \label{sfnonchiralVirtripl}
    \end{figure}

The simplest is for $k=0$, where we obtain the structure represented
 on the left in Fig.~\ref{sfnonchiralvirvac} using the decompositions~\eqref{nsf-vac-lrVir}. On the right in
Fig.~\ref{sfnonchiralvirvac}, we show the subquotient structure for
$k=1/2$ as well, which is in the fermionic sector $\bimnfer$.
For $k=1$ meanwhile we have the structure represented on
Fig.~\ref{sfnonchiralVirtripl}. More detailed analysis of the
$\VirN(2)$-module structure on each $\bimnch^{\pm}_k$  can be also found in a companion
paper~\cite{tobepublished}.

Each node with a coordinate $(\bar{n},n')$ in the diagrams on
Fig.~\ref{sfnonchiralvirvac} and Fig.~\ref{sfnonchiralVirtripl}
corresponds to a simple subquotient over
$\VirN(2)=\Vir(2)\boxtimes\overline{\Vir}(2)$ with the lowest
conformal weight $(\Delta_{n',1},\bar{\Delta}_{n,1})$ and arrows show
the action of both Virasoro algebras -- the Virasoro $\Vir(2)$ acts in
the vertical direction (preserving the coordinate $\bar{n}$), while
the right Virasoro $\overline{\Vir}(2)$ acts in the horizontal
way. Some values $(\bar{n},n')$ occur twice and those
nodes/subquotients are separated slightly for clarity; we denote top
subquotients by~$\diamond$, bottom ones by~{\scriptsize$\vardiamond$},
and subquotients in the middle level that have incoming as well as
outgoing arrows are denoted by~$\bullet$, in order to make reading the
diagrams easier. We also note that some horizontal arrows (of the same
direction) connecting a~$\diamond$ with a pair of~$\bullet$'s at the
same coordinate or a pair of~$\bullet$'s with
a~{\scriptsize$\vardiamond$}, say connecting $(\bar{1},3)$ with two
nodes at $(\bar{2},3)$ in Fig.~\ref{sfnonchiralVirtripl}, are actually
doubled which means that right-Virasoro elements corresponding to the
pair of arrows with the same source/sink map to/from a fixed linear
combination of the pair of subquotients depicted by~$\bullet$. We use
such a linear combination for each pair of~$\bullet$'s that
corresponds to a basis adapted for the left-Virasoro action\footnote{Of course, it is a matter of a
  convention, and we could choose a basis where horizontal arrows are
  not doubled but the vertical ones would be doubled. It is important to note that there exists no a basis without the doubled arrows in the corresponding diagram.} that has
no (vertical) doubled arrows and is explicitly decomposed onto
diagrams from~\eqref{sf-chiral-stagg-pic}. We finally assume in our
diagrams that right-Virasoro elements map states from two neighbour
$\diamond$'s, say $(\bar{1},3)$ and $(\bar{3},3)$ in
Fig.~\ref{sfnonchiralVirtripl}, to linearly independent subquotients
of the same coordinate, say $(\bar{2},3)$, in the middle level.

The indecomposable staggered $\Vir(2)$-modules
$\VP_{n',1}$ introduced in~\eqref{sf-chiral-stagg-pic} and appearing
in the decompositions~\eqref{nsf-bos-sl-lrVir-dec2}
and~\eqref{nsf-fer-sl-lrVir-dec2}  can be recovered by ignoring all the
horizontal arrows, while staggered $\overline{\Vir}(2)$-modules
$\bar{\VP}_{n,1}$
 are obtained by ignoring all the vertical arrows
in the  diagrams and taking an appropriate linear combination of two
simple subquotents for each
pair of middle-level nodes denoted by~$\bullet$ and sharing the same coordinate  $(\bar{n},n')$.

We note that the interchiral field $\enrg(z,\bz)$ from~\eqref{enrg-field-pl} generating our interchiral algebra $\interchalg$ belongs to the submodule identified in the diagram for the triplet sector in Fig.~\ref {sfnonchiralVirtripl} with the node {\scriptsize$\vardiamond$} at the position~$(\bar{2},2)$.

We will comment on the similarity and consistency of these figures with our results
of the (limits of) $\JTL{N}$-modules in the next subsection. Notice that the way the left and right
indecomposables are glued together in (the vacuum sector of) the non-chiral theory is similar
that what is observed for super WZW models on $\gl(1|1)$ and $\mathfrak{su}(2|1)$~\cite{SaleurSchomerus}.

Finally, we stress that the symplectic fermion theory also admits
action of $\gl(1|1)$ which connects in particular the bosonic
$\bimnbos$ and fermionic sectors $\bimnfer$ of the theory. Its action
is quite straightforward, and we refrain from discussing this for
simplicity.

\subsubsection{Indecomposable  $U\interinfin$ and $\VirN(2)$-modules}\label{sec:indecomp-Smod}
Recall that we denote the scaling limit of indecomposable
spin-chain modules $\APrTL{j}$ by $\interP{j}$. These are indecomposable but reducible modules over $U\interinfin$, recall the discussion in Sec.~\ref{sec:bimod-scal-lim}, with subquotient
structure given by diagrams in Fig.~\ref{ind-chain-mod-fig} which is an infinite analogue
of the finite ladders in Fig.~\ref{FF-JTL-mod}.
As we discussed in previous sections, the interchiral algebra $\interchalg$ acts on the completed space $\cHilb$. From its definition in~\eqref{eq:cHilb}, we obviously have a decomposition over~$\interchalg$:
\begin{equation}\label{eq:cHilb-decomp}
\cHilb=\bigoplus_{j\in\oZ} \interPco{j},
\end{equation}
where each direct summand is the completion (as defined in Sec.~\ref{sec:interch-interinfin}) of the indecomposable $U\interinfin$-modules  $\interP{j}$.
The main objective of this section is to show that the structure of $\VirN(2)$-modules
  described above is consistent with the diagrams in Fig.~\ref{ind-chain-mod-fig}.
Comparing the structure of modules over $U\interinfin$ and $\VirN(2)$, which is a subalgebra in $U\interinfinco$,  we then make conclusions on the structure of indecomposable modules $\interPco{j}$
over the interchiral algebra $\interchalg$.

  Having obtained in Sec.~\ref{sec:interinfin-sm} the content under $\VirN(2)$ of the simple
subquotients, we study a filtration\footnote{Recall that a
  filtration of $A$-module $M$ by its submodules $M_i$, with $0\leq
  i\leq n$, is called a sequence
  of embeddings $0=M_0\subset M_{1}\subset\dots \subset M_i
  \subset\dots \subset M_{n-1}\subset M_n=M$.} of $\interP{j}$ by $\VirN(2)$-modules described in
Sec.~\ref{sec:non-chiral-sympl-fer}.
Each $\interP{j}$ as a $\VirN(2)$-module
has the same as in Fig.~\ref{ind-chain-mod-fig} subquotient
structure, where a crucial point
 is that a node stands now for an
infinite direct sum of simples over $\VirN(2)$, \textit{i.e}, it is a decomposable $\VirN(2)$-module.
 \begin{figure}\centering
   \includegraphics[scale=0.9]{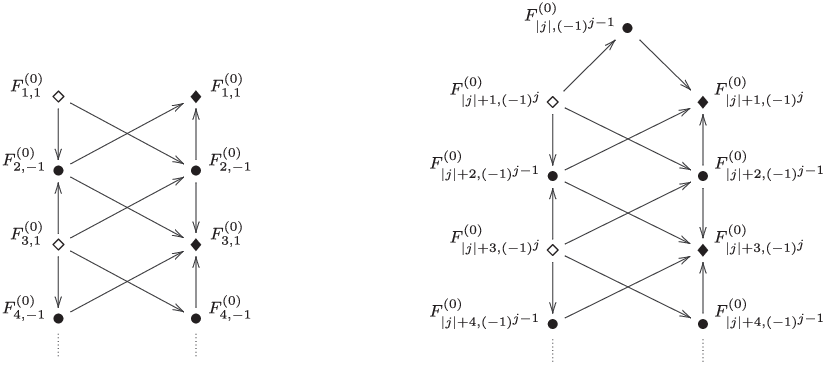}
      \caption{\sscal limit of indecomposable spin-chain modules
      $\APrTL{j}$  for $j=0$ on the left and $j\ne0$
      on the right side. These are indecomposable $U\interinfin$-modules $\interP{j}$ with simple subquotients and the Virasoro character $F^{(0)}_{|j|,(-1)^{j+1}}$ of each subquotient is given by~\eqref{Fs-Vir-content-gen}.}
    \label{ind-chain-mod-fig}
    \end{figure}
  This subquotient structure agrees
with the final results of Sec.~\ref{sec:non-chiral-sympl-fer}.
 Indeed, the $\VirN(2)$-modules~\eqref{nsf-sl-Vir} $\bimnch|_{\h=j} \equiv
\bimnch_{(j)}=\bigoplus_{k\geq |j|}\bimnbos_k$ in each sector $\h=j$
have a filtration
 \begin{equation}\label{VirN-filtr}
 \Soc\bigl(\bimnch_{(j)}\bigr)\subset\Mid\bigl(\bimnch_{(j)}\bigr)\subset \bimnch_{(j)}\to\Top\bigl(\bimnch_{(j)}\bigr)\to0
 \end{equation}
by $\VirN(2)$-submodules (which are described below)  consistent with
  the structure of $\interP{j}$
  presented on Fig.~\ref{ind-chain-mod-fig}, where each
  node is described by formulas in~\eqref{Fs-Vir-content-gen}
  and~\eqref{Fs-Vir-content}. In
  particular, for $j=0$, we obtain  from decompositions~\eqref{nsf-bos-sl-lrVir-dec2} and~\eqref{nsf-fer-sl-lrVir-dec2} the left Virasoro structure on $\bimnch$, forgetting about the $\nSU(2)$ content for a while,
\begin{equation}\label{nsf-bos-Sz-zero}
\bimnch_{(0)}^{(l)} =
\bigoplus_{n,m\geq0}\bigl(2\min(n,m)+1\bigr)\VP_{2n+1,1}\boxtimes\bar{\VX}_{2m+1,1} \; \oplus \;
\bigoplus_{n,m\geq1}2\min(n,m)\VP_{2n,1}\boxtimes\bar{\VX}_{2m,1},
\end{equation}
and the right structure $\bimnch^{(r)}$ at the grade $\h=0$ is given
by the substitutions $\VP\to\VX$ and $\bar{\VX}\to\bar{\VP}$. This
decomposition allows us  immediately to compare particular
submodules/quotients -- the terms in the filtration~\eqref{VirN-filtr} -- with the ones in $\interP{0}$ on
Fig.~\ref{ind-chain-mod-fig}. First, the socle (the maximal semisimple
submodule) $\Soc(\interP{0})$ has the Virasoro character
\begin{equation*}
\Char\left[\Soc(\interP{0})\right] = \sum_{j-\mathrm{odd}} F^{(0)}_{j,1} =
\sum_{j_1,j_2\geq0}\bigl(2\min(j_1,j_2)+1\bigr)\chi_{2j_1+1,1}\bar{\chi}_{2j_2+1,1} +
\sum_{j_1,j_2\geq1}2\min(j_1,j_2)\chi_{2j_1,1}\bar{\chi}_{2j_2,1}
\end{equation*}
which coincides with the character of $\Soc(\bimnch_{(0)})$ easily
extracted from~\eqref{nsf-bos-Sz-zero}. The same is true for the top
parts (the maximal semisimple quotient) $\Top$ of $\bimnch_{(0)}$ and
$\interP{0}$, which are isomorphic to the socle in our case. Second, we
compare the middle-level subquotient $\Mid/\Soc$ in the
filtration~\eqref{VirN-filtr} consisting of all subquotients/nodes
having in-arrows from the top and out-arrows directed to the
bottom/socle. The middle level of $\interP{0}$, constituting of two
copies for each $F^{(0)}_{2j,-1}$ (see Fig.~\ref{ind-chain-mod-fig}),
has the same $\Vir(2)\boxtimes\overline{\Vir}(2)$ content
as the $\VirN(2)$-module $\bimnch$ has at the $\h=0$ grade, where we
use~\eqref{nsf-bos-Sz-zero} and~\eqref{sf-chiral-stagg-pic}. Similarly
one can proceed for any $\h=j$ grade.

 The  Kausch's  $\nSU(2)$ action discussed in
Sec.~5 of the \first paper~\cite{GRS1}, see also Sec.~\ref{sec:interch-simple}, `splits' the
sum~\eqref{JTL-simple-Vir-sum} into sectors (direct summands) of
different isospins\footnote{In principle, it is possible to find the $\nSU(2)$ content
of the $\interP{j}$ modules in Fig.~\ref{ind-chain-mod-fig}
using the lattice realization of the Kausch's $\nSU(2)$ given in
Sec.~5.3 in~\cite{GRS1}
but we do not do it in this paper.} -- e.g., with $k\geq {j\over 2}$ and $j\geq1$, in
the bosonic case.
 For each sector, the formulas~\eqref{nsf-bos-sl-lrVir-dec2}
and~\eqref{nsf-fer-sl-lrVir-dec2} give the indecomposable structure
under $\VirN(2)$. The simplest
is for $k=0$, where we obtain the structure represented on
the left in Fig.~\ref{sfnonchiralvirvac}. By a slight change of `geometry' the
diagram for this $\VirN(2)$-module can be represented as well as on
Fig.~\ref{sfnonchiralvirvacFF},
\begin{figure}\centering
    \includegraphics[scale=0.85]{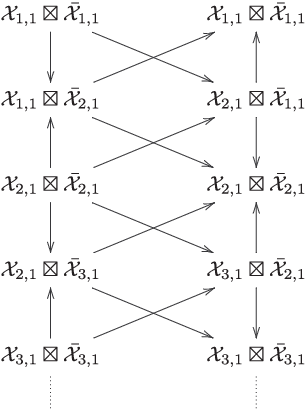}
    \caption{Two-strands
structure of the vacuum sector for non-chiral symplectic fermions.}
    \label{sfnonchiralvirvacFF}
    \end{figure}
with a pattern of arrows coincident with the
one in Fig.~\ref{ind-chain-mod-fig}. The case $k=1$ in
Fig.~\ref{sfnonchiralVirtripl} appears more complicated, but it is
only because of additional `gaps' in the arrows when compared to
$k=0$.  The decomposition is still fully compatible with
Fig.~\ref{ind-chain-mod-fig}.
Here, additional checks of consistency between structures of $\interP{j}$ and
$\bimnch_{(j)}$ involve finer filtrations.
Consider for example the following filtration of $U \interinfin$-submodules in
$\interP{j}$
\begin{equation}\label{VirN-finer-filtr}
0=\filt_{j-1}\subset\filt_j\subset\filt_{j+1}\subset\filt_{j+2}\subset\dots\subset\interP{j},
\end{equation}
where $\filt_j$ is the submodule generated from the top subquotients
$F_{k,(-1)^{k+1}}^{(0)}$, with $k\leq j$, in
Fig.~\ref{ind-chain-mod-fig}. It turns out that each term of this filtration is a direct
sum of indecomposable $\VirN(2)$ submodules, belonging in general to different
isospin-sectors. This follows from the fact that the algebra $\LQGodd$ of all
intertwining operators between
$U\interinfin$-modules gives also intertwining operators for
$\VirN(2)$-modules.
To show an explicit example of a correspondence between the terms of
the filtration~\eqref{VirN-finer-filtr} and $\VirN(2)$ submodules, we consider for simplicity again
the case $j=0$, where the first non-trivial term
in~\eqref{VirN-finer-filtr} is $\filt_1$ with the subquotient
structure
 \begin{equation}\label{filt-one-pattern}
{\footnotesize
   \xymatrix@=26pt
{
     {\Firr{1}{1}}\ar@[red][d]\ar[drr]
     &&{\Firr{1}{1}}\\
     {\Firr{2}{-1}}\ar@[red][urr]\ar[drr]
     &&{\Firr{2}{-1}}\ar[u]\ar@[red][d]\\
     &&{\Firr{3}{1}}
}
}
 \end{equation}
It is decomposed over $\VirN(2)$ into the direct sum (over all integer
isospins)
\begin{equation*}
\includegraphics[scale=0.9]{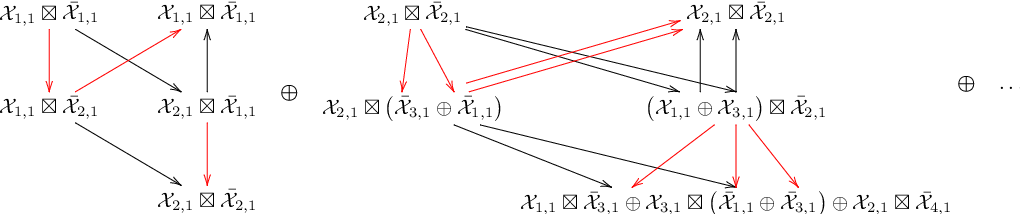}
\end{equation*}
where the first and second summands are particular
$\VirN(2)$-submodules of $\bimnbos_0$ in
Fig.~\ref{sfnonchiralvirvacFF} and $\bimnbos_1$ in
Fig.~\ref{sfnonchiralVirtripl}, respectively (we just rearranged nodes
introducing red arrows
for the right or anti-chiral Virasoro action). The other
direct summands of higher isospins appear in the same
pattern~\eqref{filt-one-pattern} with arrows and nodes in the
$U\interinfin$-picture split according to~\eqref{Fs-Vir-content-gen} in
order to
get the left-right Virasoro diagrams. All the other terms $\filt_j$ in the
filtration~\eqref{VirN-finer-filtr} can in principle be analyzed in a
similar way but pictures are more complicated and we do not give them.

\medskip

To conclude, we checked  that all arrows  present in the  diagrams for the subquotient structure of the scaling limit of
the spin-chain  modules $\APrTL{j}$ or indecomposables $\interP{j}$ over $U\interinfin$ are also in the corresponding infinite diagrams for the subquotient structure of
modules over $U\interinfinco$, which contains the Virasoro algebra $\VirN(2)$,  and vice versa, confirming our earlier results.  Our analysis was based on a decomposition onto direct summands over the left-right Virasoro algebra.
 Then, due to an isomorphism found around~\eqref{Uinterch-hom} and Conj.~\ref{conj:simples} about simple modules over $\interchalg$, we formulate our final conjecture.
  \begin{conj}
 Taking the completions $\interPco{j}$ (in the formal topology as defined in Sec.~\ref{sec:interch-interinfin}) of the indecomposable $U\interinfin$-modules  $\interP{j}$ gives
indecomposable modules over the interchiral algebra $\interchalg$  and their subquotient structure is given by the same infinite towers as for $\interP{j}$ in Fig.~\ref{ind-chain-mod-fig}.
\end{conj}

\section{Conclusion}\label{sec:concl}
This paper terminates our series on the scaling limit of the $\gl(1|1)$
spin chain. The conclusions of the analysis of \cite{ReadSaleur07-2}
in the open case do carry over to the periodic case, but at the price
of several complications. The bimodule structure of the spin chain
over $\JTL{N}$ and its centralizer $\LQGoddi$ does turn out to be
compatible with the known symplectic fermion continuum limit. However,
simple modules over $\JTL{N}$ correspond to direct sums of modules over
$\VirN(2)$, the left-right Virasoro algebra, while the known $\nSU(2)$
symmetry of Kausch is not present on the discrete spin chain in the
bulk case (in contrast with the boundary case). The lesson drawn is
that the good organizing object of the LCFT is truly the scaling limit
of the $\JTL{N}$ algebra, and that it contains more than $\VirN(2)$,
which must be extended by non chiral fields (not commuting with
$\nSU(2)$), giving rise to the {\bf interchiral
algebra} $\interchalg$. The centralizer of this interchiral algebra
remains $\LQGoddi$ in the scaling limit, and the bimodule structure
of the $\gl(1|1)$ discrete spin chain carries over identically now to
the scaling limit, exactly as in the open case.
Interestingly, we  note that the bimodule  over the pair of algebras $\interchalg$
and $\LQGoddi$ produces  acyclic complexes with the differentials $\EE
n$ and $\FF n$ ($\EE{n}^2=\FF{n}^2=0$). These can be interpreted  as non-chiral analogues of
the well-known Felder resolutions for chiral Virasoro
representations.

\medskip

The jaded reader might argue of course that the complex algebraic analysis presented here and in~\cite{GRS1,GRS2} is not really necessary nor useful to
understand most physical properties of symplectic fermions. The point however is that we are in the process of developing a strategy to tackle interacting
models for which  very little is known, and a direct solution based on a (free) action simply unavailable. In this case, we believe that the algebraic analysis is an essential tool to make progress, and answer such simple questions as which conformal fields are degenerate and which are not, which fields have logarithmic partners, what are the logarithmic couplings, {\it etc}. It is also probably possible to obtain information about fusion rules using this approach.

Sadly, things are bound to be more complicated in most cases than they were for $\gl(1|1)$. The reason is, that for $\gl(1|1)$, the representation of
$\JTL{N}$ is in fact non-faithful -- the
structure of the bulk and boundary theories are, as a result, quite
similar. General $\gl(n|m)$ or $\mathfrak{osp}(n|2m)$  spin chains, such as those necessary to study the cases $c=0$, $c=1$, will provide faithful representations, for which the
structure of the indecomposable $\JTL{N}$ modules will be considerably more involved. Their analysis will, in fact, require use of more sophisticated algebraic techniques, and will be started in~\cite{GRS4}. Note that even for $c=-2$, the $\gl(1|1)$ spin chain is non generic, and the  $\gl(2|2)$ spin chain for instance will lead, in the scaling limit, to a $c=-2$ LCFT that bears little resemblance with symplectic fermions.

\section*{Acknowledgements} We are grateful to C. Candu, I. Frenkel,
J.L. Jacobsen, V. Schomerus, I.Yu. Tipunin, and R. Vasseur for valuable
discussions. The work of A.M.G. was supported in part by Marie-Curie IIF fellowship, the
RFBR grant 10-01-00408, and the RFBR--CNRS grant
09-01-93105. The work of N.R. was supported by the NSF grants DMR-0706195 and DMR-1005895.  The work of H.S. was supported by the ANR Projet 2010
Blanc SIMI 4 : DIME.  We are grateful to the  Institut Henri Poincar\'e where this work was completed, and to the organizers of the program ACFTA for their kind hospitality. A.M.G is also grateful to Nick Read for his
kind hospitality in Yale University during May, 2012.

\section*{Appendix A: Fourier transforms}\label{sec:fourier}
\renewcommand\thesection{A}
\renewcommand{\theequation}{A\arabic{equation}}
\setcounter{equation}{0}

It is convenient to introduce Fourier transforms of the $f_j$ and
$f^{\dagger}_j$ fermions used in the definition of our $\JTL{N}$
representation in~\eqref{rep-TL-1} and~\eqref{rep-JTL-2}. We set, for $1\leq m\leq N$ (recall that we set $N=2L$),
\begin{equation}\label{def-ferm}
\ferm_{p_m} = \ffrac{1}{\sqrt{N}}\sum_{k=1}^N e^{-ikp_m}i^{-k}f_{k}, \qquad
\fermd_{p_m} = \ffrac{1}{\sqrt{N}}\sum_{k=1}^N e^{ikp_m}i^{-k}f^{\dagger}_{k}
\end{equation}
with the set of allowed momenta
\begin{equation}\label{momenta-set}
p_m=
\begin{cases}
\frac{2\pi m}{N},\qquad &L-\text{even},\\
\frac{(2m-1)\pi}{N}, &L-\text{odd},
\end{cases}
\qquad\quad 1\leq m\leq N,
\end{equation}
and with the usual anti-commutation relations
\begin{equation*}
\{\ferm_{p_1},\fermd_{p_2}\} = \delta_{p_1,p_2}, \qquad \{\ferm_{p_1},\ferm_{p_2}\} = \{\fermd_{p_1},\fermd_{p_2}\} = 0.
\end{equation*}

We then introduce the following linear combinations of fermions
\begin{align}
\chi^{\dagger}_p &=
\ffrac{1}{\sqrt{2}}\Bigl(\sqrt{\tan{\ffrac{p}{2}}}\,\fermd_{p-\frac{\pi}{2}} +
\sqrt{\cot{\ffrac{p}{2}}}\,\fermd_{p+\frac{\pi}{2}}\Bigr),&
\chi_p &=
\ffrac{1}{\sqrt{2}}\Bigl(\sqrt{\cot{\ffrac{p}{2}}}\,\ferm_{p-\frac{\pi}{2}} +
\sqrt{\tan{\ffrac{p}{2}}}\,\ferm_{p+\frac{\pi}{2}}\Bigr),\notag\\
\eta^{\dagger}_p &=
\ffrac{1}{\sqrt{2}}\Bigl(\sqrt{\tan{\ffrac{p}{2}}}\,\fermd_{p-\frac{\pi}{2}} -
\sqrt{\cot{\ffrac{p}{2}}}\,\fermd_{p+\frac{\pi}{2}}\Bigr),&
\eta_p &=
\ffrac{1}{\sqrt{2}}\Bigl(\sqrt{\cot{\ffrac{p}{2}}}\,\ferm_{p-\frac{\pi}{2}} -
\sqrt{\tan{\ffrac{p}{2}}}\,\ferm_{p+\frac{\pi}{2}}\Bigr),\label{eq:chi-eta-def}\\
\chi^{\dagger}_0 &= \fermd_{\frac{\pi}{2}},\quad \chi_0 = \ferm_{\frac{\pi}{2}},&
\eta^{\dagger}_0 &= \fermd_{\frac{3\pi}{2}},\quad \eta_0=\ferm_{\frac{3\pi}{2}},\notag
\end{align}
with momenta $p$ shifted by $\pi/2$ and taking thus values
$p=p_n=\step n$, where $\step=\frac{\pi}{L}$ and $1\leq n\leq L-1$, for even and odd $L$. The
$\chi$ and $\eta$ fermions  satisfy the anti-commutation relations
\begin{equation}\label{chi-eta-rels}
\left\{\chi^{\dagger}_p,\chi_{p'}\right\} = \left\{\eta^{\dagger}_p,\eta_{p'}\right\} =
\delta_{p,p'}, \qquad
\left\{\chi_p,\eta_{p'}\right\} =
\left\{\chi^{\dagger}_p,\eta^{(\dagger)}_{p'}\right\} =
\left\{\eta^{\dagger}_p,\chi^{(\dagger)}_{p'}\right\} = 0.
\end{equation}
These $\chi$ and $\eta$ fermions are creation and
annihilation operators for the Hamiltonian $H$ from~\eqref{hamil-def} and they were found in~\cite{GRS1}.

For the anti-periodic model, we  introduce $\ferm_p$ and $\fermd_p$ with the same formal expression~\eqref{def-ferm} but now the momenta $p_m$ take values $\frac{2\pi m}{N}$ for $L$ odd and $\frac{(2m-1)\pi}{N}$ for $L$ even, with $1\leq m\leq N$. As a
result, the values $p={\pi\over 2}, {3\pi\over 2}$ are not allowed, and
there are no zero modes. Finally,  we introduce
$\chi^{(\dagger)}_p$ and $\eta^{(\dagger)}_p$ fermions generating
Hamiltonian eigenstates from the vacuum by the same
formal definition~\eqref{eq:chi-eta-def}  but
now momenta take values $\step/2\leq p\leq \pi - \step/2$ with the
step $\step=\pi/L$.

\section*{Appendix B: Lie algebras  $\glinf$,  $\bglinf$, and $\interch$}\label{app:glinf}
\renewcommand\thesection{B}
\renewcommand{\theequation}{B\arabic{equation}}
\setcounter{equation}{0}
Recall~\cite{Kac-book} that $\glinf$ is a Lie algebra of  infinite matrices with a  \textit{finite} number of non-zero elements. This algebra admits a well-known representation (so-called  basic representation) in the free fermion models. Having fermion modes $\nferp_n$ with their conjugates $\nferm_{-n}$, the basis elements $\Egl_{n,m}$ of $\glinf$ -- the elementary matrices having identity at the position $(n,m)$ and zeros otherwise -- can be written as
\begin{equation}\label{basic-rep}
\Egl_{n,m} = \nferp_n\nferm_{-m}, \qquad n,m\in\oZ.
\end{equation}
 The  commutation relations in this algebra are then the usual ones as in a finite $\gl_N$ but now indices are in $\oZ$. The algebra $\glinf$ admits a central extension $\glinf \oplus \oC\one$  which we denote by $\glinf'$. This extension is defined using the two-cocylce $c:\glinf\times\glinf\to\oC$ given by (see~\cite{Kac-book})
\begin{equation}
c(\Egl_{n,m},\Egl_{m,n}) =  -c(\Egl_{m,n},\Egl_{n,m}) = 1, \quad\text{if}\; n\leq0, m\geq1,\quad \text{and}  \quad c(\Egl_{n,m},\Egl_{k,l}) = 0 \quad \text{otherwise}.
\end{equation}
Note that $c$ satisfies the cocycle condition that allows to modify the commutation relations in $\glinf$ as
\begin{equation}
[a+ x\one,b+ y\one] = [a,b] + c(a,b)\one,\qquad a,b\in\glinf, \; x,y\in\oC.
\end{equation}

The completed algebra $\bglinf$ is then defined as a  Lie algebra of    infinite matrices with
 a possibly infinite number of non-zero elements but the matrix still has a finite number of non-zero diagonals, those which are along the main diagonal containing the Cartan elements $\Egl_{n,n}$. This completion can be also described using a formal algebraic completion of the Cartan subalgebra in $\glinf$, see~\cite{Kac-book}.

 \medskip

 Recall that we introduced the Lie algebra $\interch$ in the end of Sec.~\ref{sec:comm-rel-interch} -- it is a Lie algebra of our local operators.
  We now define an embedding of $\interch$ into $\bglinf$, or more precisely to its central extension~$\bglinf'$ due to the presence of $L_0$ and $\bar{L}_0$ operators (see Sec.~\ref{sec:spinf}), in the following way. First, we introduce new fermionic modes $\nfer^{\alpha}_m$, for $\alpha=1, 2$ and $m\in\oZ$, and reorganize the modes of the chiral and anti-chiral symplectic fermions, $\psi^{\alpha}_n$ and $\bar{\psi}^{\alpha}_n$, such that the chiral modes  are identified with even modes of $\nfer$'s and the anti-chiral ones correspond to odd modes of $\nfer$'s:
\begin{equation}\label{nfer-sferm}
\psi_n^{\alpha} = \nfer^{\alpha}_{2n},\qquad \bar{\psi}^{\alpha}_n =  \nfer^{\alpha}_{-2n+1},\qquad n\in\oZ, \quad \alpha=1,2.
\end{equation}
Recall that we also have a pair of modes $\phi^{1,2}_0$ conjugate to $\nfer^{2,1}_0$ -- we do not use new notations for them.
Under the identification~\eqref{nfer-sferm}, it is then easy to see that the scaling limit of our $\gl_N$ algebras described in Sec.~\ref{sec:spinf}, and in more precise terms using direct limits in App.~\AppLim, is given by the basic representation~\eqref{basic-rep} of $\glinf$ where one also includes bilinears like $\nfer^{2}_n\phim_0$, {\it etc.} The Lie algebra $\interch$ which has basis elements as infinite sums of $\Egl_{n,m}$'s can be now embedded into the $\bglinf'$ algebra. Indeed, the Virasoro-like generators $L^{(l)}_{n}$ are the sums of bilinears in the $\nfer$ modes that have infinite number of terms like $\nord{\nferp_{2k}\nferm_{-2k+2n}}$ and all these terms belong to the same diagonal along the main one containing the Cartan elements  $\nferp_{k}\nferm_{-k}$. Therefore, any finite linear combination of  $L^{(l)}_{n}$ belongs to $\bglinf'$, and similarly for $\bar{L}^{(l)}_{n}$. Note then that the basis elements $\enrg^{(l)}_n$ are given by the sums containing  infinite number of terms like $\nferp_{2k}\nferm_{-2k-2n+1}$ and they also belong to $\bglinf'$. Finally, we obtain that any linear and finite combination of basis elements in $\interch$ belongs to $\bglinf'$.

\section*{Appendix C: Direct and inverse limits}\label{sec:direct-lim}
\renewcommand\thesection{C}
\renewcommand{\theequation}{C\arabic{equation}}
\setcounter{equation}{0}

In order to satisfy the more formal reader, we give here precise definition of scaling limits of spin chains and define the Lie algebra $\interinfin$ using direct-limit constructions. We first remind  a formal definition of the direct limit. A direct (or inductive) system is a pair $\{A_i,\phi_{ij}\}$ of a family of algebraic objects $A_i$ (vector spaces, algebras, {\it etc.}) indexed by an ordered set~$I$ and a family of homomorphisms $\phi_{ij}:A_i\rightarrow A_j$ for all $i\leq j$ satisfying the following properties: (i) $\phi_{ii}$ is the identity on $A_i$, and (ii) $\phi_{ik}=\phi_{jk} \phi_{ij}$ for all $i\leq j\leq k$. The \textit{direct limit} $A_{\infty}\equiv\varinjlim A_i$ of the direct system $\{A_i,\phi_{ij}\}$ is defined as the disjoint union $\coprod_i A_i/\sim$ modulo an equivalence relation:
two elements $a_i\in A_i$ and $a_j\in A_j$ in the disjoint union are equivalent if and only if they eventually become equal in the direct system, {\it i.e.}, $a_i \sim a_j$ if there is $k\in I$ such that $\phi_{ik}(a_i)=\phi_{jk}(a_j)$. We obtain from this definition canonical homomorphisms $\phi_i: A_i\rightarrow A_{\infty}$ mapping each element to its equivalence class. The algebraic operations on $A_{\infty}$ are then defined  using these maps.
Recall also that an inverse or projective system of algebraic objects  can be defined in a similar way reverting order and arrows everywhere. So, we  have a family of homomorphisms $\phi_{ij}:A_j\rightarrow A_i$, for all $i\leq j$, with similar to (i) and (ii) properties.
The \textit{inverse limit} is then defined as a particular subobject (subalgebra/subspace) in the direct product of the $A_i$'s:
\begin{equation}\label{eq:inv-lim-def}
A\equiv \varprojlim A_i = \bigl\{a\in\prod_i A_i \, | \, a_i = \phi_{ij}(a_j) \; \text{for all}\; i\leq j\bigr\}.
\end{equation}
The inverse limit $A$ is equipped with canonical projections $\proj_i:A\to A_i$ defined by taking the $i$th component of the direct product.

\subsection{The scaling limit of JTL algebras as a direct limit}\label{app:JTL-lim}
Consider a case of odd $L=N/2$ for simplicity and an embedding $\phi_L$ of the Clifford algebra  $\Clif{4L}$ into a bigger one  $\Clif{4(L+2)}$ defined by
\begin{align}\label{eq:phi-def}
\begin{split}
&\phi_L\bigl(\chi^{(\dagger)}_{p}\bigr) = \chi^{(\dagger)}_{p'}, \qquad\qquad\quad\; \phi_L\bigl(\eta^{(\dagger)}_p\bigr) = \eta^{(\dagger)}_{p'},\\
&\phi_L\bigl(\chi^{(\dagger)}_{\pi-p}\bigr) = \chi^{(\dagger)}_{\pi-p'},\qquad\qquad  \phi_L\bigl(\eta^{(\dagger)}_{\pi-p}\bigr) = \eta^{(\dagger)}_{\pi-p'},
\end{split}
\qquad\qquad \textit{for} \quad 1\leq m \leq \ffrac{L-1}{2}.
\end{align}
where we set
$p\equiv p(m)=\frac{m\pi}{L}$, with $1\leq m\leq L-1$, and for the indexing set in the bigger Clifford algebra we set $p'\equiv p'(m)=\frac{m\pi}{L+2}$, with $1\leq m\leq L+1$, and for the zero modes we set
 \begin{equation}\label{eq:phi-def-0}
\phi_L(\chi_0)=\sqrt{\ffrac{L}{L+2}}\,\chi_0, \quad
\phi_L(\chi^{\dagger}_0)=\sqrt{\ffrac{L+2}{L}}\,\chi^{\dagger}_0,\quad
\phi_L(\eta_0)=\sqrt{\ffrac{L+2}{L}}\,\eta_0, \quad
\phi_L(\eta^{\dagger}_0)=\sqrt{\ffrac{L}{L+2}}\,\eta^{\dagger}_0.
 \end{equation}
 We note that  $\Clif{4(L+2)}$ is generated by its subalgebra $\phi_L(\Clif{4L})$ and $8$ additional generators $\chi^{(\dagger)}_{\frac{\pi}{2}\pm\frac{\step'}{2}}$ and $\eta^{(\dagger)}_{\frac{\pi}{2}\pm\frac{\step'}{2}}$, with $\step'=\frac{\pi}{L+2}$. The additional fermionic operators generate in $\Hilb_{2(L+2)}$ eigenstates of highest eigenvalue of the Hamiltonian $H$ in the one-particle subspace.

 The embeddings $\phi_L$ of algebras induce the corresponding embedding of the $\Clif{4L}$-module $\Hilb_{2L}$  into the $\Clif{4L}$-module $\Hilb_{2(L+2)}$ by $\phi_L\bigl(\chi^{\dagger}_{p}\vac\bigr) = \chi^{\dagger}_{p'}\vac$,
  {\it etc.}, and $\phi_L(ab\vac) = \phi_L(a)\phi_L(b\vac)$, for any $a,b\in\Clif{4L}$. Here, we denote the embeddings of the representation spaces of the Clifford algebras by the same letter $\phi_L:\Hilb_{2L}\to\Hilb_{2(L+2)}$. We thus construct two direct systems: the one for the Clifford algebras
\begin{equation*}
\Clif{4}\xrightarrow{\;\phi_1\;}\Clif{12}\xrightarrow{\;\phi_3\;}\dots\xrightarrow{\phi_{L-2}}\Clif{4L}\xrightarrow{\;\phi_{L}\;}\Clif{4(L+2)}\xrightarrow{\phi_{L+2}}\dots
\end{equation*}
with the direct limit an infinite-dimensional Clifford algebra $\Clif{}\equiv\varinjlim_L \Clif{4L}$ generated by the rescaled generators or the symplectic fermions modes $\sferm^{1,2}_n$, $\bsferm^{1,2}_m$, and $\phi^{1,2}_0$ introduced in~\eqref{eq:sferm-lat-def} and~\eqref{eq:sferm-zero-def} and satisfying~\eqref{sferm-rel} and~\eqref{sferm-0-rel}; the second direct system is for modules over the Clifford algebras or the spin-chains
\begin{equation}\label{eq:Hilb-direct-sys}
\Hilb_{2}\xrightarrow{\;\phi_1\;}\Hilb_{6}\xrightarrow{\;\phi_3\;}\dots\xrightarrow{\phi_{L-2}}\Hilb_{2L}\xrightarrow{\;\phi_{L}\;}\Hilb_{2(L+2)}\xrightarrow{\phi_{L+2}}\dots
\end{equation}
Using the canonical homomorphisms in the definition of direct systems, the direct limit $\varinjlim_L \Hilb_{2L}$ is then a $\Clif{}$-module. This module coincides with the space $\Hilb$ of scaling states or finite-energy and finite-spin states described  in Sec.~\ref{Cliff-def}.

Notice then that the direct system of Clifford algebras defines a direct system of the matrix algebras $\gl_{2L}$ in the following way\footnote{After actually rearranging four rows and four columns of zeros for a simpler presentation which is just a matter of conventions.}
\begin{equation*}
\dots \; \xrightarrow{\phi_{L-2}}
\begin{pmatrix}
  \Egl_{m,n} &     \Egl_{m,L+n}   \\
\Egl_{L+m,n}    &    \Egl_{L+m,L+n} \\
\end{pmatrix}_{2L\times 2L}
 \xrightarrow{\;\phi_L\;}
\begin{pmatrix}
0 & 0\,\dots\,0 &         0   & 0   &  0\,\dots\,0   &  0\\
0 & \Egl_{m,n} & 0 & 0    & \Egl_{m,L+n}    & 0\\
 0 & 0\,\dots\,0 &     0       &0   &      0\,\dots\,0  & 0 \\
0 & 0\,\dots\,0 & 0    & 0  & 0\,\dots\,0 &  0 \\
0 & \Egl_{L+m,n} & 0   & 0  &  \Egl_{L+m,L+n} &0 \\
0 & 0\,\dots\,0 &      0      & 0  &  0\,\dots\,0 &    0
\end{pmatrix}_{2(L+2)\times 2(L+2)}
\!\!\!\!\!\!\! \xrightarrow{\phi_{L+2}} \; \dots
\end{equation*}
where the elements $\Egl_{i,j}$ denote usual elementary matrices, with $0\leq n,m\leq L-1$,  defined in~\eqref{gl-stand-basis}
and we abuse notations denoting the embeddings of matrix algebras by the same $\phi_L$ as for the Clifford algebras. The direct limit $\varinjlim_L \gl_{2L}$ then gives the infinite-dimensional Lie algebra that we call $\glinf$ -- an algebra of infinite matrices  with a finite number of non-zero elements. An identification of this limit with the standard representation~\eqref{basic-rep} of $\glinf$ is given via the identification~\eqref{nfer-sferm} of the fermionic generators.

In order to get the central extension $\gl'_{\infty}$  algebra  having  the proper normally ordered basis, necessary for physics, see~\eqref{gl-inf-2}, we should
 map the normally ordered generators of $\gl'_{2L}=\gl_{2L}\oplus\oC\one$, i.e., we define
\begin{equation}\label{phi-prime}
\phi'_L:\gl'_{2L}\to\gl'_{2(L+2)},\qquad \text{with}\quad
  \phi'_L(\Egl'_{m,n})=\Egl'_{m,n}
\end{equation}
   and set $\phi'_L(\one)=\one$, where $\one$ is the identity or the central element.
In particular, mapping   the normally ordered generators of $\interfin{N}'\subset\gl'_{2L}$ subalgebra (defined in~\eqref{def-prime-alg} and Sec.~\ref{S_N-def})  as
  $\phi'_L(\Asp'_{m,n})=\Asp'_{m,n}$, etc.,
we get in the direct limit $\varinjlim_L \interfin{2L}'$
what we call the $\interinfin$ Lie algebra
 and the identity is mapped under the canonical homomorphisms $\interfin{N}'\to\interinfin$  to the central element in $\interinfin$.

It turns out that the  Lie algebra $\interinfin$ has $\strt_{m,n}$, $\bstrt_{m,n}$ and $\enrg_{m,n}$ with $m,n\in\oZ$ as its basis or equivalently the universal enveloping algebra $U\interinfin$ has the symplectic-fermion representation
on $\Hilb$ generated by $\one$ and
\begin{equation}\label{app:sympl-ferm-interLie}
\enrg_{n,m} = S_{\alpha\beta} \psi^{\alpha}_n \bar{\psi}^{\beta}_m,\quad
\strt_{k,l}=J_{\alpha\beta}{\sferm^{\alpha}_{k}\sferm^{\beta}_{l}},\quad \bstrt_{r,s}=J_{\alpha\beta}{\bsferm^{\alpha}_{r}\bsferm^{\beta}_{s}}, \qquad\text{with}\quad n,m,k,l,r,s\in\oZ,
\end{equation}
as was also pointed out in Sec.~\ref{sec:interinfin}.


 Using Thm.~\ref{Thm:US-JTL-iso}, we can  define the scaling limit of JTL algebras using the direct system for  $\interfin{N}'$ algebras.
  Indeed, the generators of $\interfin{N}'$ are also generators of the image $\repgl\bigl(\JTL{N}\bigr)$. Therefore, the direct system for the enveloping algebras of $\interfin{N}'$ defines the direct system for the images of the JTL algebras
\begin{equation}\label{JTL-direct}
\repgl\bigl(\JTL{2}\bigr)\xrightarrow{\;\phi'_1\;}\repgl\bigl(\JTL{6}\bigr)\xrightarrow{\;\phi'_3\;}\dots\xrightarrow{\phi'_{L-2}}
\repgl\bigl(\JTL{2L}\bigr)\xrightarrow{\;\phi'_{L}\;}\repgl\bigl(\JTL{2(L+2)}\bigr)\xrightarrow{\phi'_{L+2}}\dots
\end{equation}
where $\phi'_i$ are defined in~\eqref{phi-prime} and are embeddings of the finite-dimensional associative algebras.
Note that an existence of embeddings for JTL algebras is a non-trivial (comparing to the open case) result and it is very hard to express these embeddings in terms of diagrams or the initial generators $e_j$'s and $u^2$.
Nevertheless, we found a special basis in which a system of embeddings is easily constructed and the direct system~\eqref{JTL-direct} gives the limit -- an enveloping algebra of $\interinfin$ --  which we call  \textit{the scaling limit of JTL} algebras.

Now, we can consider the direct system~\eqref{eq:Hilb-direct-sys} of spin-chains as the direct system of JTL modules (by restriction of the action from $\Clif{4L}$ to $\JTL{2L}$) and therefore the direct limit space $\Hilb=\varinjlim_L \Hilb_{2L}$ has canonically the $U\interinfin$-module structure: an element $a$ of $\interinfin$ has its repesentative in the subalgebra $\interfin{N}$,  for some $N$, or simply saying $a\in\interfin{N}\subset\interinfin$ for large enough $N$, and then the action of $a$ is given by the action of this representative on $\Hilb$, with the latter defined by the direct limit construction of the $\interfin{N}$-modules. The result obviously does not depend on the choice of $N$.

 We then show that the monomorphisms $\phi_L$ in~\eqref{eq:phi-def}-\eqref{eq:Hilb-direct-sys} commute also with the action of the centralizer $\centJTL$ or the image of $\LQGodd$ in each term $\Hilb_{2L}$.

\begin{Thm}\label{thm:phi}
The embeddings $\phi_L$ in the direct system~\eqref{eq:Hilb-direct-sys} commute with the actions of the $\JTL{2L}$ and $\centJTL$ algebras, or equivalently~\eqref{eq:Hilb-direct-sys} defines a direct system of bimodules over $\bigl(\JTL{2L},\centJTL\bigr)$.
\end{Thm}
\begin{proof}
The statement about the $\JTL{2L}$ action was proved by the previous discussion around~\eqref{JTL-direct} noting that we have the equality $\phi'_L=\phi_L$ for homomorphisms of associative algebras.
Recall then the expressions~\eqref{eq:EE-chi-eta} for the $\LQGodd$ generators. We first note that for all $L\geq1$ the embeddings $\phi_L$ defined in~\eqref{eq:phi-def}-\eqref{eq:phi-def-0} commute with the action of $\h$ or the $S^z$ operator because $\phi_L$ do not change the number of fermionic operators in a homogeneous $v\in\Hilb_{2L}$. Using~\eqref{eq:phi-def-0}, we see that the $\phi_L$'s commute also with $\E=\EE 0=(-1)^{S^z}\sqrt{N}\chi_0^{\dagger}$ and $\F=\FF 0=-i\sqrt{N}\eta_0$ generators.
It is thus enough to show that $\phi_L$ commutes with the operators
\begin{equation*}
A_L=\sum_{p=\step}^{\pi-\step}\eta_p\chi_{\pi-p},\qquad\qquad
B_L=\sum_{p=\step}^{\pi-\step}
\chi^{\dagger}_{\pi-p}\eta^{\dagger}_p,
\end{equation*}
for any $L\geq1$ and the sum is with the step $\step=\frac{\pi}{L}$ as usual. For any $v\in\Hilb_{2L}$ and odd $L$ we have
\begin{equation*}
\phi_L(A_L v) = \sum_{\substack{p'=\step'\\p'\ne\frac{\pi}{2}\pm\frac{\step'}{2}}}^{\pi-\step'}\eta_{p'}\chi_{\pi-p'} \phi_L(v) = A_{L+2}\phi_L(v),
\end{equation*}
with $\step'=\frac{\pi}{L+2}$ and in the last equality we used the fact that $\chi_p$ is an annihilation operator and thus the image $\phi_L(v)$ is in the kernel of the operators $\chi_{\frac{\pi}{2}\pm\frac{\step'}{2}}$. Similar calculation for even $L$ and the $B_L$, $\f^L$ and $\e^L$ operators finishes our proof.
\end{proof}

We can of course construct direct limits of any submodules over $\JTL{N}$ in $\Hilb_{N}$.
In particular, for any simple $\JTL{N}$-subquotient $\AIrrTL{j}{(-1)^{j+1}}$ that appear in $\Hilb_N$, as a fundamental representation of $\ssp_{N-2}$ due to Cor.~\ref{cor:JTL-sp-simples}, we have an isomorphic submodule in (the socle of) $\Hilb_N$. This fact follows from the explicit analysis of the module structure in~\cite{GRS2} or see Fig.~\ref{FF-JTL-mod}. Fixing such an irreducible submodule, say in the fermionic basis, the  direct system~\eqref{eq:Hilb-direct-sys}
thus gives a direct system of the simple modules $\AIrrTL{j}{(-1)^{j+1}}$ for different values of $N$. We denote their direct limit by
$\interM{j}$, for any integer $j\geq 1$, and these infinite-dimensional spaces have canonically the action of $U \interinfin$ as explained above.

\begin{Prop}\label{prop:lim-simples}
The direct limit modules $\interM{j}$ over $U \interinfin$ are irreducible for any integer $j\geq 1$.
\end{Prop}
\begin{proof}
Assume that for some $j$ the $\interinfin$-module $V\equiv\interM{j}$ is reducible and has thus an invariant subspace $W\subset V$. Take any element $v\in V$ such that it belongs to the complement of $W$. Then, by definition of the direct limit, there exist a positive number $N$ such that the element $v$ has its representative in the subspace $\AIrrTL{j}{(-1)^{j+1}}[N]\subset V$, where $\AIrrTL{j}{(-1)^{j+1}}[N]$ is just notation for the simple module over $\interfin{N}$, or over $\ssp_{N-2}$. Simply saying $v$ is in the subspace $\AIrrTL{j}{(-1)^{j+1}}[N]$ for large enough $N$. Take then the intersection $W'$ of the $\interfin{N}$-module $\AIrrTL{j}{(-1)^{j+1}}[N]$ with $W$. The $\interfin{N}$ acts on $W'$, otherwise it would contradict to the fact that $W$ is the invariant subspace for $\interinfin$ and thus for $\interfin{N}\subset\interinfin$. We also have that $W'$ is a proper subspace of $\AIrrTL{j}{(-1)^{j+1}}[N]$, otherwise $v$ would belong to $W'$ and it contradicts to the assumption that $v$ is in the complement to $W$. We thus obtain a contradiction to the fact that $\AIrrTL{j}{(-1)^{j+1}}[N]$ is irreducible for $\interfin{N}$.
\end{proof}

Using the same idea of a reduction from $\Hilb$ to its subspace $\Hilb_N\subset\Hilb$ for large enough $N$ and transporting thus an assumption in $\Hilb$ to the corresponding assumption in $\Hilb_N$, one can similarly prove that the direct summands $\APrTL{j}$ in $\Hilb_N$ (these are indecomposable but reducible $\JTL{N}$-modules from Fig.~\ref{FF-JTL-mod}) have the direct limit denoted by $\interP{j}$ as an indecomposable but reducible $U\interinfin$-module and its subquotient structure (with simple subquotients) is described by the infinite tower in Fig.~\ref{FF-JTL-mod}.

\newcommand{\rodd}{\rho\bigl(\LQGodd\bigr)}

\subsection{Inverse limit of centralizers}\label{app:inv-lim}
By Thm.~\ref{thm:phi}, we have a direct system of bimodules over $\bigl(\JTL{2L},\LQGodd\bigr)$ and the direct limit space~$\Hilb$ is canonically a bimodule over $(U\interinfin,\LQGodd)$.
The action of $U\interinfin$ on $\Hilb$ is described in~\eqref{app:sympl-ferm-interLie} and the $\LQGodd$ action has the fermionic expressions computed in our first paper~\cite{GRS1}:
\begin{equation}\label{app:EE-FF}
  \EE n =
\left[\sum_{m>0}
\Bigl(\ffrac{\sfermp_{m}\sfermp_{-m}}{m} -
\ffrac{\bsfermp_{m}\bsfermp_{-m}}{m}\Bigr)\right]^n\psi_0^2,\qquad
\FF n =
\left[\sum_{m>0}
\Bigl(\ffrac{\sfermm_{m}\sfermm_{-m}}{m} -
\ffrac{\bsfermm_{m}\bsfermm_{-m}}{m}\Bigr)\right]^n\psi_0^1,
\end{equation}
with the representation of the Cartan element $\h$ as
\begin{equation}\label{app:h}
\h =- i/2 \bigl(\sfermp_0\phim_0 + \sfermm_0\phip_0\bigr) + \sum_{m>0}
\ffrac{1}{m}\bigl(\sfermp_{-m}\sfermm_{m} + \sfermm_{-m}\sfermp_{m} + \bsfermp_{-m}\bsfermm_{m} + \bsfermm_{-m}\bsfermp_{m}\bigr)
\end{equation}
while the generator $\K = (-1)^{2\h}$.
This action can be obtained by constructing an {\sl inverse} system for the spin-chain representations $\rho_N$ of $\LQGodd$
as\footnote{This inverse system actually corresponds to the usual one~\cite{[Donk]} for the $\q$-Schur algebras -- quotients of  $\LQG$.}
\begin{equation}\label{QG-inverse}
\rho_2\bigl(\LQGodd\bigr)\xleftarrow{\;p_1\;}\rho_6\bigl(\LQGodd\bigr)\xleftarrow{\;p_3\;}\dots\xleftarrow{p_{L-2}}\rho_{2L}\bigl(\LQGodd\bigr)\xleftarrow{\;p_{L}\;}
\rho_{2(L+2)}\bigl(\LQGodd\bigr)\xleftarrow{p_{L+2}}\dots
\end{equation}
where $\rho_N$ is defined in~\eqref{eq:EE-chi-eta} and $p_n$ are surjective homomorphisms defined in terms of the fermionic generators as
\begin{equation}\label{eq:p-def}
p_L\bigl(\chi^{(\dagger)}_{p'}\bigr)=\Bigl(1-\delta_{p',\frac{\pi}{2}\pm\frac{\step'}{2}}\Bigr)\chi^{(\dagger)}_p,\qquad
p_L\bigl(\eta^{(\dagger)}_{p'}\bigr)=\Bigl(1-\delta_{p',\frac{\pi}{2}\pm\frac{\step'}{2}}\Bigr)\eta^{(\dagger)}_p,
\qquad 1\leq m\leq L+1,
\end{equation}
where as usual
$p\equiv p(m)=\frac{m\pi}{L}$ and $p'\equiv p'(m)=\frac{m\pi}{L+2}$,
and for the zero modes as
 \begin{equation}\label{eq:p-def-0}
p_L(\chi_0)=\sqrt{\ffrac{L+2}{L}}\,\chi_0, \quad
p_L(\chi^{\dagger}_0)=\sqrt{\ffrac{L}{L+2}}\,\chi^{\dagger}_0,\quad
p_L(\eta_0)=\sqrt{\ffrac{L}{L+2}}\,\eta_0, \quad
p_L(\eta^{\dagger}_0)=\sqrt{\ffrac{L+2}{L}}\,\eta^{\dagger}_0.
 \end{equation}
In other words, we define first $p_L$ as a vector-space homomorphism $\Clif{4(L+2)}\to\Clif{4L}$ such that its kernel equals the cokernel of $\phi_L$.
Then using~\eqref{eq:EE-chi-eta}, it is easy to check that the surjective maps $p_L:\Clif{4(L+2)}\surj\Clif{4L}$ of vector spaces  induce homomorphisms of subalgebras $\rho_{2(L+2)}\bigl(\LQGodd\bigr)\subset\Clif{4(L+2)}$ onto $\rho_{2L}\bigl(\LQGodd\bigr)$ and correspondingly homomorphisms $p_L:\Hilb_{2(L+2)}\surj\Hilb_{2L}$ of modules over $\rho_{2(L+2)}\bigl(\LQGodd\bigr)$.
Then, the inverse limit
\begin{equation}\label{app:inv-lim-U}
\rho\bigl(\LQGodd\bigr)=\varprojlim \rho_{2L}\bigl(\LQGodd\bigr)
 \end{equation}
defines an infinite-dimensional associative algebra  which we identify with a quotient of $\LQGodd$ that admits among irreducible representations only one-dimensional ones, see representation theory of $\LQGodd$ in~\cite[Sec.~3]{GRS2}. We call this inverse limit as \textit{the scaling limit of the JTL centralizers} $\centJTL$. The corresponding inverse limit space $\cHilb=\varprojlim \Hilb_{2L}$ of the $\LQGodd$-modules is a module over $\rho\bigl(\LQGodd\bigr)$ in~\eqref{app:inv-lim-U}. The action is defined componentwise following the definition~\eqref{eq:inv-lim-def}. The explicit  action~\eqref{app:EE-FF} and~\eqref{app:h} of the generators of the inverse limit on~$\cHilb$ is calculated  using~\eqref{eq:p-def} and~\eqref{eq:p-def-0} and the identification in~\eqref{eq:sferm-lat-def} and~\eqref{eq:sferm-zero-def}.

The direct and inverse systems of the spaces $\Hilb_{N}$ have a useful coherence property in the sense of the following simple lemma.

\begin{Lemma}\label{lem:coh}
The direct system $\phi_L:\Hilb_{2L}\inj\Hilb_{2(L+2)}$ from~\eqref{eq:Hilb-direct-sys} and defined in~\eqref{eq:phi-def} and~\eqref{eq:phi-def-0}
satisfies
\begin{equation}\label{eq:coherence}
p_L\circ\phi_L = \mathrm{id},\qquad\qquad \text{for}\quad L\geq1,
\end{equation}
where the projective system $p_L:\Hilb_{2(L+2)}\surj\Hilb_{2L}$ is defined in~\eqref{eq:p-def} and~\eqref{eq:p-def-0}.
\end{Lemma}

\begin{rem}\label{rem:coh}
As a consequence of Lem.~\ref{lem:coh}, the algebra $\interfin{2L}$ acts on the image of $p_L$ and this action is identical to the one on the subspace $\Hilb_{2L}\subset\Hilb_{2(L+2)}$.  Moreover, it is straightforward to check that the map $p_L$ intertwines the $\interfin{2L}$ action on  $\Hilb_{2(L+2)}$ and $\Hilb_{2L}$. 
\end{rem}

Now, we discuss a relation between the two limits, $\Hilb$ and $\cHilb$, so far constructed.
Recall that the Hamiltonian $H$ in the scaling limit~\eqref{L0-def} of the JTL algebras gives the grading (or energy) operator $E=L_0+\bar{L}_0$, which is a bounded below operator, while the momentum operator $P=P(0)$ in~\eqref{HPn-def} gives the conformal spin operator $S=L_0-\bar{L}_0$ (not bounded).
Following the definition of $\phi_L$ and $p_L$ and using commutation relations of the fermionic generators with $E$ and $S$, we conclude that the direct limit is a bi-graded vector space with the decomposition $\Hilb=\oplus_{n\geq0, m\in\oZ}\, \Hilb^{(n,m)}$ onto generalized eigenspaces or root spaces of $(E,S)$, while the projective limit $\cHilb$ is the direct product $\cHilb=\prod_{n\geq0, m\in\oZ} \Hilb^{(n,m)}$.
Therefore, the inverse limit $\cHilb$ contains the direct limit space $\Hilb$ as a dense subspace (in the formal topology discussed in Sec.~\ref{sec:interch-interinfin}) and the basis in $\Hilb$ is also a basis in $\cHilb$.
Thus, the space $\cHilb$ can be considered as a formal completion of $\Hilb$ and the canonical projections $\proj_N:\cHilb\twoheadrightarrow\Hilb_N$ reduce on the dense subspace $\Hilb$ to projections $\proj_N:\Hilb\twoheadrightarrow\Hilb_N$. Correspondingly, the action of $\rho\bigl(\LQGodd\bigr)$ reduces to the action on $\Hilb$ given by the same formulas~\eqref{app:EE-FF} and~\eqref{app:h}. Let $\Phi_N:\Hilb_{N}\inj\Hilb$ denotes the canonical embedding defined by the direct system~\eqref{eq:Hilb-direct-sys}. Using Lem.~\ref{lem:coh} and Rem.~\ref{rem:coh}, we formulate a corollary.
\begin{Cor}\label{cor:coh}
The canonical projections $\proj_N:\Hilb\surj\Hilb_N$ are in the coherence property with the canonical embeddings  $\Phi_N:\Hilb_{N}\inj\Hilb$:
\begin{equation}\label{eq:canon-coh}
\proj_N\circ\Phi_N = \mathrm{id},\qquad\qquad \text{for even}\quad N\geq2.
\end{equation}
Therefore, the image of the surjective map $\proj_N:\Hilb\surj\Hilb_N$ can be identified with the subspace $\Hilb_N\subset\Hilb$.
The map $\proj_N$ intertwines the $\interfin{N}$ action on  $\Hilb$ and $\Hilb_{N}$.
\end{Cor}

\medskip

It turns out that the centralizer of the algebra $U \interinfin$ action  on $\Hilb$ is given by the representation $\rho$ of $\LQGodd$ defined in~\eqref{app:inv-lim-U}.

\begin{Thm}\label{thm:centr-interLie}
The centralizer $\cent_{U\interinfin}$ of the enveloping algebra $U\interinfin$ action~\eqref{app:sympl-ferm-interLie}
on $\Hilb$
 equals $\rodd$, with the action in~\eqref{app:EE-FF} and~\eqref{app:h}.
\end{Thm}
\begin{proof}
The centralizer of $U\interinfin$ obviously commutes with the normally ordered operators $\nord{\strt_{n,-n}}$ and $\nord{\bstrt_{n,-n}}$, see~\eqref{app:sympl-ferm-interLie} and also discussion in Sec.~\ref{sec:interinfin}, and therefore with the energy operator
\begin{equation*}
E=L_0+\bar{L}_0=-\sum_{n>0}\bigl(\nord{\strt_{n,-n}} + \nord{\bstrt_{n,-n}}\bigr) - \nord{\strt_{0,0}},
 \end{equation*}
which is an element of the completed algebra $\interinfinco$ but on each $v\in\Hilb$ the sum reduces to a finite sum. Let $\Hilb^{(n)}$ denotes the generalized eigenspace (or root space) of eigenvalue $n$ for the $E$ operator, i.e., we have $\Hilb=\oplus_{n}\Hilb^{(n)}$.
From the partition function analysis given in Sec.~\ref{sec:gen-func}, we conclude that $n$ is non-negative integer and each $\Hilb^{(n)}$ is finite-dimensional.
 We thus have that  the centralizer is a subalgebra in the direct product
$\prod_{n\geq 0} \Endo\bigl(\Hilb^{(n)}\bigr)$ of finite matrix algebras.

Using the simple fact that the fermionic operators $\sferm^{1,2}_{-m}$ and $\bsferm^{1,2}_{-m}$ change an eigenvalue of $E$
 by $m$,
we have by construction of the direct system~\eqref{eq:Hilb-direct-sys}, with the identifications in~\eqref{eq:sferm-lat-def} and~\eqref{eq:sferm-zero-def}, that for a fixed positive $j$ there exists a big enough $N>0$ such that any $v\in\Hilb^{(j)}\subset\Hilb$ has its representative in the $N$-site spin-chain $\Hilb_N\subset\Hilb$. Using the canonical embedding $\Phi_{N}:\Hilb_{N}\inj\Hilb$  we can consider any such $v$ as an element in $\Hilb_N$ (or in other words, we observe that the fixed energy subspaces $\Hilb^{(n)}$ stabilize with the rise of $N$.)

Assume then that there exists a non-zero operator $O\in\prod_{n\geq 0} \Endo\bigl(\Hilb^{(n)}\bigr)$ such that it commutes with the action of $U\interinfin$ but is not an element of the algebra $\rodd$. In other words, writing the operator $O$ componentwise as $O=(O_0,O_1,O_2,\dots)$, with $O_j\in\Endo(\Hilb^{(j)})$, there exists a positive number $j$ such that the $j$th component $O_j$ of $O$ can not be written as a finite linear combination of words in the generators~\eqref{app:EE-FF} and~\eqref{app:h}. Following the preceding discussion we can take a big enough $N$ such that the $\Hilb^{(j)}$ is in the subspace $\Hilb_N\subset\Hilb$ (under the canonical embedding $\Phi_N$).
By the coherence property~\eqref{eq:canon-coh} in Cor.~\ref{cor:coh}, this subspace is identified with the image of the canonical projection $\proj_N:\Hilb\surj\Hilb_N$ and therefore on the subspace $\Hilb_N\subset\Hilb$ we have the action of $\rodd$ that factorizes up to the action of $\rho_N\bigl(\LQGodd\bigr)$. The latter also acts on $\Hilb^{(j)}\subset\Hilb_N$.
 By the assumption on the operator $O$, we have that the composition $\proj_N O \proj_N\in\Endo\, \Hilb_N$ is not an element of $\rho_N\bigl(\LQGodd\bigr)$, otherwise its action on $\Hilb^{(j)}\subset\Hilb_N$ could be written in terms of the $\rodd$ generators and it contradicts to the initial assumption. On the other hand, the composition $\proj_N O \proj_N$ commutes with the action of the subalgebra $\interfin{N}\subset\interinfin$ or the $\JTL{N}$ algebra (recall Thm.~\ref{Thm:US-JTL-iso}) because $P_N$ intertwines the action of $\interfin{N}$ (see Cor.~\ref{cor:coh}) and $O$ commutes with any subalgebra of $\interinfin$ by the assumption. Therefore, the operator $\proj_N O \proj_N$ equals\footnote{For brevity, we write $\f^{N/2}$ instead $\rho_N(\f^{N/2})$, etc.}
 \begin{equation}\label{eq:aa}
 \proj_N O \proj_N = a_1\f^{N/2} + a_2\e^{N/2}, \qquad \text{for}\quad a_1,a_2\in\LQGodd,
  \end{equation}
 because the centralizer of $\JTL{N}$ on $\Hilb_N$ is the algebra generated by $\rho_N\bigl(\LQGodd\bigr)$ and $\f^{N/2}$ and $\e^{N/2}$ (see Sec.~\ref{subsec:cent}). It is then easy to show a contradiction. Indeed, taking similarly the projection onto $\Hilb_{N+2}$, which contains the $\Hilb_N$, we conclude that $P_{N+2}O P_{N+2}$ commutes with the $\interfin{N+2}$ or the $\JTL{N+2}$ action and by~\eqref{eq:aa} it equals $b+\rho_{N+2}(a_1\f^{N/2}+a_2\e^{N/2})$, where $b\in\cent_{\JTL{N+2}}$ such that $b\in\mathrm{ker}P_N$. By this
 equality we obtain that $\rho_{N+2}(a_1\f^{N/2}+a_2\e^{N/2})$ is in $\cent_{\JTL{N+2}}$ and it is possible only if (i) the operator $\rho_{N+2}(a_1\f^{N/2} + a_2\e^{N/2})$ is in $\rho_{N+2}\bigl(\LQGodd\bigr)$ and thus $\proj_N O \proj_N\in\rho_{N}\bigl(\LQGodd\bigr)$ or (ii) $a_1=a_2=0$ and thus $\proj_N O \proj_N=0$. The first assertion contradicts to the assumption on $O$ and the last assertion contradicts to the assumption that $\proj_N O \proj_N$ is a non-zero operator. Thus we obtain the statement of the theorem.
\end{proof}

\section*{Appendix D: Emergence of the interchiral algebra from the lattice}\label{sec:emerge-latVir}
\renewcommand\thesection{D}
\renewcommand{\theequation}{D\arabic{equation}}
\setcounter{equation}{0}

 We look  at  the Fourier transforms of $e_j$'s introducing $q=n\pi/L$, for $n\in\oZ$, and defining the operators
\begin{equation}\label{inter}
H(n) = -\sum_{j=1}^N e^{-iqj} e_j=\sum_{p}
\left[1+e^{iq}+ie^{-ip}-ie^{i(p+q)}\right] \fermd_{p}\,
\ferm_{p+q+\pi}
\end{equation}
which were obtained  in our \first paper~\cite{GRS1}. These  can be rewritten using the more convenient $\chi$ and $\eta$ fermions from~\eqref{eq:chi-eta-def} as
\begin{multline}\label{lat-Vir-H}
H(n) =
2e^{iq/2}\Biggl(\sqrt{\sin{q}}\, \chi^{\dagger}_{0}\bigl(\chi_{q} +
\eta_{q}\bigr) + \sum_{\substack{p=\step\\\text{step}=\step}}^{\pi-q-\step}
\sqrt{\sin{(p)}\sin{(p+q)}} \bigl(\chi^{\dagger}_{p}\,\chi_{p+q} -
\eta^{\dagger}_{p}\,\eta_{p+q}\bigr)\\
+ \sum_{\substack{p=\step\\\text{step}=\step}}^{q-\step}
\sqrt{\sin{(p)}\sin{(q-p)}}\bigl(\chi^{\dagger}_{\pi-p}\,\eta_{q-p} +
\eta^{\dagger}_{\pi-p}\,\chi_{q-p}\bigr)
+ \sqrt{\sin{q}}\, \bigl(\chi^{\dagger}_{\pi-q} +
\eta^{\dagger}_{\pi-q}\bigr)\eta_{0}\Biggr),
\end{multline}
where  $0<n<L$, while a similar expression can be written for other values of the modes $n$.

In order to take the scaling limit for $n$ close to $L$ we define $k=L-n$ and consider operators $\tilde{H}(k)=H(L-k)$ and study their scaling limit for finite $k$. We thus first substitute $n\to L-k$ and $q\to \pi-q'$, with $q'=k\pi/L$, in~\eqref{lat-Vir-H}. Then,
using the limit~\eqref{eq:sferm-lat-def}-\eqref{eq:sferm-zero-def} to the fermions
$\sferm^{1,2}$ and $\bsferm^{1,2}$ and linearizing the dispersion relation, we obtain in
the scaling limit (keeping terms in the sums with  the momentum $p$ close to $0$ or to $\pi$), with finite positive
$k$,
\begin{multline*}
-i\ffrac{L}{2\pi}\tilde{H}(k)\mapsto\sfermp_{0}\bigl(\bsfermm_{k} - \sfermm_{-k}\bigr)
+ \sum_{m=1}^{k-1}\bigl(\sfermp_{-m}\bsfermm_{k-m} -
\bsfermp_{m}\sfermm_{m-k}\bigr)
+ \sum_{m\geq k+1}\bigl(\sfermp_{-m}\bsfermm_{k-m} -
\bsfermp_{m}\sfermm_{m-k}\bigr)\\
+ \sum_{m\geq1}\bigl(\sfermp_{m}\bsfermm_{m+k} -
\bsfermp_{-m}\sfermm_{-m-k})
+ \bigl(\sfermp_{-k}-\bsfermp_{k}\bigr)\sfermm_{0}.
\end{multline*}
Gathering all terms and using the relations~\eqref{sferm-rel}, we finally obtain the contribution corresponding to low-lying
excitations over the ground state,
\begin{equation*}
\ffrac{L}{2\pi i}\tilde{H}(k) \mapsto  S_{\alpha \beta}\sum_{m\in\oZ}\sferm^{\alpha}_{m}\bsferm^{\beta}_{m+k}  \equiv
\enrg_k,\qquad k>0,
\end{equation*}
which we denote in what follows by $\enrg_k$ or   $\enrg^{(0)}_k$. A similar computation gives the same scaling limit of $\ffrac{L}{2\pi i}\tilde{H}(k)$ for negative finite $k$.

\medskip

We then consider the Fourier transformation of the momentum operator $P$,
\begin{equation*}
P(n)=\ffrac{i}{2}\sum_{j=1}^N e^{-iqj} [e_j,e_{j+1}], \qquad q=\ffrac{n\pi}{L},
\end{equation*}
which gives~\cite{GRS1} on a finite spin-chain the expression in $\chi$-$\eta$ fermions:
\begin{multline*}
P(n) = 2e^{iq}\Biggl(\cos{\ffrac{q}{2}}\sqrt{\sin{q}}\, \chi^{\dagger}_{0}\bigl(\chi_{q} -
\eta_{q}\bigr) + \sum_{\substack{p=\step\\\text{step}=\step}}^{\pi-q-\step}
\cos{\bigl(p+\ffrac{q}{2}\bigr)}\sqrt{\sin{(p)}\sin{(p+q)}} \bigl(\chi^{\dagger}_{p}\,\chi_{p+q} +
\eta^{\dagger}_{p}\,\eta_{p+q}\bigr)\\
- \sum_{\substack{p=\step\\\text{step}=\step}}^{q-\step}
\cos{\bigl(p-\ffrac{q}{2}\bigr)}\sqrt{\sin{(p)}\sin{(q-p)}}\bigl(\chi^{\dagger}_{\pi-p}\,\eta_{q-p} -
\eta^{\dagger}_{\pi-p}\,\chi_{q-p}\bigr)
- \cos{\ffrac{q}{2}}\sqrt{\sin{q}}\, \bigl(\chi^{\dagger}_{\pi-q} -
\eta^{\dagger}_{\pi-q}\bigr)\eta_{0}\Biggr)
\end{multline*}
Once again,  to take the scaling limit for $n$ close to $L$ we  consider operators $\tilde{P}(k)=P(L-k)$ and study their scaling limit for finite $k$.  It turns out that  the scaling limit (for finite $k$) is
\begin{equation}
-\ffrac{L^2}{\pi^2}\tilde{P}(k) \mapsto S_{\alpha
 \beta}\sum_{m\in\oZ}(2m+k)\sferm^{\alpha}_{m}\bsferm^{\beta}_{m+k}\equiv
 \enrg^{(1)}_k,\qquad k\in\oZ,
\end{equation}
where the operators $\enrg^{(l)}_k$ are introduced in~\eqref{enrgl-def}, for $l\geq 0$ and $k\in\oZ$.

\subsection{Higher Hamiltonians and the Lie algebra $\interfin{N}$}\label{app:Hln-Sn}
Recall that  in Sec.~\ref{Howe} we gave  a Lie algebraic $\interfin{N}$ description of the JTL algebra in terms of fermion bilinears.
We consider here  a different basis in $\interfin{N}$ that is spanned by the generalized higher Hamiltonians $H_l(n)$ introduced~\eqref{hham-def}.
These Hamiltonians for positive modes $n$ were expressed in~\eqref{Hln-pos} in the basis of $\interfin{N}$. In order to check that the linear span of all $H_l(n)$'s gives the vector space $\interfin{N}$ we need also negative-mode expressions. For $l\geq0$ and $1\leq n\leq L$, we
find a slightly different formula
\begin{multline}\label{Hln-neg}
H_l(-n) =
2e^{-iq\frac{l+1}{2}}\Biggl( \sum_{m=n+1}^{L-1}
\cos{l\bigl(p-\ffrac{q}{2}\bigr)}\sqrt{\sin{(p)}\sin{(p-q)}}\, \Asp_{m,m-n}\\
-\half\sum_{m=1}^{n-1}
\cos{l\bigl(p-\ffrac{q}{2}\bigr)}\sqrt{\sin{(p)}\sin{(q-p)}}\bigl(\Csp_{L-m,L+m-n}
+(-1)^l \Bsp_{m,n-m}\bigr)\\
+\cos{\ffrac{lq}{2}}\sqrt{\sin{q}}\bigl(\Egl_{0,L-n} - (-1)^l \Egl_{0,L+n} + \Egl_{n,L} -
 (-1)^l\Egl_{2L-n,L} \bigr) \Biggr).
\end{multline}

We checked then up to $N=20$ that a basis in the linear span of all $H_l(n)$'s, with $l\geq0$ and $n\in\oZ$,  resides  actually in the range $-\pi + l +1\leq q = n\pi/L \leq \pi + l + 2$ and $l$ -- the index for the family -- runs from $0$ to $L-1$. Moreover, the number of linearly independent $H_l(n)$'s in this range is given by the dimension of the Lie algebra $\interfin{N}$ introduced in Dfn.~\ref{S_N-def}. We thus conclude that the $H_l(n)$'s should be closed under commutations (which is really hard to compute explicitly) and that they  just give  a different basis in the Lie algebra $\interfin{N}$.

\section*{Appendix E: The character formulas}
\renewcommand\thesection{E}
\renewcommand{\theequation}{E\arabic{equation}}
\setcounter{equation}{0}

We provide here an  analysis of the left-right Virasoro $\VirN(2)=\Vir(2)\boxtimes\overline{\Vir}(2)$ content of the
simple $\JTL{N}$ modules in the scaling limit using a general strategy (via the XXZ spin chain) that will be useful for other models as well.
We recall the basic fact~\cite{Cardy} that
\begin{equation}
 \mathrm{Tr}\, e^{-\beta_R(H-Ne_0)}e^{-i\beta_I P}\to \mathrm{Tr}\, q^{L_0-c/24}\bar{q}^{\bar{L}_0-c/24}\label{charform}
 \end{equation}
where $H$ and $P$ are a lattice Hamiltonian (normalized such that
 the velocity of sound is unity) and momentum (see \cite{GRS1}
 and below), $e_0$ is the ground state energy per site in the
 scaling limit, $q(\bar{q})=\exp\left[-{2\pi\over N}(\beta_R\pm
 i\beta_I)\right]$ with $\beta_{R,I}$ real and $\beta_R>0$, and $N$ is
 the length of the chain. On the right-hand side, we have $L_0$ and $\bar{L}_0$ as zero modes of the stress-energy tensor, and $c$ is the corresponding central charge. The trace on the left is taken over a
 subspace (of scaling states) of the spin chain, and the trace on the right is over the
 states occurring in this subspace in the scaling limit.

The calculation of an expression such as the trace on the left
of~\eqref{charform} is most easily done when the spin chain is the
well known XXZ chain with appropriate value of the deformation
parameter $\q$ and appropriately twisted boundary conditions. To be
more specific, we consider the Hamiltonian
\begin{equation}
\HXXZ=\sum_{j=1}^{2L-1}
\bigl(\sigma^+_j\sigma^-_{j+1}+\sigma_j^-\sigma_{j+1}^+\bigr) + e^{2iK}\sigma^{+}_{2L}\sigma_{1}^{-}+ e^{-2iK}\sigma^{-}_{2L}\sigma_{1}^{+}
+\sum_{j=1}^{2L}\ffrac{\q+\q^{-1}}{4}\sigma_j^z\sigma_{j+1}^z,
\end{equation}
where the $2\times2$-matrices $\sigma^{\pm}$, $\sigma^z$ are Pauli
matrices.  The momentum operator $\PXXZ$ can be chosen either by
using the exact eigenvalues of the translation operator on the spin
chain, or by using the general formulas in~\cite{GRS1}. Parametrizing
$\q=e^{i\pi/(x+1)}$, it has been known for a long time thanks to
Coulomb gas and Bethe ansatz arguments~\cite{PasquierSaleur,Rittenbergetal} that %
 \begin{equation}
   \hbox{Tr}_{{\cal H}_{[j]}} e^{-\beta_R(\HXXZ-2Le_0)}e^{-i\beta_I \PXXZ}\to F_{j,e^{2iK}}
 \end{equation}
where
 \begin{equation}
  F_{j,e^{2iK}}=\ffrac{q^{-c/24}\bar{q}^{-c/24}}{P(q)P(\bar{q})}
     \sum_{n\in \oZ}q^{\left(xj+(x+1)(n+K/\pi)\right)^{2}-1\over 4x(x+1)}
    \bar{q}^{\left(xj-(x+1)(n+K/\pi)\right)^{2}-1\over 4x(x+1)}\label{contlimtr}
     \end{equation}
and the subspace ${\cal H}_{[j]}$ is the subspace of  spin $S^z=j$, with $-L\leq
j\leq L$, of the spin-chain of length $N=2L$ (also denoted by $\APrTL{j}$),
and $P(q)$ is defined in~\eqref{Pq}.

On the other hand, it is also well known \cite{PasquierSaleur} that this XXZ Hamiltonian can be written  as
\begin{equation}
\HXXZ=-\sum_{i=1}^{2L}e_i,
\end{equation}
where the $e_i$ are expressed in the spin-$1/2$
basis~\cite{PasquierSaleur} and, together with $u^2$, provide a
representation of the Jones--Temperley--Lieb algebra $\JTL{N}$
whenever $e^{2ijK}=1$, for $j\neq 0$, and $e^{2iK}=\q^2$, for $j=0$.
This representation has for dimension the usual binomial coefficient
$\binom{2L}{L+j}$, and is believed to be isomorphic to the standard
module $\AStTL{|j|}{e^{\pm i\phi}}$ for $\q$ generic and $\phi=-2K$,
where the `$+$' is for a positive $j$, and the `$-$' is for a negative
value of $j$. The standard modules (defined and discussed
in~\cite{GRS2}) on opposite sectors $S^z=\pm j$ are conjugate to each
other with respect to the  bilinear pairing
$\AStTL{|j|}{z}\times\AStTL{|j|}{z^{-1}}\to\oC$ from~\cite{GL} invariant with respect to the action of the affine TL algebra denoted there by $\ATL{N}$.  The
pairing has a trivial radical only in the generic cases.

When $\q$ is a root of unity, the representation obtained from the
twisted XXZ chain is not isomorphic to a standard module any
longer. But it is well established~\cite{PasquierSaleur} that the
traces over the modules and their scaling limit behave smoothly across
those points. This means that we can use~\eqref{contlimtr} to
obtain the generating function of conformal weights (that is,
eigenvalues of $L_0$ and $\bar{L}_0$) in the scaling limit of the modules
occurring in our $\gl(1|1)$ spin chain.

 Let us
 specialize now to $\q=i$ ($x=1$) where $c=-2$ and
 \begin{equation}
   F_{j,e^{2iK}}=\ffrac{q^{1/12}\bar{q}^{1/12}}{P(q)P(\bar{q})}
     \sum_{n\in \oZ}q^{\left(j+2(n+K/\pi)\right)^{2}-1\over 8}
    \bar{q}^{\left(j-2(n+K/\pi)\right)^{2}-1\over 8}
     \end{equation}
(note that $F_{0,e^{2iK}}=F_{0,e^{-2iK}}$). Recall the Kac formula for the value of the central charge $c=-2$
 \begin{equation}
 h_{r,s}={(2r-s)^2-1\over 8},
 \end{equation}

Introduce now
the character of the Kac representation
\begin{equation}
K_{r,s}={q^{h_{rs}}-q^{h_{r,-s}}\over q^{c/24}P(q)}={q^{(2r-s)^2/8}-q^{(2r+s)^2/8}\over\eta(q)}
\end{equation}
where  the Dedekind's eta function is $\eta(q)=q^{1/24}P(q)$. A short calculation establishes the crucial formula
 \begin{equation}
 F_{j,\q^{2j+2k}}-F_{j+k,\q^{2j}}=\sum_{r=1}^\infty K_{r,k}\bar{K}_{r,k+2j}\label{useful},
 \end{equation}
 with $\q=i$ here.

From the structure of the spin chain modules (see
Sec.~\ref{ind-chain-bimod-subsec} and Fig.~\ref{FF-JTL-mod}), we
deduce the traces
\begin{equation}
F_{j,(-1)^{j+1}}^{(0)}\equiv \lim_{L\to\infty} \mathrm{Tr}_{\AIrrTL{j}{(-1)^{j+1}}} e^{-\beta_R(H(0)-Le_0)}e^{-i\beta_I P(0)}
\end{equation}
 in each of
the simple subquotients $\AIrrTL{j}{(-1)^{j+1}}$
 appearing in the
(scaling limit of) $\JTL{}$-modules $\APrTL{k}$:
\begin{equation}
F_{j,(-1)^{j+1}}^{(0)}=\sum_{j'\geq j}
(-1)^{j'-j}\sum_{r=1}^\infty K_{r,1}\bar{K}_{r,2j'+1},\qquad j\leq k.
\end{equation}

On the other hand, recall the characters of the simples of Virasoro at $c=-2$
\begin{equation}
\chi_{j,1}={q^{(2j-1)^2/8}-q^{(2j+1)^2/8}\over\eta(q)}
\end{equation}
from which it follows that
\begin{equation}
K_{r,2j+1}=\sum_{s=-r+1}^r \chi_{j+s,1}.
\end{equation}
Moreover, we formally define the character when $r$ is negative by the same formula, so we have
$\chi_{-j,1}=-\chi_{j,1}$. It follows that
\begin{equation}
K_{r,1}=\chi_{r,1}.
\end{equation}
  Straightforward algebra then leads to the key result
  \begin{equation}\label{app:Fs-Vir-content-gen}
  F_{j,(-1)^{j+1}}^{(0)}=\sum_{j_1,j_2> 0}^* \chi_{j_1,1}\overline{\chi}_{j_2,1}
  \end{equation}
  where the sum is done with the following constraints:
  \begin{eqnarray}\label{app:cond-sum-star}
  |j_1-j_2|+1\leq j,\qquad 
  j_1+j_2-1\geq j, \qquad 
  j_1+j_2 - j = 1\; \text{mod} \; 2
  \end{eqnarray}
 (note this is equivalent to treating $j$ not as a spin but as a degeneracy, {\it i.e.}, setting $j=2s+1$, $j_i=2s_i+1$ and combining $s_1$ and $s_2$ to obtain spin $s$).

\end{document}